\documentclass{article}
\usepackage{jheppub}
\usepackage[utf8]{inputenc}
\usepackage[compat=1.1.0]{tikz-feynman} 
\tikzfeynmanset{warn luatex=false}
\usetikzlibrary { decorations.pathmorphing, decorations.pathreplacing, decorations.shapes }

\tikzset{
        ribbon/.style={
            preaction={
                    draw,
                    line width=0.25cm,
                    black
                },
                draw,
                line width=0.2cm,
                white
             } }

\usepackage{amsmath,amsthm}
\usepackage{mathrsfs}
\usepackage{amscd,amssymb, amsfonts, verbatim,subfigure, enumerate}
\usepackage[mathcal]{eucal}
\usepackage[super]{nth}

\newtheorem{theorem}{Theorem}[subsection]
\newtheorem{definition}{Definition}[subsection]
\newtheorem{conjecture}{Conjecture}[subsection]

\newtheorem{thm-def}{Theorem/Definition}[theorem]
\newtheorem{proposition}[theorem]{Proposition}
\newtheorem{lemma}[theorem]{Lemma}

\newcommand{\be}{\begin{equation}}
\newcommand{\ee}{\end{equation}}
\newcommand{\til}{\widetilde}

\newcommand{\half}{\tfrac{1}{2}}

\title{Holography for integrable field theories} 
\author[a]{Kevin Costello}
\author[b]{Joaquin Liniado}
\affiliation[a]{Perimeter Institute for Theoretical Physics,\\
  31 Caroline St N, Waterloo, ON N2L 2Y5, Canada}
\affiliation[b]{School of Mathematics, University of Edinburgh,\\
  James Clerk Maxwell Building, Peter Guthrie Tait Road, Edinburgh EH9 3FD, UK}
\emailAdd{kcostello@perimeterinstitute.ca}
\emailAdd{jliniado@ed.ac.uk}

\renewcommand{\i}{i}
\newcommand{\wbar}{\br{w}}

\newcommand{\PV}{\op{PV}}

\newcommand{\dpa}[1]{\frac{\partial}{\partial #1}}

\newcommand{\eps}{\epsilon}

\newcommand{\what}{\widehat}

\newcommand{\mscr}{\mathscr}

\newcommand{\br}{\overline}

\newcommand{\iso}{\cong}
\newcommand{\C}{\mathbb C}
\newcommand{\CP}{\mathbb{CP}}

\newcommand{\Q}{\mbb Q}

\newcommand{\Oo}{\mscr O}

\newcommand{\op}{\operatorname}

\newcommand{\mbb}{\mathbb}
\newcommand{\mf}{\mathfrak}
\newcommand{\mc}{\mathcal}

\newcommand{\ip}[1]{\left\langle #1 \right\rangle}
\newcommand{\abs}[1]{\left| #1 \right|}

\newcommand{\R}{\mbb R}
\renewcommand{\d}{\mathrm{d}}

\newcommand{\dbar}{\br{\partial}}

\abstract{
We develop a holographic dual to the integrable field theories built from $4d$ Chern-Simons theory. The dual theory is a topological string on a generalized Calabi-Yau manifold, which is entirely encoded in a geometric object we call a \emph{planar spectral curve}. We describe the RG flow for a planar integrable field theory as a geometric flow on the moduli of planar spectral curves. This description is valid at all orders in the 't Hooft coupling. We compute this flow explicitly up to two loops in many examples, and find an exact match with known field theory results. At one loop this includes a wide class of models. At two loops we reproduce the RG flow for the principal chiral model with WZ term, and for a class of integrable $\sigma$-models with target $U(N) \times U(N)$.

 }

\begin{document}
\maketitle

\section{Introduction}

Four-dimensional Chern-Simons theory \cite{Costello:2013zra,Costello:2017dso,Costello:2018gyb,Costello:2019tri, Delduc:2019whp, Vicedo:2019dej} is a general framework for building continuum and discrete two-dimensional integrable models.  This paper is concerned with $4d$ CS frameworks that engineer continuum integrable models such as the principal chiral model (PCM).  In this context, $4d$ CS produces a \emph{classically} integrable theory associated to semi-simple gauge group $G$, and a \emph{classical spectral curve}: 
\begin{definition}
A classical spectral curve is
\begin{itemize} 
\item A Riemann surface $\Sigma$.
\item A divisor $D$.
\item A one-form $\alpha \in H^0(\Sigma, K_\Sigma(2 D))$ on $\Sigma$ with second order poles on $D$.
\item A decomposition 
\footnote{Two classical spectral curves are equivalent if $(\Sigma,\alpha)$ are the same and $\Sigma_{\pm} \subset \Sigma'_{\pm}$. This means that the only important data in the decomposition $\Sigma = \Sigma_+ \cup \Sigma_-$ is whether a zero of $\alpha$ is in the chiral patch or the anti-chiral patch.}
 $\Sigma = \Sigma_+ \cup \Sigma_-$, where $\alpha$ has no zeroes on $\Sigma_+ \cap \Sigma_-$. 
\end{itemize}
\end{definition}

The integrable field theory is engineered from this data as follows. We consider a $G$-gauge theory on $\Sigma \times \R^2$ with action\footnote{Our normalization is such that if $\alpha = k z^{-1} \d z$, then $3$-dimensional theory obtained by dimensionally reducing along the circle $\abs{z} = 1$ is Chern-Simons at level $-k\hbar$. This normalization differs by a factor of $-2/\hbar$ from that used in \cite{Lacroix:2025ias}.   } 
\begin{equation} 
-\frac{1}{8 \pi^2 \hbar \i}  \int \alpha CS(A) 
\end{equation}
where $ CS(A),$ is the Chern-Simons $3$-form of a $G$-gauge field ($A$ is taken up to the addition of a multiple of $\alpha$, so it only has $3$ components).  Also $\hbar$ is a loop counting parameter and $\lambda = N \hbar$ is the 't Hooft coupling. 

 The gauge field has certain boundary conditions at poles and zeroes of $\alpha$. At a second order pole we require that $A = 0$. At zeroes of $\alpha$ we allow $A_w$ or $A_{\wbar}$ to have certain poles, depending on whether the zero is in $\Sigma_+$ or $\Sigma_-$ (here $w = x+\i y$ where $x,y$ are coordinates on $\R^2$).
 
In \cite{Costello:2019tri} it was shown that the effective theory obtained by compactifying $4d$ Chern-Simons to $\R^2$ on $\Sigma$ is an integrable field theory with spectral parameter living in $\Sigma$.  Many (perhaps all) integrable field theories are built by variants of this construction.  The simplest example is when $\Sigma = \CP^1$ and the one-form is
\begin{equation} 
\frac{\d z}{z^2} (z^2 - a^2) 
\end{equation} 
where we place chiral and anti-chiral zeroes at $\pm a$.  This model leads to the principal chiral model for $G$. This is the $\sigma$-model with target the group manifold $G$, and action
\begin{equation} 
\frac{a}{4\pi}\int_{\R^2} \ip {g^{-1} \partial g, g^{-1} \dbar g } 
\end{equation}
A natural generalization of the principal chiral model is obtained when $\alpha$ is a one-form on $\CP^1$ with $n$ second-order poles and $2n-2$  first-order zeroes, $n-1$ of which are chiral and $n-1$ are anti-chiral.  In this case, the resulting integrable field theory is a $\sigma$-model on $G^{n-1}$, with $3n-4$ independent coupling constants encoded in the moduli of the one-form $\alpha$.  There are very many generalizations, when different boundary conditions are imposed or when $\Sigma$ has genus greater than zero. 
  

In this paper we give a geometric definition of \emph{planar} spectral curve, which deforms that of classical spectral curve.   We give a holographic argument to show that every planar spectral curve gives rise to a planar integrable field theory, i.e.\ an integrable field theory where the fields have $U(N)$ indices and we take the 't Hooft limit.  The holographic analysis is based on embedding $4d$ Chern-Simons in the topological string on certain generalized Calabi-Yau manifolds.

The moduli of planar spectral curves has a geometrically defined vector field, which is derived from a holographic description of the RG flow. Holography implies that this corresponds to the RG flow of the corresponding planar integrable field theory, at all orders in the 't Hooft expansion.

For models with genus $0$ spectral curve, we give an explicit algorithm to determine the all-order RG flow.  We implement this algorithm to one and two loop orders and compare with field theory results in the following examples.
\begin{enumerate}
\item
At one loop we show that this flow matches known results for essentially all models with genus zero spectral curve.  This reproduces field theory calculations of Lacroix, Levine and Wallbeg \cite{Lacroix:2025ias}.
\item
At two loops, our flow matches the   RG flow of the principal chiral model with WZW term, computed by Hoare, Levine and Tseytlin \cite{Hoare:2019mcc}.  
\item Also at two loops, Levine and Tseytlin \cite{Levine:2021fof} computed the RG flow of the integrable $G\times G$ models, two PCMs with WZ terms coupled through their currents. These models have six parameters in total, four of which flow. The two-loop flow is given by very complicated rational functions of the parameters, with  over 100 terms.  Our holographic flow matches the field theory computation exactly.   
\end{enumerate}
In field theory, the two-loop RG flow depends on a scheme (some terms, however, are scheme-independent).  In our holographic approach we also need to choose a scheme,  and we find that the most natural holographic scheme matches exactly the scheme used in \cite{Hoare:2019mcc,Levine:2021fof}.
\section{Planar spectral curves}
In this section we will give the definition of planar spectral curve.  We will give the definition in several steps.

Let us first recall the following simple definition.
\begin{definition}
A $\lambda$-connection on a line bundle $L$ on a Riemann surface is an operator
$$
\nabla : \Omega^0(L) \to \Omega^{1,0}(L)
$$
so that $[\nabla, \dbar_L] = 0$ and  
\begin{equation*} 
\begin{split} 
 \nabla(f s) - f \nabla s = (\lambda \partial f) s\\
  \end{split} 
\end{equation*}
\end{definition}
When $\lambda = 0$ a $\lambda$-connection is a holomorphic one form and when $\lambda = 1$, a $\lambda$-connection is just a holomorphic connection.  When $\lambda$ is non-zero, if $\nabla$ is a $\lambda$-connection then $\lambda^{-1} \nabla$ is a connection.  

\begin{definition}
Given a $\lambda$-connection $\nabla$ on the canonical bundle $K_{\Sigma}$ of a surface $\Sigma$, a holomorphic one-form $\what{\alpha}$ satisfies the Bernoulli equation if
\begin{equation*} 
\nabla \what{\alpha} = \what{\alpha}^2 
\end{equation*}
\end{definition}
One way to understand this definition is the following:
\begin{lemma} 
 $\what{\alpha}$ satisfies the Bernoulli equation if and only if, when we trivialize the canonical bundle of $\Sigma$ using $\what{\alpha}$, then the $\lambda$-connection $\nabla$ in this frame is $\lambda \partial + \what{\alpha}$.  
\end{lemma}
\begin{proof} 
For any one-form $\omega$, the $\lambda$-connection in the frame given by $\omega$ is $\lambda \partial + \omega^{-1} \nabla \omega$. Clearly $\what{\alpha}^{-1} \nabla \what{\alpha} = \what{\alpha}$.  
\end{proof} 

More generally, if we have a divisor $D$, we can consider  $\lambda$-connections on $K_{\Sigma}(k D)$ where the connection has a pole of order $k$ on $D$. A section of $K_{\Sigma}(k D)$ satisfies the Bernoulli equation if
\begin{equation} 
 \nabla \what{\alpha} = \what{\alpha}^2 
\end{equation} 
where now both sides are sections of $K_{\Sigma}(2 k D)$.

If, in a patch, we trivialize $\nabla = \lambda \partial_z + \phi(z) \d z$, and write $\what{\alpha} = g(z) \d z$, then $g(z)$ satisfies the equation
\begin{equation} 
\lambda g'(z) + \phi(z)g(z) = g(z)^2  
\end{equation} 
which is a first-order Bernoulli-Ricatti equation. 

Now suppose we have a $\lambda$-connection $\nabla_\lambda$, where $\lambda$ is a formal parameter. Then, at $\lambda = 0$, we have a one form 
\begin{equation} 
\alpha = \nabla_{\lambda = 0}. 
\end{equation} 
\begin{lemma} \label{lemma:uniquebernoulli} 
 Assuming $\alpha$ is non-zero, there is locally a unique meromorphic solution $\what{\alpha}(\lambda)$ to the Bernoulli equation in series in $\lambda$. 

This has the property that
$$
\what{\alpha}_{\lambda =0} = \alpha = \nabla_{\lambda = 0}. 
$$
If we work in a coordinate patch so that
 $$
\nabla = \lambda \partial_z + \partial_z \gamma
$$
then
$$
\what{\alpha} = -\lambda \d \log \int^z e^{-\gamma/\lambda}  
$$
where the contour is chosen so that there is a small $\lambda$ asymptotic expansion.

These statements continue to hold if we have a divisor $D$, and a $\lambda$-connection with a pole of order $k$ on $D$ on the bundle $K_{\Sigma}(k D)$. 
\end{lemma}
\begin{proof}
Working in a patch, write $\alpha = \phi(z) \d z$ and $\nabla = \lambda \partial + \phi(z) \d z$. Then $g(z)$ satisfies Bernoulli-Ricatti equation
\begin{equation} 
\lambda g'(z) + \phi(z)g(z) = g(z)^2  
\end{equation}
It is easy to check that this has a unique solution where $g(z)$ is a series in $\lambda$.  (Globally, as we will see later, there can be obstructions to building a solution).  

A simple explicit calculation shows that the integral expression for $\what{\alpha}$ satisfies the Bernoulli equation.  
\end{proof}

Now let us define the concept of planar spectral curve.  
\begin{definition}
A \emph{planar spectral curve} is a Riemann surface $\Sigma = \Sigma_+ \cup \Sigma_-$,  together with:
\begin{enumerate}
\item A divisor $D$, which also depends on  
\item On $\Sigma_+$, a $\lambda$-connection $\nabla_+$ on the bundle $K_{\Sigma_+}(2 D)$, with second order poles on $D$. 
\item On $\Sigma_-$, a $-\lambda$-connection $\nabla_-$, also on the bundle $K_{\Sigma_-}(2 D)$, with second order poles on $D$.  
\item On $\Sigma_+ \cap \Sigma_-$ a one-form $\what{\alpha} \in K_{\Sigma}(2 D)$ which satisfies the Bernoulli equation both for $\nabla_+$ and $\nabla_-$: 
\begin{equation} 
\begin{split} 
\nabla_+ \what{\alpha} &= \what{\alpha}^2 \\
 \nabla_- \what{\alpha} &= \what{\alpha}^2  
\end{split} 
\end{equation} 
\end{enumerate}
Equivalently, on $\Sigma_+ \cap \Sigma_-$,
\begin{equation} 
\nabla_- = - \nabla_+ + 2 \what{\alpha} 
\end{equation} 
\end{definition}
The concept of planar spectral curve depends on $\lambda$. If $\lambda$ is a finite number, we can talk about a $\lambda$-planar spectral curve: this is a Riemann surface $\Sigma$ with the data $(\nabla_+, \nabla_-, \what{\alpha})$ as listed above.  However, we are more often interested in a family of $\lambda$-planar spectral curves where $\lambda$ is a formal parameter.  In that context, all of the data depends on $\lambda$, including the moduli of the surface $\Sigma$ and the divisor $D$.    (In mathematical language, we would say that all objects live over the base ring $\C[[\lambda]]$ of series in $\lambda$). For most of the paper we will work in series in $\lambda$, in which case the definition of planar spectral curve must be interpreted in this way. 

When $\lambda = 0$, a planar spectral curve reduces to a classical spectral curve. The  continuous parameters of a classical spectral curve are the moduli of $\Sigma$, the divisor $D$, and the one-form $\alpha \in H^0(\Sigma, K_\Sigma(2 D))$. There are $4g-4+ 3 \op{deg} D$ parameters, unless $D$ is empty in which case there are $4g-3$ parameters.

Later we will show later that a classical spectral curve can be deformed into a planar spectral curve (working in series in $\lambda)$ in the following contexts:
\begin{itemize} 
 \item If $\alpha$ has at least one pole. 
\item If $\alpha$ has no poles and the number of zeroes of $\alpha$ on $\Sigma_+$ and on $\Sigma_-$ agree when counted with multiplicity. 
\end{itemize}

If $\alpha$ has no poles, then (as we will see later) $\what{\alpha}$ has poles on $\Sigma_+$ whose total residue is $2 \pi \i \lambda$ times the number of zeroes of $\alpha$. This explains why the second condition is necessary.

The moduli space of planar spectral curves lives in a family of moduli spaces which at $\lambda = 0$ is the moduli of classical spectral curves. General formal arguments from deformation theory tell us that, if there are no obstructions to building a planar spectral curve, then at each order in $\lambda$ the ambiguity in extendng a planar spectral curve to the next order is exactly the tangent space to the moduli space of classical spectral curves.  

This ambiguity can be thought of as the same as the ambiguity in building a quantum theory from a classical theory: at each order in the coupling  we are free to deform the theory by adding on finite counter-terms, which are terms appearing in the classical action.

\begin{conjecture}
To every planar spectral curve is associated a planar integrable field theory with 't Hooft coupling $\lambda$. 

This planar integrable field theory quantizes the classical integrable theory associated to its classical limit. 
\end{conjecture}
If we accept this conjecture, then the moduli space of planar spectral curves is a scheme-independent way of describing the couplings of a planar integrable theory.  Since the dimension of this moduli space is the same as the dimension of the moduli of classical spectral curves, there are the same number of coupling constants.

With the definition of planar spectral curve we have presented, the corresponding integrable field theory need not be unitary. One should be able to strengthen the axioms to enforce unitarity.   

There is a related approach to integrable field theories based on affine Gaudin models \cite{Vicedo:2017cge,  Delduc:2018hty, Vicedo:2019dej}. At least at genus zero, the integrable field theories built from $4d$ CS and from affine Gaudin models are the same.   It was shown in \cite{Lacroix:2018fhf,   Kotousov:2022azm} that when studying \emph{chiral} integrable models from the affine Gaudin perspective, the one-form defining the theory classically needs to be upgraded to a $\lambda$-connection \footnote{Their analysis was non-planar, and $\lambda$ was $\hbar$ times the dual Coxeter number which in the case of $SU(N)$ is $N$.}.  This is consistent with our holographic approach.  

\subsection{Outline of the holographic duality}
This conjecture is based on a holographic analysis of $4d$ Chern-Simons theory, which we will explain briefly now and in more detail later.
\begin{theorem}
For every classical spectral curve $(\Sigma = \Sigma_+ \cup \Sigma_-, \alpha)$, there is a rank $2$ real vector bundle $V \to \Sigma$, and a generalized Calabi-Yau structure on 
\begin{equation} 
X_{\Sigma} = V \times \R^2. 
\end{equation}

$4d$ Chern-Simons is the theory on a brane wrapping $\Sigma \times \R^2 \subset X_{\Sigma}$, where $\Sigma$ is the zero section of the rank two vector bundle $V$. 

Every planar spectral curve $\what{\Sigma}$  quantizing the classical spectral curve gives rise to a new generalized Calabi-Yau $\what{X}_{\what{\Sigma}}$ from backreacting the brane supporting $4d$ Chern-Simons. 
\end{theorem}
This holographic set-up is related to the one studied in \cite{Gaiotto:2026qai} for minimal model holography, and to \cite{Budzik:2023xbr} in the study of holography for twisted $N=1$ gauge theory. 

If we work in series in $\lambda$, the back-reacted geometry is a new generalized Calabi-Yau structure on the manifold $X_{\Sigma} \setminus (\Sigma \times \R^2)$, where the differential forms defining the generalized CY have perscribed poles on the locus we have removed.

In principle, our method works non-perturbatively, but we do not have a very strong understanding of the non-perturbative geometry. However, given a non-perturbative planar spectral curve with some natural additional hypothesis on $\what{\alpha}$, we can build a backreacted geometry where the 't Hooft coupling $\lambda$ takes finite (real) values.  This geometry is only defined, roughly, on the region far away from the brane we have removed. (In usual holography, the corresponding region is near the boundary of $AdS$ space).  This phenomenon is to be expected on physical grounds, because the theories we are considering have a mass-gap. 

The non-perturbative geometry should, in principle, encode the spectrum of massive particles of the dual field theory.  We leave this question to future work; all the explicit computations in this paper are perturbative. 

Before we move on, let us clarify an important point regarding what is meant by the holographic dual of a model like  the PCM. Holography is concerned with gauge theories, and gauge invariant objects have all $U(N)$ indices contracted. Nothing in the dual string theory knows about quantities with free $U(N)$ indices.  

When we say we are studying the holographic dual of the PCM, what we mean more precisely is the following.  If we have an integrable model with $U(N)$ symmetry given by a current $J$, we can couple it to two-dimensional topological BF theory with gauge group $U(N)$. This has action
\begin{equation}
\int B F(A) + A J
\end{equation}
where $B$ is an adjoint-valued scalar.   We will couple an integrable model to a copy of $U(N)$ BF theory for every $U(N)$ symmetry it has.  This has the effect of projecting to $U(N)$ invariant quantities in the integrable model, but it \emph{does not} change the dynamics of the theory. 

If we only consisder operators that do not involve $B$, they are precisely the $U(N)$ invariant operators in the original theory before gauging. This is because on-shell $F(A) = 0$ so we can not build operators from the field strength. 

Further, correlation functions of these operators are exactly the same as they were before we coupled to gauge theory. This is because, in BF theory, the propagator connects $B$ and $A$; so if we consider operators which do not involve $B$, there are no gauge theory propagators.

Finally, the RG flow of the gauged model is exactly the same as that of the model before gauging. This is again because gauge theory propagators do not occur in any Feynman diagram that can contribute to the RG flow. 

We show in section \ref{sec:neuman} that coupling to $2d$ BF theory can be interpreted, in $4d$ Chern-Simons, as using Neumann rather than Dirichlet boundary conditions at poles in the one-form.  We should do this at all second-order poles, so that for the PCM we couple to BF theory for the group $U(N) \times U(N)$, acting on the left and the right. 

\subsection{The RG flow on planar spectral curves}
Next let us explain the RG flow which is derived from this holographic analysis.   We will see that every planar spectral curve has a canonical deformation, so that the moduli of planar spectral curves has a canonical vector field.

 Given a planar spectral curve, we let  
\begin{equation} 
V = \frac{1}{\what{\alpha}}  
\end{equation}  
be the vector field on $\Sigma_+ \cap \Sigma_-$ which is inverse to $\what{\alpha}$.   
\begin{lemma}
\label{lemma:bernoullivector} For a $\lambda$-connection $\nabla$ on the canonical bundle, a one-form
$\what{\alpha}$ satisfies the Bernoulli equation if and only if the Lie derivative of $V = 1/\what{\alpha}$ preserves $\nabla$. 
\end{lemma}
\begin{proof} 
Recall that $\what{\alpha}$ satisfies the Bernoulli equation if and only if, in the trivialization of the canonical bundle given by $\what{\alpha}$, $\nabla$ takes the form
\begin{equation} 
\nabla = \lambda \partial + \what{\alpha}. 
\end{equation}
The statement we are proving is a consequence of this.  The statement we need to prove is that, for all holomorphic one-forms $\omega$, 
\begin{equation} 
\nabla ( \mc{L}_V \omega ) = \mc{L}_V ( \nabla \omega)  
\end{equation}
where both sides are sections of $K_{\Sigma}^{\otimes 2}$ and $\mc{L}_V$ is the Lie derivative.

We can verify this equation by trivializing the canonical bundle using $\what{\alpha}$, which amounts to taking $\omega = f \what{\alpha}$ for some function $f$.  Then, since $\mc{L}_V (\what{\alpha}) = 0$, we have 
$\mc{L}_V \omega = \what{\alpha} \mc{L}_V f$. Since the connection one-form in the frame given by $\what{\alpha}$ is $\what{\alpha}$,  and 
$$\nabla \omega = \nabla (f \what{\alpha}) = \lambda \what{\alpha} \partial f + \what{\alpha}^2 f.$$  
Then
\begin{equation} 
\begin{split} 
 \nabla (\mc{L}_V \omega ) &= \what{\alpha}  \lambda \partial (\mc{L}_V f) + \what{\alpha}^2 (\mc{L}_V f)  \\
\mc{L}_V (\nabla \omega) &= \mc{L}_V\left(  \lambda\what{\alpha} (\partial f) + \what{\alpha}^2  f \right) 
\end{split} 
\end{equation}
These agree, since $\mc{L}_V$ commutes with both $\what{\alpha}$ and $\partial$. 
\end{proof}

We want to define a canonical deformation of any planar spectral curve. That is, to any planar spectral cuve $\Sigma$, we want to define a family of planar spectral curves $\Sigma_{\mu}$ which depend on the scale parameter $\mu$ of the integrable theory (our convention is that large $\mu$ means moving to the UV).  We will define $\Sigma_{\mu}$ in series in $\log \mu$.

\begin{definition}\label{def:rgflow}
Given a planar spectral curve $(\Sigma, \nabla_+ \nabla_-)$ define a family planar spectral curve $\Sigma_\mu$ as follows. 

$\Sigma_\mu$ will have two patches $\Sigma_+$, $\Sigma_-$, which are the same as the two patches on the original $\Sigma$.  The connections $\nabla_{\pm}$ on $\Sigma_{\pm}$ will be unchanged. 

The dependence on the scale parameter $\mu$ will be entirely through the gluing map identifying a patch in $\Sigma_+$ with the corresponding patch in $\Sigma_-$.    Let $U_+ \subset \Sigma_+$ be $\Sigma_+ \cap \Sigma_-$, viewed as an open in $\Sigma_+$, and similarly define $U_- \subset \Sigma_-$.  We identify $\Sigma_+$ with $\Sigma_-$ by gluing $U_+$ to $U_-$  via the map
\begin{equation}
\begin{split} 
 \rho_{\mu} &: U_+ \to U_- \\
\rho_{\mu} &= e^{2\lambda V \log \mu}  
\end{split} 
\end{equation}   
\end{definition}
The key point of this definition is the following.

\begin{lemma} 
 $\Sigma_{\mu}$ defined in this way defines a family of planar spectral curves depending on $\mu$. 
\end{lemma}
\begin{proof}
The axioms of a planar spectral curve are an equation relating $\nabla_+$ and $\nabla_-$ on $\Sigma_+ \cap \Sigma_-$.    We need to check that this property continues to hold when $\Sigma_+$ and $\Sigma_-$ are glued by the non-trivial map $\rho_{\mu}$.

On $\Sigma_{\mu}$, the patches $\Sigma_{+,\mu}$ and $\Sigma_{-,\mu}$ are the same, but we have applied the coordinate transformation $\rho_{\mu}$ to $\Sigma_{+,\mu}$ before gluing.  We therefore need to check that the connections $\rho_{\mu}^\ast \nabla_+$ and $\nabla_-$ satisfies the axioms of a planar spectral curve.
 
 By lemma \ref{lemma:bernoullivector}, we have $\mc{L}_V \nabla_+ = 0$. Since $\rho_{\mu}$ is an exponential of the transformation $V$, this implies that
\begin{equation} 
\rho_{\mu}^\ast \nabla_{+} = \nabla_{+} 
\end{equation} 
so the properties of a planar spectral curve continue to hold.

\end{proof}
It is important to note that this apparently simple expression for the RG flow holds to all orders in the 't Hooft coupling $\lambda$.  In practice, as we will see later, this flow can be quite complicated, because the relationship between $\what{\alpha}$ and $\nabla_+$ is complicated.

Of course, if we want the RG flow as a vector field, we deform $\Sigma$ by modifying the map gluing $\Sigma_+$ to $\Sigma_-$ to first order using the vector field $V$. More abstractly,  we can view $V$ as being an element of the \v{C}ech cohomology group  $\breve{H}^1(T \what{\Sigma})$, using the open cover $\what{\Sigma} = \Sigma_+ \cup \Sigma_-$. We can interpret it as a Beltrami differential deforming the Riemann surface $\what{\Sigma}$.   

At genus $0$, the moduli of the surface $\Sigma$ can not be deformed, so that the deformation of a planar spectral curve is entirely encoded in how $\nabla_+$ and $\nabla_-$ deform. We have the following explicit description.
\begin{lemma} 
Let $\Sigma = \CP^1$, and suppose that $\Sigma_{\pm}$ are discs which we denote by $D_{\pm}$. Let $V = \frac{1}{\what{\alpha}}$ and let us write $V= V_+ + V_-$, where $V_+$ is regular on $D_+$ and $V_-$ is regular on $D_-$.  Then, under the RG flow, 
\begin{equation} 
\begin{split} 
\delta_{RG} \nabla_+ &= -2 \lambda\mc{L}_{V_+} \nabla_+\\ 
 \delta_{RG} \what{\alpha} &= -2 \lambda\mc{L}_{V_+} \what{\alpha} = 2 \lambda \mc{L}_{V_-} \what{\alpha} \\
 \delta_{RG} \nabla_- &= 2 \lambda\mc{L}_{V_-} \nabla_-\\
\end{split} 
\end{equation}
\end{lemma}
\begin{proof} 
Let $z$ be the standard coodinate on $\CP^1$.  Let us write $V = g(z) \partial_z$, $V_+ = g_+(z) \partial_z$, $V_- = g_-(z) \partial_z$.  

Let us assume that $D_-$ contains $\infty$ and $D_+$ contains $0$.  Before deforming, $z$ defines a coodinate on $D_+$ and $1/z$ on $D_-$.  After applying the RG flow, we have a coodinate $z_+$ on $D_+$ and $1/z_-$ on $D_-$, which satisfy
\begin{equation} 
z_+ + 2\lambda \eps g(z_+)  = z_-. 
\end{equation}
Here $\eps$ is the parameter of the deformation and we work modulo $\eps^2$.   
We have
\begin{equation} 
\eps g(z_+) = \eps g_+(z_+) + \eps g_-(z_-) 
\end{equation}
because we are working modulo $\eps^2$. 

Therefore, the relation is
\begin{equation} 
z_+ + 2 \lambda \eps g_+(z_+) = z_- - 2 \lambda \eps g_-(z_-). 
\end{equation}
Let us define a new coordinate 
\begin{equation} 
\begin{split} 
w_+ &=  z_+ + 2 \lambda \eps g_+(z_+)\\
w_- &=  z_- - 2 \lambda \eps g_-(z_-)
\end{split}
 \end{equation}
Here $w_+$ is a coordinate on $D_+$ and $1/w_-$ is a coordinate on $D_-$.  The relation is simply $w_+ = w_-$, so that $w = w_+$ is a global coordinate after performing the gluing.  

To determine how $\nabla_+$ and $\what{\alpha}$ vary under the RG flow, we need to determine how they appear in the new global coordinate $w$.  Any geometric object, when written as a function of $w_+$ instead of $z_+$,   is transformed by the Lie derivative of the vector field $-2 \lambda \eps \mc{L}_{V_+}$. 

For instance, $z_+ = w_+ - 2 \lambda \eps V_+ (z_+)$ and $V_+(z_+) = g(z_+)$.  Similarly, if $\what{\alpha} = \what{\phi}_+(z_+) \d z_+$, then, in the new coordinates, 
\begin{equation}
\begin{split} 
 \what{\alpha} = \what{\phi}_+(w_+) \d w_+ - 2 \lambda \eps \mc{L}_{V_+} ( \phi_+(w_+) \d w_+) .  
\end{split} 
\end{equation}
This implies that
\begin{equation} 
\delta_{RG} \what{\alpha} = -2\lambda \mc{L}_{V_+} \what{\alpha} 
\end{equation}
as desired.  A similar argument shows that $\delta_{RG} \nabla_+ = -2 \lambda \mc{L}_{V_+} \nabla_+$.  In fact, this is a consequence of the statement for $\what{\alpha}$, because the $\lambda$-connection $\nabla_+$ is determined by $\what{\alpha}$.

\end{proof}

The justification for calling this the RG flow is that, in the generalized CY manifold $\what{X}_{\what{\Sigma}}$ associated to the planar spectral curve,  this flow measures the failure of the geometry to be scale invariant.  

Before backreacting, the generalized CY manifold $X_{\Sigma} = V \times \R^2$ has a natural action of $\R_{> 0}$, by scaling $\R^2$ and the fibres of $V$ in opposite ways.   When we backreact, we do not expect the resulting geometry to be scale invariant; the failure of scale invariance corresponds to the RG flow.   We will prove the following proposition later:  
\begin{proposition}
The generalized CY manifold associated to a planar spectral curve is not invariant under scaling. Instead,  the scaling action modifies the generalized CY structure according to the RG flow on the moduli of planar spectral curves. 
\end{proposition}
This leads to the conjecture that the RG flow on the space of planar spectral curves is equivalent to the one on planar integrable field theories, at all orders in the 't Hooft coupling. 
\subsection{One-loop RG flow}
At one-loop order, the interpretation of the RG flow of integrable models in terms of holomorphic Chern-Simons is now quite well understood \cite{Hassler:2020xyj, Delduc:2020vxy,Derryberry:2021rne,Levine:2022hpv,Levine:2023wvt,Lacroix:2024wrd,Lacroix:2025ias} . In particular, in  \cite{Derryberry:2021rne} a conjecture was made for the one-loop RG flow in terms of the periods of the classical spectral curve.  It was shown in \cite{Lacroix:2024wrd, Lacroix:2025ias} that this conjecture is true for essentially all examples at genus $0$ and $1$\footnote{We should also mention the beautiful recent work \cite{Komatsu:2026cqz} which shows that in time-dependent integrable models, the time dependence of the coupling must be by the RG flow. It would be fascinating to give a holographic derivation of their results.}.

At leading order in $\lambda$, the RG flow we have described gives a flow on the classical spectral curve.  This flow is given by the Beltrami differential
\begin{equation} 
\frac{2 \lambda}{\alpha} \in \breve{H}^1(\Sigma, T \Sigma). \label{eqn:oneloopRG} 
\end{equation} 
We will show that this flow coincides with the one described in \cite{Derryberry:2021rne}.  (In section \ref{sec:genus0} we give a more explicit expression for the RG flow for genus $0$ models, and a more explicit proof of the fact that the holographic flow coincides with the one conjectured in \cite{Derryberry:2021rne}). 

Let us explain the flow of \cite{Derryberry:2021rne}. The continuous moduli of a classical spectral curve $(\Sigma= \Sigma_+ \cup \Sigma_-, \alpha)$ are entirely encoded in the periods of $\alpha$. There are three types: periods along closed cycles of $\Sigma$, and periods around a pole of $\alpha$, and periods between zeroes of $\alpha$.  (There are obvious linear relationships between these periods). The conjecture is that the flow of the classical spectral curve is such that
\begin{enumerate} 
\item Periods around poles of $\alpha$, or around closed cycles of $\Sigma$, do not flow.
\item Periods for paths between two chiral zeroes, or two anti-chiral zeroes, do not flow.
\item The period for a path $\gamma$ from a chiral zero to an anti-chiral zero flows by
\begin{equation} 
\dpa{\log \mu} \int_\gamma \alpha = 2 \lambda. 
\end{equation} 
\end{enumerate} 
\begin{lemma} 
The RG flow given by the Beltrami differential \eqref{eqn:oneloopRG} is equivalent to the one described by the flow of periods. 
\end{lemma}
\begin{proof}
We let $V$ be the meromorphic vector field $1/\alpha$ (abusing notation slightly; this is the $\lambda \to 0$ limit of what we called $V$ earlier). 

We can represent the \v{C}ech cocycle associated to $V$ using any open coveron which $V$ is regular on the double overlaps.  We can therefore choose $\Sigma_+$ to be a disc which contains all chiral zeroes, and assume that $\Sigma_+$ and $\Sigma_-$ intersect on an annulus. We can take a coordinte $z$ on $\Sigma_+$ so that $\Sigma_+$ is the region $\abs{z} < 1$ and $\Sigma_+ \cap \Sigma_-$ is the region $1-\eps < \abs{z} < 1$.  In this patch the one-form $\alpha$ is $\phi(z) \d z$, and $V = \phi(z)^{-1} \partial_z$.

The deformed surface is  obtained by cutting out the disc $\abs{z} < 1$ and regluing after applying the coordinate transformation $1 + 2 \lambda V$, working modulo $\lambda^2$. The one-form on the surface is the same on the two patches; since $\mc{L}_V \alpha = 0$ this one-form makes sense on the new surface.
 
Any closed cycle on $\Sigma$, or cycle around a pole of $\alpha$, can be assumed to be in $\Sigma_-$. Clearly, the period of this will not flow under the Beltrami differential associated to $V$.  Similarly, any period between anti-chiral zeroes will be in $\Sigma_-$ and will not flow. We can choose a basis  of  the linearly independent periods so that the periods between chiral zeroes lie in $\Sigma_+$ and also do not flow.

Now consider a period between a chiral and anti-chiral zero, along a path $\gamma(t)$ which at $0$ is at a chiral zero and at $2$ is at an anti-chiral zero. We need to compare the period of this path on the original surface $\Sigma$, where the two patches are glued by the identity, and on the surface where they are glued by the identity plus $2 \lambda V$.  

This path passes through the region $1 - \eps < \abs{z} < 1$  and by deforming the path we can assume that for $\gamma(t) = t$ for $1 - \eps \le t \le 1$.   Now let us apply the RG flow, which cuts out the region $\abs{z} < 1$ and reglues after applying the coordinate transformation $1 + 2 \lambda V$.  After this gluing procedure, the path $\gamma$ looks like  $\gamma(t) = t +2 \lambda \phi(t)^{-1}$ for $1 - \eps \le t < 1$,  but to it remains $\gamma(t)$ for $t > 1$. 

 This path is not connected. To make it connected we add an infinitesimal path from $1$ to $1 + 2\lambda \phi(1)^{-1}$.  The period on the new surface is the period on this new, connected path; we see that it differs from the period on the original surface by the period of the one-form along the infinitesimal path from $1$ to $1 + 2 \lambda \phi(1)^{-1}$. Since the one-form is $\phi(z) \d z$, the variation of the period is exactly $2 \lambda$, as desired. 
\end{proof}
\subsection{Line defects in $4d$ Chern-Simons}
Classically, $4d$ Chern-Simons has line defects which are labelled by representations of the Yangian. These defects live at points in $\Sigma$ and wrap a line in the two-dimensional space-time. These defects correspond to non-local conserved charges in integrable field theories.

In the planar theory, the position of these line defects flows under the RG flow.

The flow is very simple: we have seen that the planar spectral curve flows by a Beltrami differential represented, in the \v{C}ech description, by a vector field $2 \lambda V = \lambda/\what{\alpha}$ on $\Sigma_+ \cap \Sigma_-$.

We show that a line defect in the chiral patch flows by the vector field $\lambda V$, and in the anti-chiral patch flows by the vector field $- \lambda V$.

We can consider a purely chiral planar spectral curve, where $\Sigma = \Sigma_+$. This corresponds to a purely chiral two-dimensional field theory, and the line defects are integrable  Kondo defects \cite{Gaiotto:2020dhf}.  In this case, only the defect flows, not the spectral curve itself.  As the defect flows, it traces out a (parameterized) path in $\Sigma_+$, which we can describe in terms of the $\lambda$-connection on $\Sigma_+$. 

Given any embedded path $\gamma \subset \Sigma_+$, the connection $\lambda^{-1} \nabla_+$ on $T^\ast \Sigma_+$ restricts to a complex connection one on the cotangent bundle of $\gamma$.  
\begin{lemma} 
The RG trajectory of a line defect is a path on which $\lambda^{-1} \nabla_+$ restricts to a real connection on the cotangent line to the path. 
\end{lemma}
\begin{proof} 
 Since the path is a trajectory of $\lambda V$, the one form $\lambda^{-1} \what{\alpha}$ is real on the path. If we trivialize the cotangent bundle of the path using $\lambda^{-1} \what{\alpha}$, the connection $\lambda^{-1} \nabla_+$  has $1$-form $\lambda^{-1} \what{\alpha}$, which is real. 
\end{proof}

\subsection{Order defects}
$4d$ Chern-Simons has another kind of defect, called an order defect, where at a point $z \in \Sigma$ we introduce chiral or anti-chiral free fermions living in some representation of the gauge algebra. Since we are working with $GL(N)$ $4d$ CS, we will take fermions living in $N_f$ copies of the fundamental plus anti-fundamental.  

We find, as part of our holographic analysis, that a chiral order defect at a point $\Sigma_+$ introduces a pole in the the $\lambda$-connection $\nabla_+$ so that the monodromy of the connection $\lambda^{-1} \nabla_+$ around the defect is $e^{-2 \pi \i N_f / N}$.  Similarly, for an anti-chiral order defect at  a point $\Sigma_-$, the monodromy of the connection $-\lambda^{-1}\nabla_-$ is $e^{-2\pi \i \br{N}_f / N}$.  

This provides a cross-check of our normalizations. Let us consider $\Sigma = \CP^1$, $\alpha = \d z$, and let us introduce a single chiral order defect at $z = 0$.  Then, $\nabla_+$ becomes
	\begin{equation} 
\partial_z - \hbar \d \log z.
\end{equation}
Effectively, $\alpha = \d z  - \hbar \d \log z$.

We normalized our action to be (incorporating factors of $\hbar$) 
$-\frac{1}{8 \pi^2 \hbar \i}  \int \alpha CS(A)$.  Upon reduction to $3$-dimensions on the circle  $\abs{z} = 1$, with $\alpha = \d z - \hbar \d \log z$ this gives Chern-Simons theory at level $1$.  The boundary of the $3d$ space-time that corresponds to $z = 0$ has Neumann boundary conditions coupled to chiral free fermions. Neumann boundary conditions coupled to free fermions are consistent exactly at level $1$, because the free fermion algebra is the Kac-Moody algebra at level $1$.   

\subsection{Planar spectral curves with nodes}
	$4d$ Chern-Simons on Riemann surfaces with nodes appear when one studies the $\lambda$-deformation of integrable models \cite{Delduc:2019whp}. In this section we will briefly describe how one needs to modify the concept of planar spectral curve in the case of nodal curves. 
    
    A classical spectral curve with a node can be described as follows. If $\Sigma$ is the nodal surface, we let $\Sigma'$ be the normalization, obtained by cutting $\Sigma$ at the nodal point. We have two special points $p,q \in \Sigma'$ which are glued to obtain $\Sigma$.  

	The classical spectral curve structure on $\Sigma$ is given by a decomposition of $\Sigma$ into chiral and antichiral regions, as before;  and a one-form on the normalization $\Sigma'$ with first order poles at $p$ and at $q$, with opposite residue.  On the nodal surface $\Sigma$, a one-form like this is called a logarithmic one-form.  

	To describe planar spectral curves with nodes, we need only say how the connection $\nabla_+$ or $\nabla_-$ behave at the nodes. Let us assume that the points $p,q$ are in the chiral patch. Then, we ask that $\nabla_+$ is a connection, with a first-order pole at $p$ and $q$,  on the bundle of one-forms with first-order poles at $p$ and at $q$. In a local trivialization near $p$, we have
	\begin{equation} 
	\nabla_+ = \lambda\partial + a \d \log z + \text{ regular }. 
	\end{equation} 
	The coefficient $a$ of $\d \log z$ is independent of the frame, and we call it the residue of $\nabla_+$. We ask that the residues of $\nabla_+$ at $p$ and $q$ are equal and opposite.   

This definition is motivated by considering the limit of a smooth planar spectral curve as it develops a node.
	\section{Building planar spectral curves and anomaly inflow}
	In this section we will show how to build planar spectral curves.  

We will first show that, at a point when the classical one-form $\alpha$ has a zero, the quantized one-form $\what{\alpha}$ acquires a pole.  We view this as a kind of anomaly inflow. At a zero of $\alpha$, there are chiral (or anti-chiral) degrees of freedom, which will have a chiral anomaly.  The effective action $\int \what{\alpha} CS(A)$ fails to be gauge invariant when $\what{\alpha}$ has a pole, and this failure cancels the chiral anomaly. 

Fix any $\lambda$ connection $\nabla_+$ on $\Sigma_+$, and we let $\what{\alpha}$ be any solution to the Bernoulli equation.  
	\begin{lemma}
	Let $D \subset \Sigma$ be a disc which does not contain any poles of $\nabla$. Then, for $\lambda$ sufficiently small, the contour integral 
	\begin{equation} 
	\frac{1}{\lambda 2 \pi \i} \oint_{\partial D} \what{\alpha} 
	\end{equation}
is the number of zeroes of the original one form $\alpha$ inside the contour. 
	 	\end{lemma}
	\begin{proof}
	One way to check this is to use the integral expression for $\what{\alpha}$ and use standard results on the small-$\lambda$ asymptotics of such integrals.

	Here is a simpler argument. If we trivialize the canonical bundle using $\what{\alpha}$, then, in this frame,
\begin{equation} 
\nabla = \lambda \partial + \what{\alpha}. 
\end{equation} 
Recall that $\nabla$ is regular on $D$. The connection $\lambda^{-1} \nabla$ is a holomorphic connection, and therefore flat on $D$.  This implies that the monodromy of $\frac{1}{\lambda} \nabla$  on $\partial D$ is trivial.  This tells us that
	\begin{equation} 
	\exp \left( \frac{1}{\lambda} \oint_{\partial D} \what{\alpha}   \right) = 1 
	\end{equation} 
	Therefore, $\frac{1}{2 \pi \i \lambda} \oint \what{\alpha}$ is an integer, which can not change as we vary $\lambda$ unless a pole of $\what{\alpha}$ crosses the contour. We can therefore calculate this integer by working in series in $\lambda$.

	To check the value of the integer, it suffices to compute $\what{\alpha}$ to order $\lambda$ in coordinates. If $\nabla = \lambda \partial + z^n \d z$, then $\what{\alpha} = g(z) \d z$ where $g(z)$ satisfies
	\begin{equation} 
	\lambda \partial_z g(z) + z^n g(z) = g(z)^2. 
	\end{equation} 
	Writing $g(z) = g_0(z)+ \lambda g_1(z)$ and working modulo $\lambda^2$, we see that $g_0(z) = z^n$ and  
	\begin{equation} 
	 n z^{n-1}  =  z^n g_1(z)
	\end{equation} 
	so that $g_1(z) = n/z$. 

	\end{proof}
If we have an anti-chiral zero, then the residue of $-\frac{1}{\lambda 2\pi\i} \what{\alpha}$ will be the order of the zero.   

\subsection{Holographic analysis of chiral anomalies}
Because the total residue of $\what{\alpha}$ must vanish, there is a potential obstruction to building a planar spectral curve.  If $\Sigma$ is a closed surface, and the original one-form $\alpha$ has no poles,  then there must be the same number of chiral and anti-chiral defects (counted with multiplicity).      

From the perspective of the integrable field theory, this condition is required to cancel the anomalies associated with gauging chiral degrees of freedom.  Let us sketch why, in the case when there are only chiral defects, the effective two-dimensional theory is anomalous.

With only chiral defects, $4d$ Chern-Simons on $\R^2 \times \Sigma$ is a deformation \cite{Costello:2019tri} of holomorphic BF theory on $\C \times \Sigma$. Holomorphic BF theory has fields $A \in \Omega^{0,1}(\C \times \Sigma, \mf{g})$ and $B \in \Omega^{2,0}(\C \times \Sigma, \mf{g})$ with action
\begin{equation} 
\int \op{tr}(B \wedge F(A)) . 
\end{equation} 
The deformation is obtained by adding on the term $\alpha \op{tr}(A \partial A)$ and modifying the gauge transformations so that $\delta B = \alpha \partial \chi$ (where $\chi$ is the generator of gauge transformations).  If $w$ is a coordinate on $\C$, then $\frac{1}{\alpha} B$ can be identified with the $w$ component of the 4d Chern-Simons gauge field; note that it has poles at the zeroes of $\alpha$, which is what we impose for chiral disorder defects. 

The reason for introducing this reformulation of $4d$ Chern-Simons with disorder defects is to discuss the chiral anomaly in two dimensions.  If we just consider holomorphic BF theory, and not its deformation to $4d$ CS, one finds the effective $2d$ theory is a gauged $\beta-\gamma $ system. The zero mode of $A$ contributes spin zero bosonic fields $\gamma \in H^{0,1}(\Sigma,\mf{g})$, and the zero mode of $B$ contribute spin one bosonic fields $\beta \in H^{1,0}(\Sigma, \mf{g})$.  Before introducing the $4d$ CS deformation, the action is simply $\beta \dbar \gamma$, and we BRST reduce by the adjoint action of $\mf{g}$.

This system is anomalous unless $g = 1$. We have $g$ copies of the adjoint-valued $\beta-\gamma$ system, which has Kac-Moody level $g$ times the critical level (here $g$ is the genus of $\Sigma$).   The deformation which in four dimensions is $\alpha \op{tr}(A \partial A)$ gives, in two dimensions, a modification of the current generating the $\mf{g}$ action on the $\beta-\gamma$ system. The current is originally $J_c = \beta_a \gamma_b f_c^{ab}$; the deformation adds a term $\ip{\alpha,\partial \gamma_c}$ (where $\alpha$ is viewed as a linear function on $H^{0,1}(\Sigma)$). This modification does not change the BRST anomaly.

This BRST anomaly can be cancelled if we introduce fermionic fields, which will shift the level back to the critical level. Let us take $\mf{g} = \mf{sl}_N$. Then, we can cancel the anomaly by introducing $(g-1)$ adjoint valued (complex) chiral fermions $\psi_a, \psi^a$; or alternatively, by introducing $2(g-1)N$ fundamental complex chiral fermions $\psi_i, \psi^i$.  
 
Let us check that this anomaly cancellation can be seen holographically (this will serve as a cross-check on our normalizations and signs).  We have seen that having $N_f$ chiral fermions modifies the set-up so that $\nabla_+$, and $\what{\alpha}$, will have a pole with residue $-2 \pi \i \lambda\frac{N_f}{N}$.  Since, on a surface of genus $g$, a one form has $2g-2$ zeroes, before we introduce any fermions, the sum of the residues of $\what{\alpha}$ is $2 \pi \i \lambda (2g-2)$.  If we introduce also $(2g-2)N$ fundamental chiral fermions, the total residue will be zero and the anomaly is cancelled.  

\subsection{Behaviour of $\what{\alpha}$ at poles}
In most examples of interest, the original one-form $\alpha$ has poles, which will also contribute to the poles of $\what{\alpha}$ and affect the analysis of anomalies.  
  
	Let us consider the behaviour of $ \lambda$ connection $\nabla_{+}$ near a second order pole. There exists a coordinate $z$ near such a pole, and a trivialization of the bundle of forms with second order poles, so that $\nabla_{+}$ takes the form
	\begin{equation} 
	\nabla_{\pm} =  \lambda \partial + z^{-2} \d z + c(\lambda) z^{-1} \d z + \text{ regular terms }. 
	\end{equation}
	It is important to note that $c(\lambda)$ is independent of gauge choice. Different gauges are different trivializations of the bundle of one-forms with second order poles at $z = 0$; therefore gauge transformations are holomorphic functions of $z$ with no pole at $z = 0$.  We say the \emph{level} of $\nabla_{+}$ at a second order pole is $\frac{1}{\hbar} c(\lambda)$, where we recall that $\lambda = \hbar N$. 

We use the term level, because, if we reduce to $3$ dimensions on a circle surrounding the pole, this is the level of $3d$ Chern-Simons theory.

	For a $-\lambda$ connection, we make the same definition except reversing the sign.  	\begin{lemma} 
	If $\nabla_+$ has level $k$ at a second order pole, then $\what{\alpha}$ has residue $2 \pi \i \hbar k$ at this pole.  
	\end{lemma}
	\begin{proof} 
	Recall that $\nabla_+$ is a connection, with a second order pole at a point $p$, on the bundle of one-forms with a second order pole $p$.  Further, $\what{\alpha}$ is a one-form with a second-order pole at $p$. If we frame this bundle using $\what{\alpha}$, then $\nabla_+  = \lambda \partial + \what{\alpha}$. This immediately implies that the residue of $\what{\alpha}$ is $2 \pi \i \hbar k$. 
	\end{proof}
\subsection{Building planar spectral curves}
	Now suppose we have a classical spectral curve $(\Sigma, \Sigma_+, \Sigma_-,\alpha)$. Suppose that $\Sigma_+ \cap \Sigma_-$ is an annulus. Let us upgrade $\alpha$ to a $\lambda$-connection on $\Sigma_+$ and a $-\lambda$-connection on $\Sigma_-$.  

	We say the \emph{total level} on $\Sigma_+$ is the sum of:
	\begin{enumerate} 
	\item $N$ times the number of zeroes of $\alpha$ on $\Sigma_+$ counted with multiplicity.
	\item The levels at each second order pole of $\nabla_+$.  
	\end{enumerate} 
	\begin{lemma} 
	Suppose that the total level on $\Sigma_+$ and on $\Sigma_-$ agree. Let $\what{\alpha}_{\pm}$ be the solutions to the Bernoulli equation coming from $\nabla_{\pm}$. 

	Then there exists a coordinate transformation 
	\begin{equation} 
	\phi : \Sigma_+ \cap \Sigma_- \to \Sigma_+ \cap \Sigma_- 
	\end{equation}
	so that $\phi^\ast \what{\alpha}_- = \what{\alpha}_+$. Further, we can assume that $\phi$ is the identity modulo $\lambda$. 
	\end{lemma}
	This means that if we define a new surface $\what{\Sigma}$, which has patches glued $\Sigma_{\pm}$ glued by $\phi$, then $\what{\Sigma}$ is a planar spectral curve. 
	\begin{proof} 
	We assumed that $\Sigma_+ \cap \Sigma_-$ is an annulus. Let $C \in \Sigma_+ \cap \Sigma_-$ be the circle in the annulus. If $k$ is the total level on $\Sigma_+$, then
	\begin{equation} 
	\oint_C \what{\alpha}_+ = 2 \pi \i \lambda k.  
	\end{equation}
	Here we give $C$ the orientation it receives as the boundary of $\Sigma_+$.  Similarly, 
	\begin{equation} 
	\oint_{\br{C}} \what{\alpha}_- = - 2 \pi \lambda k. 
	\end{equation}
	Here $\br{C}$ indicates the opposite orientation, and we have the opposite sign on the right hand side because $\nabla_-$ is a $-\lambda$ connection.    
	Therefore,  $\oint_C \what{\alpha}_+  =\oint_C \what{\alpha}_-$, so that $\alpha_{\pm}$ are cohomologous.

	This implies that we can build the desired $\phi$. By induction, suppose we have build $\phi_{< n}$ modulo $\lambda^n$ so that $\phi_{< n}^\ast \what{\alpha}_- = \what{\alpha}_+$ modulo $\lambda^n$.  Let us arbitrarily extend $\phi_{< n}$ to a transformation $\til{\phi}_n$ defined modulo $\lambda^{n+1}$. 

	Since $\what{\alpha}_{\pm}$ are cohomologous, and $\til{\phi}_n$ is the identity modulo $\lambda$, $\til{\phi}_n^\ast \what{\alpha}_-$ and $\what{\alpha}_+$ remain cohomologous.  Therefore
	\begin{equation} 
	\til{\phi}_n^\ast \what{\alpha}_- = \what{\alpha}_+ + \lambda^n \partial \gamma \mod \lambda^{n+1}. 
	\end{equation}
	for some function $\gamma$.  We can correct $\til{\phi}_n$ to
	\begin{equation} 
	\phi_n = \til{\phi}_n + \lambda^n W 
	\end{equation}
	where $W$ is the vector field
	\begin{equation} 
	W = - \frac{\gamma}{\alpha}  
	\end{equation}
	on $\Sigma_+ \cap \Sigma_-$. Since $\alpha = \what{\alpha}_-$ modulo $\lambda$, we see that
	\begin{equation} 
	\lambda^n \mc{L}_W \what{\alpha}_- = -\lambda^n \partial \gamma \mod \lambda^{n+1} 
	\end{equation} 
	so that
	\begin{equation} 
	\phi_n^\ast \what{\alpha}_- = \what{\alpha}_+ 
	\end{equation}
	and we can continue the induction. 
	\end{proof} 
	Note that the $\phi$ we have built is not unique: if $V = 1/\what{\alpha}_-$, then we can compose $\phi$ with the exponential of $V f(\lambda)$ for any function of $\lambda$.  This ambiguity is related to the choice of scale, since $V$ acts as the RG flow.  

	Now, in the case $\Sigma$ is of genus $0$, the surface $\what{\Sigma}$ glued from $\phi$ must also be of genus $0$.  This means that in this case, instead of modifying the surface, we can change the $\lambda$-connections $\nabla_{\pm}$ so that the Bernoulli equation holds.  We will give an explicit algorithm for how to modify $\nabla_{\pm}$ shortly.  More conceptually, we can see how  to do this using the Birkhoff decomposition.  We can write our change of coordinates $\phi$ as $\phi = (\phi_+)^{-1} \phi_-$ where $\phi_+$ extends across $\Sigma_+$ and $\phi_-$ extends across $\Sigma_-$. Then, we have
	\begin{equation} 
	\phi_-^\ast \what{\alpha}_- = \phi_+^\ast \what{\alpha}_+.  
	\end{equation}  
	The form $\phi_{+}^\ast \what{\alpha}_+$ solves the Bernoulli equation for $\phi_+^\ast \nabla_+$, and similarly $\phi_-^\ast \what{\alpha}_-$ solves the equation for $\phi_-^\ast \nabla_-$.  It follows that $(\phi_+^\ast \nabla_+, \phi_-^\ast \nabla_-)$ define a planar spectral curve, since their solutions to the Bernoulli equation agree.

\section{Explicit computations in genus $0$} \label{sec:genus0}
In this section we will compute some explicit planar spectral curves and two-loop RG flows for genus zero surfaces. 
	
	\subsection{Planar spectral curves for theories with chiral and anti-chiral disorder defects}
    \label{sec:chiralandantichiral}
	We will start by considering the $4d$ Chern-Simons model from a one form with $n+1$ second order poles and $2n$ zeroes.  We will place one of the second order poles at $\infty$; the remaining poles are $p_i$. The zeroes will be denoted  $q_i^{\pm}$, where the superscript $\pm$ indicates whether they are chiral or anti-chiral zeroes. The one-form is
	\begin{equation} 
	\alpha = \d z \phi(z) =  \d z \frac{\prod (z-q_i^+)(z-q_j^-) } { \prod (z - p_k)^2 } 
	\end{equation}
	
    In our computations, we will take the chiral and anti-chiral patches to  be as large as possible.  Thus, we let the chiral patch $\Sigma_+$ be the complement of the $q_i^-$, and the anti-chiral patch $\Sigma_-$ be the complement of the $q_i^+$.  
    
	On the chiral patch we will have a $\lambda$-connection with second order poles at $p_k$, on the bundle of one-forms with second order poles at $p_k$.  This can be viewed as a meromorphic $\lambda$-connection on $\CP^1$, with second-order poles at the $p_k$ and arbitrary poles at the $q_i^-$. Similarly, on the anti-chiral patch we will have a $-\lambda$-connection, on the same bundle, with second-order poles at the $p_k$ and arbitrary poles at the $q_i^+$. 

The $\pm\lambda$ connections are connections on a non-trivial bundle. To write formulae for them, we need to consider a frame.   In a patch near some $p_k$, we can trivialize the bundle of one-forms with second order poles at $p_k$ by $(z-p_k)^{-2} \d z$. In this frame the $\lambda$-connection will be of the form
	\begin{equation} 
	\nabla_+ = \lambda \partial + \alpha + \lambda \alpha^+_1 + \dots  
	\end{equation}
The forms $\alpha^+_i$ will have second order poles at $p_k$. 

If, instead, we trivialize using the one-form $\d z$, this expression will change by adding the terms $2 \lambda (z - p_k)^{-1} \d z$ which have first order poles at the $p_k$.  Since we allow the connection one-form to have second order poles at the $p_k$, we can just absorb these terms into the one-loop ambiguity $\alpha^+_1$. 

In what follows we will trivialize the canonical bundle using $\d z$ to keep formulae simple.  Recall that we take $\infty$ to be one of the poles in our one-form, so that $\d z$ is a section of the bundle of one-forms with second-order poles at $p_k$ and $\infty$, and is a framing of this bundle away from the finite $p_k$.   

In the trivialization $\d z$, the $\lambda$-connection on the chiral patch and $-\lambda$ connection on the anti-chiral patch take the form
	\begin{equation} 
    \begin{split}
    \nabla_+ &= \lambda \partial + \alpha + \lambda \alpha^+_1 + \dots  \\
	\nabla_- &=  - \lambda \partial + \alpha + \lambda \alpha^-_1 + \dots  
	\end{split}
    \end{equation} 
    where the correction terms $\alpha^{\pm}_i$ have second-order poles at $\infty$ and $p_k$, are regular at $q_i^{\pm}$ and have arbitrary poles at $q_i^{\mp}$. 
    
	In order for this data to define a planar spectral curve, we need the unique solutions $\what{\alpha}_{\pm}$ to the Bernoulli equation for $\nabla_{\pm}$ to agree.

	The statement that the solutions to the Bernoulli equation for $\nabla_{\pm}$ agree implies that we can express $\nabla_-$ in terms of $\nabla_+$: 
	\begin{equation} 
	\nabla_- = - \nabla_+ + 2 \what{\alpha}_+. 
	\end{equation}
	Since $\nabla_-$ must be regular on $D_-$ except for second order poles at $p_i$, this equation constrains the poles that can appear in $\what{\alpha}_+$ and in $\nabla_+$ on $D_-$.   

	Our strategy to build the planar spectral curve is to choose counter-terms $\alpha_i^+$ in the definition of $\nabla_+$, so that, if we then find the unique solution to the Bernoulli equation $\what{\alpha}^+$ built from these counter-terms, the connection $\nabla_- = - \nabla_+ + 2\what{\alpha}^+$ is regular at $q_i^-$.

Let us write an explicit algorithm to implement this strategy. Since we have trivialized the canonical bundle using $\d z$, we can write one-forms as scalar functions. Let us introduce some notation for these scalar functions: 
\begin{equation}
\begin{split}
\nabla_+ &= \lambda \partial + \phi (z) \d z + \sum_{i \ge 1} \lambda^i \phi_i^+ \d z \\
\what{\alpha} &= \what{\phi}(z) \d z = \sum \lambda^i \what{\phi}_i (z) \\
\nabla_- &= -\lambda \partial + \phi (z) \d z+ \sum_{i \ge 1} \lambda^i \phi_i^- \d z
\end{split}
\end{equation}
We have the following constraints on the poles of these scalar functions:
\begin{enumerate}
\item The $\phi_i^+$ are regular at the points $q_i^+$ and $\infty$, and can have second order poles at the points $p_i$ and arbitrary poles at $q_i^-$.  
\item The $\phi_i^-$ are regular at $q_i^-$, $\infty$ and can have second order poles at $p_i$. They can have arbitrary poles at $q_i^+$.
\item The $\what{\phi}_i$ can have second order poles at $p_i$, and arbitrary poles at $q_i^{\pm}$. Further $\what{\phi}_i$ is regular at $\infty$.   
\end{enumerate}
Further, these functions satisfy some equations which come from the Bernoulli equation and the identity that $\nabla_- = - \nabla_+ + 2 \what{\alpha}$.   This second identity simply tells us that
\begin{equation}
\phi_i^-(z) = 2 \what{\phi}_i(z) - \phi_i^+(z). 
\end{equation}

The Bernoulli equation, at order $n$ in $\lambda$, tells us that
\begin{equation}
\partial_z \what{\phi}_{n-1} + \phi(z) \what{\phi}_n(z) + \sum_{i = 1}^n \phi_i^+(z) \what{\phi}_{n-i}(z) = \sum_{r = 0}^n \what{\phi}_r(z) \what{\phi}_{n-r}(z).
\end{equation}
At order $0$, the Bernoulli equation tells us that 
\begin{equation}
 \what{\phi}_0(z) = \phi(z)
\end{equation}
Using this, we can slightly rewrite the Bernoulli equation to give an expression for $\what{\phi}_n(z)$:
\begin{equation}
 \what{\phi}_n(z) = \phi(z)^{-1} \partial_z \what{\phi}_{n-1}(z) + \phi_n^+(z) + \sum_{i = 1}^{n-1}  \phi(z)^{-1}  ( \phi_i^+(z) - \what{\phi}_i(z) )\what{\phi}_{n-i}(z) \label{equation:bernoulli_inductive}
\end{equation}
Note that the right hand side does not contain $\what{\phi}_n(z)$, so that this gives an inductive way to determine $\what{\phi}_n^+(z)$ from $\phi_i^+(z)$.    We will repeatedly use the simple expression for $\what{\phi}_1(z)$:
\begin{equation}
\what{\phi}_1(z) = \partial_z \log \phi(z) + \phi_1^+(z). \label{equation_bernoulli1storder}
\end{equation}
From this we find a more complicated expression for $\what{\phi}_2(z)$:
\begin{equation}
\what{\phi}_2 (z) = \phi(z)^{-1} \partial_z \phi_1^+(z) + \phi(z)^{-1} \partial_z^2 \log \phi(z)   + \phi_2^+(z) - \phi(z)^{-1} (\partial_z \log \phi(z) ) ( \partial_z \log \phi(z) + \phi_1^+(z) ). \label{equation_bernoulli2ndorder}
\end{equation}

The strategy to build the planar spectral curve is the following.
\begin{enumerate}
\item We suppose we have built a planar spectral curve to order $n-1$, so that we have functions $\phi_k^{\pm}$ and $\what{\phi}_k$ satisfying the equations above and the  constraints on their poles, for $k < n$.
\item We use the Bernoulli equation in the form \ref{equation:bernoulli_inductive} to determine $\what{\phi}_n$, as a sum of $\phi_n^+$ and quantities already determined.   
\item Then we write
\begin{equation}
\begin{split}
\phi_n^- &= - \phi_n^+ + 2 \what{\phi}_n \\
&= \phi_n^+(z) +  2\phi(z)^{-1} \partial_z \what{\phi}_{n-1}(z) + 2\sum_{i = 1}^{n-1}  \phi(z)^{-1}  ( \phi_i^+(z) - \what{\phi}_i(z) )\what{\phi}_{n-i}(z) 
\end{split}
\end{equation}
The expression for $\phi_n^-$ may have poles at $q_i^-$, which is not allowed.  We choose $\phi_n^+$ to cancel these poles, using the fact that $\phi_n^+$ can have arbitrary poles at the $q_i^-$ to achieve this. 
\end{enumerate}
We can see by induction that at every step, $\what{\phi}_n$ only has second order poles at $p_i$ and is regular at $\infty$. The order of poles of $\what{\phi}_n$ at $q_i^{\pm}$ can grow with $n$. 

We also see that at each order, there is some freedom to choose $\phi_n^+$.  We are allowed to add to $\phi_n^+$ an expression which is regular at the $q_i^-$, $q_j^+$, $\infty$ and has second order poles at $p_i$.  In other words, the poles in the ambiguity in $\phi_n^+$ at each order are exactly the poles in our original twist function $\phi(z)$. 

Note that also that, in this construction, we are free to make the $p_i$ depend on $\lambda$. Therefore, at each order in $\lambda$, the planar spectral curve has the following moduli: 
\begin{enumerate}
\item The positions of the $p_i$.
\item The ambiguity in the choice of $\phi_n^+$, which is a function with second order poles at the $p_i$ which is regular at $\infty$.
\end{enumerate}
These moduli are exactly the moduli of the original classical spectral curve. Different choices of $\phi_n^+$ at each order will correspond to adding on finite counter-terms in the quantum field theory, and so to different schemes. 

In what follows we will always use a scheme where the positions $p_i$ of the poles do not depend on $\lambda$. 

Now let us implement this algorithm working modulo $\lambda^2$. We have
 \begin{equation}
 \phi(z)  = \frac{\prod (z-q_i^+)(z-q_j^-) } { \prod (z - p_k)^2 } 
 \end{equation}
 From equation \ref{equation_bernoulli1storder} we have
 \begin{equation}
 \what{\phi}_1(z) = \phi_1^+(z) + \sum \partial_z \log (z - q_i^+) + \sum \partial_z \log (z - q_j^-) -2 \sum \partial_z \log (z - p_i).
 \end{equation}
 We have 
 \begin{equation}
 \phi_1^-(z) = \phi_1^+(z) +2 \sum \partial_z \log (z - q_i^+) + 2\sum \partial_z \log (z - q_j^-) -4\sum \partial_z \log (z - p_i).
 \end{equation}
 Our constraints on $\phi_1^-(z)$ mean it must  be regular at $q_j^-$.  To achieve this, we set
 \begin{equation}
\phi_1^+(z) = - 2 \sum \partial_z \log (z - q_j^-)   + 2 \sum \partial_z\log (z - p_j).
 \end{equation}
The terms with $\partial_z \log (z - q_j^-)$ are necessary to cancel the poles at $q_j^-$; the other terms are a scheme choice.  

 We end up with
	\begin{equation} 
	\begin{split} 
	\nabla_+ &= \lambda \partial + \phi(z) \d z   - 2 \lambda \sum \partial \log (z-q_j^-)  +  2\lambda  \sum  \partial \log (z - p_i)  \\ 
	\nabla_- &=   -\lambda \partial + \phi(z) \d z + 2 \lambda \sum \partial \log     (z-q_i^+)  - 2 \lambda \sum  \partial \log ( z - p_i) \\  
	 \what{\alpha} =& \phi(z) \d z + \lambda \sum \partial \log (z - q_j^+) - \sum \lambda \partial \log (z - q_j^- )     
	\end{split} 
	\end{equation}
Note that the expressions are symmetric under sending $\lambda \to - \lambda$, and switching the roles of chiral and anti-chiral quantitites. Note also we are continuing to write the connections $\nabla_{\pm}$ in the frame given by $\d z$.  

One can continue to run the algorithm to determine $\phi_2^+(z)$ and $\what{\phi}_2(z)$, but the expressions quickly become very complicated and not particularly illuminating. 
    
\subsection{Computing the $\beta$-function}
    \label{sec:computingbeta}
We have given an expression for the RG flow of a genus zero planar spectral curve.  As before, we let $V = \frac{1}{\what{\alpha}}$ and let us write $V= V_+ + V_-$, where $V_+$ is regular on the chiral patch (i.e.\ away from the $q_i^-$) and $V_-$ is regular on the anti-chiral patch (away from the $q_i^+$).   Then, under the RG flow, 
	\begin{equation} 
	\begin{split} 
	\delta_{RG} \nabla_+ &= -2 \lambda\mc{L}_{V_+} \nabla_+\\ 
	 \delta_{RG} \what{\alpha} &= -2 \lambda\mc{L}_{V_+} \what{\alpha} = 2 \lambda \mc{L}_{V_-} \what{\alpha} \\
	 \delta_{RG} \nabla_- &= 2 \lambda\mc{L}_{V_-} \nabla_-\\
	\end{split} 
	\end{equation}

It is worth highlighting the following consequence of these formulae:
\begin{lemma}
There is always a scheme choice so that all $\beta$-function coefficients are rational numbers.
\end{lemma}
\begin{proof}
Let us assume that the $p_i$, $q_i^{\pm}$ are rational.  It is clear that the inductive construction of the $\phi_n^+$, $\phi_n^-$, $\what{\phi}_n$ that we can choose these functions to be in $\Q(z)$, the field of rational functions in $z$ with rational coefficients.  The vector field $V$ therefore also has coefficients in $\Q(z)$. Similarly, we can choose a decomposition $V = V_+ + V_-$ so that $V_{\pm}$ also have coefficients in $\Q(z)$.   The flow is given by applying the Lie derivative of $V_+$ to $\nabla_+$, $\nabla_-$ and $\what{\alpha} = \what{\phi} \d z$. It immediately follows that $\delta_{RG} \phi_n^{\pm}$ and $\delta_{RG} \what{\phi}_n$ are still in $\Q(z)$. 
\end{proof}

	Now let us apply this to the example above. As before, we 
	\begin{equation} 
	\alpha = \phi(z) \d z 
	\end{equation} 
	where as before
	\begin{equation} 
	\phi(z) =  \frac{\prod (z-q_i^+)(z-q_j^-) } { \prod (z - p_k)^2  }  
	\end{equation}
	We gave an algorithm to determine 
	\begin{equation} 
	\what{\alpha} = \what{\phi}(z) \d z 
	\end{equation}
    up to an ambiguity at each order in $\lambda$ of a finite counter-term (i.e. a scheme choice). We saw that, in a particular scheme, 
	\begin{equation} 
	 \what{\phi} = \phi + \lambda \sum (z-q_i^+)^{-1} - \lambda\sum (z-q_j^-)^{-1} +   O(\lambda^2).
	\end{equation} 
	Then
	\begin{equation} 
	V = \frac{1}{\what{\phi}} \partial_z 
	\end{equation}
	To compute the RG flow, we need to compute the positive part $V_+$ of $V$. It is given by the contour integral:
	\begin{equation}
    \label{ec:contourVplus}
	V_+ = \oint_w \frac{\d w}{2 \pi \i (w-z)} \frac{1}{\what{\phi}(w)} \partial_z  
	\end{equation}
	where the integral is over a contour that encloses points $z$, $q_i^+$ but not the points $q_i^-$. (Note that $1/\what{\phi}$ is zero at the $p_i$, so it doesn't matter if the contour encloses some of these points).   

	Note that there is some ambiguity in the choice of $V_+$, because there are regular holomorphic vector fields on $\CP^1$. Different choices are related by global coordinate transformations. 

	We will study how the RG flow varies $\what{\alpha}$, or equivalently $\what{\phi}$. This encodes the information of the RG flow of $\nabla_{\pm}$, but it is easier to compute because $\what{\alpha}$ transforms as a one-form and not a connection. 

	\begin{equation}
    \label{ec:RGflowcontour}
	\begin{split} 
	 \delta_{RG} \what{\alpha} &= -2\lambda \mc{L}_{V_+} \what{\alpha} \\
	&= -2 \lambda \d z \partial_z   \oint_w \frac{1}{2 \pi \i  (w - z)} \d w \frac{\what{\phi}(z)}{\what{\phi}(w)}  
	\end{split} 
	\end{equation}
	This is of course equivalent to
	\begin{equation}
	 \delta_{RG}\what{\phi}(z) 
	= -2 \lambda  \partial_z   \oint_w \frac{1}{2 \pi \i  (w - z)} \d w\frac{\what{\phi}(z)}{\what{\phi}(w) } 
	\end{equation}
	Note that the pole at $w = z$ does not contribute, because the residue is independent of $z$. Only the poles at $w = q_i^+$ contribute, so we have
	\begin{equation}
	 \delta_{RG}\what{\phi}(z) 
	= -2 \lambda  \partial_z \sum_j \oint_{\abs{w - q_j^+} = \eps} \frac{1}{2 \pi \i  (w - z)} \d w\frac{\what{\phi}(z)}{\what{\phi}(w) } 
	\end{equation}
This expression is valid to all orders in $\lambda$. However, we will only implement it to compute the two-loop RG flow, because $\what{\phi}(z)$ becomes complicated at second order in $\lambda$. In the examples below, we evaluate the contour integral explicitly by writing $1/\what{\phi}$ in a partial fraction decomposition and using the residue theorem. This will allow us to compare the resulting flow with known results in the literature.

\subsection{RG flow and periods}
In \cite{Derryberry:2021rne} a conjecture was presented for the one-loop RG flow in terms of the periods. The conjecture stated that 
\begin{equation}
\delta_{RG} \int_{q_i^+}^{q_j^-} \phi(z) \d z \propto \lambda
\end{equation}
where the constant of proportionality is universal and depends on the overall normalization of the Chern-Simons action.  Here we will verify that our contour integral expression for the one-loop RG flow satisfies this property.  (The conjecture also states that periods between two chiral, or two anti-chiral, zeroes does not flow; this will be easy to verify in our context).

At one loop, we have the contour integral expression
\begin{equation}
	 \delta_{RG}\phi(z) 
	= -2 \lambda  \partial_z \sum_j \oint_{\abs{w - q_j^+} = \eps} \frac{1}{2 \pi \i  (w - z)} \d w\frac{\phi(z)}{\what{\phi}(w) } 
	\end{equation}
Now let us consider the flow of the period of $\phi$, from $q_j^+$ to $q_i^-$.  The flow of $\phi$ is a sum of contributions from chiral zeroes $q_j^+$. In the flow of the period from $q_j^+$ to $q_i^-$, only the contribution of $\delta_{RG} \phi$ from $q_j^+$ will appear. 

Because $\delta_{RG} \phi$ is a total derivative, we only get contributions from the end points of the integral from $q_j^+$ to $q_i^-$, giving
\begin{equation}
\delta_{RG} \int_{q_j^+}^{q_i^-} \phi(z) \d z = - 2 \lambda \oint_{\abs{w - q_j^+} = \eps} \frac{1}{2 \pi \i  (w - q_i^-)} \d w\frac{\phi(q_i^-)}{{\phi}(w) } + 2 \lambda \lim_{z \to q_j^+} \oint_{\abs{w - q_j^+} = \eps} \frac{1}{2 \pi \i  (w - z)} \d w\frac{\phi(z)}{{\phi}(w) } 
\end{equation}
Since $\phi(q_i^-) = 0$, this is
\begin{equation}
\delta_{RG} \int_{q_j^+}^{q_i^-} \phi(z) \d z =  2 \lambda \lim_{z \to q_j^+} \oint_{\abs{w - q_j^+} = \eps} \frac{1}{2 \pi \i  (w - z)} \d w\frac{\phi(z)}{{\phi}(w) } 
\end{equation}
On the right hand side, the contour integral over $w$ encloses $q_j^+$ but \emph{does not} enclose $z$. This is why we have written the expression in a way where we first perform the integral over $w$ and then take the limit as $z \to q_j^+$.  If, instead, we chose a contour which enclosed both $z$ and $q_j^+$, we would get zero, because in the limit as $z$ approaches $q_j^+$, $\phi(z)$ becomes zero.

Therefore, the total variation of the period is given by the difference between these two contours, which is 
\begin{equation}
\begin{split}
\delta_{RG} \int_{q_j^+}^{q_i^-} \phi(z) \d z &= - 2 \lambda \lim_{z \to q_j^+} \oint_{\abs{w - z} = \eps} \frac{1}{2 \pi \i  (w - z)} \d w\frac{\phi(z)}{{\phi}(w) }\\
&= 2 \lambda
\end{split}
\end{equation}
which is exactly what we expected. 

It is easy to show that the periods from $q_j^+$ to $q_i^+$ and from $q_j^-$ to $q_i^-$ do not flow.

\subsection{The PCM}

As our first example, we consider the principal chiral model. Its one-form has double poles at the origin and at $\infty$, a chiral zero at $a$ and an anti-chiral zero at $-a$,
\begin{equation}
\phi=\frac{a^2-z^2}{z^2}.
\end{equation}
The chiral zero is $q^+=a$ and the anti-chiral zero $q^-=-a$.  Compactifying $4d$ Chern--Simons theory along $\CP^1$ with this one-form gives the PCM, which in the conventions of \cite{Levine:2021fof} reads
\begin{equation}
S=\frac{1}{4\pi\hbar}\int\Big[-\frac12 s\operatorname{Tr}\big[J_+J_-\big]\Big],\qquad J=g^{-1}\d g,
\end{equation}
with $s=-2a$, as we show in Appendix~\ref{app:GxG}. The two-loop RG flow of the PCM is \cite{Hoare:2019mcc, Levine:2021fof}
\begin{equation}
\delta_{RG}s=\lambda+\frac{\lambda^2}{2s}\,.
\end{equation}
There, the loop counting parameter $\alpha'$ and the dual Coxeter number $c_G$ only appear in the combination $\alpha'c_G$. Identifying $\alpha'$ with our $\hbar$, and using $c_G=N$ for $\mathfrak{gl}_N$, this combination is $\lambda=\hbar N$. Since $s=-2a$, we expect
\begin{equation}
\delta_{RG}a=-\tfrac12\lambda+\tfrac18\lambda^2a^{-1}.
\end{equation}
This will be exactly what we find.

It is interesting to note that in this example (and in fact all examples we study) the two-loop RG flow is entirely planar. This is because the coefficient of the two-loop flow is $c_G^2$ which, for $U(N)$, is $N^2$, and sub-planar terms would have smaller powers of $N$.

Running the algorithm to first order in $\lambda$, we find
\begin{equation}
\begin{split}
\nabla_+&=\lambda\partial+\phi(z)\d z+2\lambda\Big(\frac{1}{z}-\frac{1}{z+a}\Big)\d z,\\
\nabla_-&=-\lambda\partial+\phi(z)\d z+2\lambda\Big(\frac{1}{z-a}-\frac{1}{z}\Big)\d z,
\end{split}
\end{equation}
which are regular at the chiral zero $a$ and at the anti-chiral zero $-a$ respectively. The term $-2\lambda\partial_z\log(z+a)$ in $\nabla_+$ is needed for $\nabla_-$ to be regular at $-a$, while the terms $\pm2\lambda\partial_z\log z$ are a choice of scheme. The solution of the Bernoulli equation is then $\what{\alpha}=\what{\phi}(z)\d z$ with
\begin{equation}
\what{\phi}=\frac{a^2-z^2}{z^2}+\frac{\lambda}{z-a}-\frac{\lambda}{z+a}+\mathcal O(\lambda^2)=\frac{a^2-z^2}{z^2}-\frac{2a\lambda}{a^2-z^2}+\mathcal O(\lambda^2).
\end{equation}
Its order-$\lambda$ part has no pole at the origin, so the residue of $\what{\phi}$ there is that of $\phi$, and thus no level is generated. This is the effect of the scheme choice above.

As mentioned in the end of section~\ref{sec:computingbeta}, to compute $V_+$ we first expand $1/\what{\phi}$ in $\lambda$,
\begin{equation}
\frac{1}{\what{\phi}}=\frac{z^2}{a^2-z^2}+2\lambda\frac{az^4}{(a^2-z^2)^3}+\mathcal O(\lambda^2),
\end{equation}
and write each order as a partial fraction decomposition,
	\begin{equation}
	\begin{split} 
	 V_+ &= g(z) \partial_z\\ 
	g(z) &=  \frac{a}{2 (z+a)} - \frac{1}{2}  + \lambda  \frac{3}{8 (z + a)} -\lambda  \frac{ 5 a}{8 (z + a)^2} + \lambda \frac{a^2}{4 (z + a)^3}\\
	&= -  \frac{z}{2(z+a)} + \lambda  \frac{z(3z+a)}{8(z+a)^3} 
	\end{split} 
	\end{equation}
	We are free to add any constant to $g$, and we have chosen the constant $-\half$. 

	The RG flow is
	\begin{equation} 
	\delta_{RG} \what{\phi} = -2 \lambda \partial_z (g(z) \what{\phi}(z)) 
	\end{equation}
	and
	\begin{equation}
	g(z) \what{\phi}(z) =- \half \frac{a-z}{z }   -\lambda \left( \frac{1}{4 (z - a)} + \frac{ 1}{4 (a + z)} - \frac{1}{8 z} \right) 
 \end{equation}
Since
\begin{equation}
\delta_{RG} \what{\phi} = \delta_{RG} a \left( \dpa{a} \what{\phi}  \right)\,,
\end{equation}

the equation is then
\begin{equation}
\begin{split} 
( \delta_{RG} a)  \dpa{a} \left(  \frac{a^2-z^2}{z^2}  +  \lambda   (z-a)^{-1} - \lambda  (z+a)^{-1} \right)&= -2 \lambda \partial_z  g(z) \what{\phi}(z) \\
&= \lambda \partial_z \left(  \frac{a-z}{z }   +\lambda  \frac{1}{2 (z - a)} + \lambda \frac{1}{2 (a + z)} - \lambda \frac{1}{4 z}  \right)\\ 
&=- \lambda \frac{a}{z^2} - \lambda^2\frac{1}{2}\left( (z-a)^{-2} + (z+a)^{-2} \right)   +\lambda^2 \frac{1}{4 z^2} 
\end{split}  
\end{equation}
Implementing the derivative in $a$ on the left hand side, we see that
\begin{equation}
\delta_{RG} a \left(  \frac{2 a}{z^2} + \lambda (z-a)^{-2} + \lambda (z+a)^{-2}    \right) =
 -\lambda \frac{a}{z^2} - \lambda^2\frac{1}{2}\left( (z-a)^{-2} + (z+a)^{-2} \right)   +\lambda^2 \frac{1}{4 z^2} 
\end{equation}
from which we conclude that
\begin{equation} 
\delta_{RG}a = -\half \lambda  + \lambda^2 \tfrac{1}{8} a^{-1}   
\end{equation}
exactly as expected.

As a vector field on the space of coupling constants, the RG flow is $-\tfrac12\lambda\partial_a+\tfrac18\lambda^2a^{-1}\partial_a$. This is scheme independent: any change of coordinates $a\mapsto f(a)$ which would absorb the two-loop term would involve $\log a$. In perturbative QFT, $a^{-1}$ is the coupling and one only performs algebraic changes of coordinates.

\subsection{The PCM with Wess-Zumino term}
	Next we add a Wess--Zumino term. We keep the zeroes $q^\pm=\pm a$ and the double pole at $\infty$, and move the other double pole to $b$, so that
	\begin{equation}
	\phi = \frac{a^2-z^2}{(z-b)^2}\,.
	\end{equation}
	The residue of $\phi$ at $z=b$ is $-2b$. Compactifying $4d$ Chern--Simons theory along $\CP^1$ with this one-form gives the PCM with WZ term, which in the conventions of \cite{Levine:2021fof} reads
\begin{equation}
S=\frac{1}{4\pi\hbar}\int\Big[-\frac12 s\operatorname{Tr}\big[J_+J_-\big]+k\mathcal L_{\rm WZ}(g)\Big],
\end{equation}
with $s=-2a$ and $k=2b$, as we show in Appendix~\ref{app:GxG}. The level is thus minus the residue of $\phi$ at $b$. The two-loop RG flow of the PCM with WZ term is \cite{Hoare:2019mcc,Levine:2021fof}
\begin{equation}
\delta_{RG}s=\lambda\Big(1-\frac{k^2}{s^2}\Big)\Big[1+\frac{\lambda}{2s}\Big(1-\frac{3k^2}{s^2}\Big)\Big],\qquad \delta_{RG}k=0.
\end{equation}
Since $s=-2a$ and $k=2b$, we expect
\begin{equation}
\begin{split}
\delta_{RG}a&=-\tfrac12\lambda\Big(1-\frac{b^2}{a^2}\Big)+\tfrac18\lambda^2a^{-1}\Big(1-\frac{b^2}{a^2}\Big)\Big(1-\frac{3b^2}{a^2}\Big),\\
\delta_{RG}b&=0.
\end{split}
\end{equation}
This will be exactly what we find.

	 Running the algorithm to first
	order in $\lambda$ we find
	\begin{equation}
	\begin{split}
	\nabla_+&=\lambda\partial+\phi(z)\,\d z+2\lambda\Big(\frac{1}{z-b}-\frac{1}{z+a}\Big)\d z\,,\\
	\nabla_-&=-\lambda\partial+\phi(z)\,\d z+2\lambda\Big(\frac{1}{z-a}-\frac{1}{z-b}\Big)\d z\,,
	\end{split}
	\end{equation}
	which are regular at $q^+=a$ and at $q^-=-a$ respectively. As in the PCM, the term $-2\lambda\,\partial_z\log(z+a)$
	is needed for $\nabla_-$ to be regular at $-a$, and the term $2\lambda\,\partial_z\log(z-b)$ is the scheme
	choice that keeps the residue of $\what{\phi}$ at the double pole $b$ equal to that of $\phi$, so that the level is
	not shifted. The solution of the Bernoulli equation is then $\what{\alpha}=\what{\phi}\,\d z$, and we have
	\begin{equation}
	\what{\phi} = \frac{a^2-z^2}{(z-b)^2} + \lambda (z-a)^{-1} - \lambda (z+a)^{-1} = \frac{a^2-z^2}{(z-b)^2} - \lambda\frac{2a}{a^2-z^2}
	\end{equation}
	and
	\begin{equation}
	\frac{1}{\what{\phi}(z)} = \frac{(z-b)^2}{a^2-z^2} + 2\lambda\frac{a\,(z-b)^4}{(a^2-z^2)^3}
	\end{equation}
	(working modulo $\lambda^2$). As in the PCM, $V_+$ is the sum of the terms of the partial fractions decomposition
	of $1/\what{\phi}$ with poles at the anti-chiral zero $-a$, up to a constant. The partial fractions decompositions are
	\begin{equation}
	\begin{split}
	\frac{(z-b)^2}{a^2-z^2} &= \frac{(a+b)^2}{2a(z+a)} - \frac{(a-b)^2}{2a(z-a)} - 1\,,\\
	2\frac{a\,(z-b)^4}{(a^2-z^2)^3} &= \frac{(a+b)^4}{4a^2(z+a)^3} - \frac{(a+b)^3(5a-3b)}{8a^3(z+a)^2} + \frac{3(a^2-b^2)^2}{8a^4(z+a)}\\
	&\quad - \frac{(a-b)^4}{4a^2(z-a)^3} - \frac{(a-b)^3(5a+3b)}{8a^3(z-a)^2} - \frac{3(a^2-b^2)^2}{8a^4(z-a)}\,,
	\end{split}
	\end{equation}
	from which we conclude that $V_+=g(z)\partial_z$ with
	\begin{equation}
	g(z) = \frac{(a+b)^2}{2a(z+a)} + \gamma_0 + \lambda\left[\frac{(a+b)^4}{4a^2(z+a)^3} - \frac{(a+b)^3(5a-3b)}{8a^3(z+a)^2} + \frac{3(a^2-b^2)^2}{8a^4(z+a)} + \gamma_1\right]\,,
	\end{equation}
	where, as in the PCM, we have added a constant $\gamma_0+\lambda\gamma_1$, which absorbs the constant $-1$ of the decomposition. The RG flow $\delta_{RG}\what{\alpha}=-2\lambda\mc{L}_{V_+}\what{\alpha}$ becomes
\begin{equation}
\delta_{RG}\what{\phi}=-2\lambda\partial_z\big(g\what{\phi}\big), \label{eqn_pcmwzwflow}
\end{equation}
We of course have 
\begin{equation}
\delta_{RG} \what{\phi} = (\delta_{RG} a) \dpa{a} \what{\phi} + (\delta_{RG} b) \dpa{b} \what{\phi}.
\end{equation}
Rather than identifying all terms in the rational functions on both sides of \eqref{eqn_pcmwzwflow}, as we did in the case of the PCM without WZW term, we will instead expand both sides in series in $(z-b)^{-1}$.  Matching the polar terms will be sufficient to determine the RG flows $\delta_{RG} a$ and $\delta_{RG} b$.

 Expanding in series in $(z-b)^{-1}$, with $O(1)$ denoting terms regular at $z=b$, we have
	\begin{equation}
	\label{ec:ghatphipcmwz}
	-2\lambda\,\partial_z\big(g\what{\phi}\big)=\frac{4\lambda(a^2-b^2)\,g(b)}{(z-b)^3}+\frac{2\lambda\big[(a^2-b^2)\,g'(b)-2b\,g(b)\big]}{(z-b)^2}+O(1)\,,
	\end{equation}
	while
	\begin{equation}
	\label{ec:deltaphipcmwz}
	\delta_{RG} a\,\partial_a\what{\phi}+\delta_{RG} b\,\partial_b\what{\phi}=\frac{2(a^2-b^2)\,\delta_{RG} b}{(z-b)^3}+\frac{2a\,\delta_{RG} a-4b\,\delta_{RG} b}{(z-b)^2}-\frac{2\,\delta_{RG} b}{z-b}+O(1)\,.
	\end{equation}
	Comparing \eqref{ec:ghatphipcmwz} and \eqref{ec:deltaphipcmwz} term by term, the simple pole gives $\delta_{RG} b=0$. Matching the triple pole term gives $g(b)=0$, which leads to the choice for $\gamma_0$, $\gamma_1$
	\begin{equation}
	\gamma_0 = -\frac{a+b}{2a}\,,\qquad \gamma_1 = \frac{3b\,(a^2-b^2)}{8a^4}\,.
	\end{equation}
    
     The double pole term gives
	\begin{equation}
	\delta_{RG} a=\frac{\lambda\,(a^2-b^2)}{a}\,g'(b)= -\frac{1}{2}\lambda\Big(1-\frac{b^2}{a^2}\Big) + \lambda^2\,\frac{(a^2-b^2)(a^2-3b^2)}{8a^5}\,,
	\end{equation}
	exactly as expected. For $b=0$ it reduces to the PCM, and both terms vanish at the WZW points $b=\pm a$.

\subsection{The $G\times G$ integrable model}

\label{sec:GxG}

Finally, we consider the general integrable $G\times G$ model \cite{Delduc:2018hty,Delduc:2020vxy} with one-form
\begin{equation}
\phi=-\frac{(z-\zeta_1^+)(z-\zeta_2^+)(z-\zeta_1^-)(z-\zeta_2^-)}{z^2(z-z_2)^2},
\end{equation}
where without loss of generality we took $z_1=0$. It has two chiral zeroes $\zeta_1^+,\zeta_2^+$ and two anti-chiral zeroes $\zeta_1^-,\zeta_2^-$. The residues of $\phi$ at $0$ and $z_2$ are $-k_1$ and $-k_2$, with
\begin{equation}
\begin{aligned}
k_1&=\frac{2\zeta_1^+\zeta_2^+\zeta_1^-\zeta_2^-}{z_2^3}-\frac{\zeta_1^+\zeta_2^+(\zeta_1^-+\zeta_2^-)+\zeta_1^-\zeta_2^-(\zeta_1^++\zeta_2^+)}{z_2^2},\\
k_2&=\frac{(z_2-\zeta_1^+)(z_2-\zeta_2^+)(2z_2-\zeta_1^--\zeta_2^-)+(z_2-\zeta_1^-)(z_2-\zeta_2^-)(2z_2-\zeta_1^+-\zeta_2^+)}{z_2^2}\\
&\quad-\frac{2(z_2-\zeta_1^+)(z_2-\zeta_2^+)(z_2-\zeta_1^-)(z_2-\zeta_2^-)}{z_2^3}.
\end{aligned}
\end{equation}
Compactifying $4d$ Chern--Simons theory along $\mathbb{CP}^1$ with this one-form gives the action \cite{Costello:2019tri,Levine:2021fof}
\begin{equation}
S=\frac{1}{4\pi\hbar}\int\bigg[-\frac12\begin{pmatrix}s&t+b\\ t-b&u\end{pmatrix}_{\!ij}\operatorname{Tr}\big[J_+^{(i)}J_-^{(j)}\big]+k_i\mathcal L_{\rm WZ}\big(g^{(i)}\big)\bigg],
\end{equation}
with $J^{(i)}=g^{(i)-1}\mathrm{d}g^{(i)}$ for $i=1,2$. It describes two PCMs with WZ terms, coupled through their currents. We obtain the couplings of this action in terms of the zeroes and poles of $\phi$ by matching it with the action of \cite{Delduc:2020vxy}. This matching is done carefully in Appendix~\ref{app:GxG}, and gives
\begin{equation}
\label{ec:sutbintermsofmoduli}
\begin{aligned}
s&=\frac{\zeta_1^-\zeta_2^-(\zeta_1^++\zeta_2^+)-\zeta_1^+\zeta_2^+(\zeta_1^-+\zeta_2^-)}{z_2^2},\\
u&=\frac{(z_2-\zeta_1^+)(z_2-\zeta_2^+)(2z_2-\zeta_1^--\zeta_2^-)-(z_2-\zeta_1^-)(z_2-\zeta_2^-)(2z_2-\zeta_1^+-\zeta_2^+)}{z_2^2},\\
t&=\frac{\zeta_1^-\zeta_2^-(z_2-\zeta_1^+)(z_2-\zeta_2^+)-\zeta_1^+\zeta_2^+(z_2-\zeta_1^-)(z_2-\zeta_2^-)}{z_2^3},\\
b&=-\frac{\zeta_1^+\zeta_2^+(z_2-\zeta_1^-)(z_2-\zeta_2^-)+\zeta_1^-\zeta_2^-(z_2-\zeta_1^+)(z_2-\zeta_2^+)}{z_2^3}.
\end{aligned}
\end{equation}
The two-loop RG flow of this model was computed in \cite[eq.~(B.9)]{Levine:2021fof}. With the same identification $\alpha'c_G=\hbar N=\lambda$ as for the PCM, it reads $\delta_{RG}k_1=\delta_{RG}k_2=0$ and
\begin{equation}
\label{ec:rgflowGxG}
\delta_{RG}\begin{pmatrix}s&t+b\\ t-b&u\end{pmatrix}_{\!ij}=\lambda\frac{F^{(4)}_{ij}}{(su-t^2)^2}+\lambda^2\frac{F^{(9)}_{ij}}{(su-t^2)^5},
\end{equation}
where $F^{(4)}_{ij}$ and $F^{(9)}_{ij}$ are homogeneous polynomials of degrees $4$ and $9$ in $(s,t,u,b,k_1,k_2)$, given in Appendix~\ref{app:Fpolynomials}. We now reproduce this flow from our holographic algorithm. The computation is carried out step by step in a Mathematica notebook, included as a supplementary file with this paper. It is organized in five blocks that follow the steps of this section, so that the reader can reproduce each step, including the expressions that are too long to display here, and verify the agreement with \cite{Levine:2021fof}.

The chiral zeroes are $\zeta_1^+,\zeta_2^+$ and the anti-chiral zeroes are $\zeta_1^-,\zeta_2^-$. Running the algorithm to first order in $\lambda$, we find
\begin{equation}
\begin{split}
\nabla_+&=\lambda\partial+\phi(z)\d z+2\lambda\Big(\frac{1}{z}+\frac{1}{z-z_2}-\frac{1}{z-\zeta_1^-}-\frac{1}{z-\zeta_2^-}\Big)\d z,\\
\nabla_-&=-\lambda\partial+\phi(z)\d z+2\lambda\Big(\frac{1}{z-\zeta_1^+}+\frac{1}{z-\zeta_2^+}-\frac{1}{z}-\frac{1}{z-z_2}\Big)\d z,
\end{split}
\end{equation}
which are regular at the chiral and at the anti-chiral zeroes respectively. The solution of the Bernoulli equation is then $\what{\alpha}=\what{\phi}(z)\d z$ with
\begin{equation}
\what{\phi}=\phi+\lambda\Big(\frac{1}{z-\zeta_1^+}+\frac{1}{z-\zeta_2^+}-\frac{1}{z-\zeta_1^-}-\frac{1}{z-\zeta_2^-}\Big)+\mathcal O(\lambda^2).
\end{equation}
Its order-$\lambda$ part has no poles at $0$ and $z_2$, so the residues of $\what{\phi}$ there are those of $\phi$, and no level is generated. This is a consequence of the terms $\pm2\lambda\partial_z\log z$ and $\pm2\lambda\partial_z\log(z-z_2)$ in $\nabla_\pm$, which are a choice of scheme.

As in the previous examples, we write $V_+=g(z)\partial_z$ with
\begin{equation}
g(z)=\sum_{i=1,2}\frac{A_i}{z-\zeta_i^-}+\gamma_0+\lambda\bigg[\sum_{i=1,2}\sum_{n=1}^{3}\frac{C_{i,n}}{(z-\zeta_i^-)^n}+\gamma_1\bigg]+\mathcal O(\lambda^2),
\end{equation}
where the coefficients $A_i$ and $C_{i,n}$ are computed in Block~2 of the supplementary notebook, and $\gamma_0$ and $\gamma_1$ are fixed below.

The RG flow $\delta_{RG}\what{\alpha}=-2\lambda\mc{L}_{V_+}\what{\alpha}$ becomes
\begin{equation}
\label{ec:rglowequat}
\delta_{RG}\what{\phi}=-2\lambda\partial_z\big(g\what{\phi}\big),
\end{equation}
where the left-hand side now comes from the flow of the five moduli $\zeta_1^\pm$, $\zeta_2^\pm$ and $z_2$,
\begin{equation}
\delta_{RG}\what{\phi}=\frac{\partial\what{\phi}}{\partial\zeta_1^+}\delta_{RG}\zeta_1^++\frac{\partial\what{\phi}}{\partial\zeta_2^+}\delta_{RG}\zeta_2^++\frac{\partial\what{\phi}}{\partial\zeta_1^-}\delta_{RG}\zeta_1^-+\frac{\partial\what{\phi}}{\partial\zeta_2^-}\delta_{RG}\zeta_2^-+\frac{\partial\what{\phi}}{\partial z_2}\delta_{RG}z_2.
\end{equation}
The variations of the moduli determine the $\beta$-functions of $s$, $t$, $u$ and $b$ through the chain rule, using \eqref{ec:sutbintermsofmoduli}. Equation \eqref{ec:rglowequat} is therefore a system of equations for $\delta_{RG}\zeta_1^\pm$, $\delta_{RG}\zeta_2^\pm$ and $\delta_{RG}z_2$, which we now solve. The first step is to expand both sides near the double poles $p\in\{0,z_2\}$. Since the order-$\lambda$ part of $\what{\phi}$ is regular there, the singular terms of $\what{\phi}$ near $p$ are those of $\phi$,
\begin{equation}
\what{\phi}=\frac{B_p}{(z-p)^2}+\frac{R_p}{z-p}+\mathcal O(1),
\end{equation}
with
\begin{equation}
\begin{aligned}
&B_0=-\frac{\zeta_1^+\zeta_2^+\zeta_1^-\zeta_2^-}{z_2^2},\qquad R_0=-k_1,\\
&B_{z_2}=-\frac{(z_2-\zeta_1^+)(z_2-\zeta_2^+)(z_2-\zeta_1^-)(z_2-\zeta_2^-)}{z_2^2},\qquad R_{z_2}=-k_2.
\end{aligned}
\end{equation}
Since $V_+$ is regular on the chiral patch which includes $p$, $g$ has a Taylor expansion $g=g(p)+g'(p)(z-p)+\mathcal O\big((z-p)^2\big)$. Multiplying the two,
\begin{equation}
\label{ec:generalcasegphi}
g\what{\phi}=\frac{g(p)B_p}{(z-p)^2}+\frac{g'(p)B_p+g(p)R_p}{z-p}+\mathcal O(1),
\end{equation}
so that near $p$ the right-hand side of \eqref{ec:rglowequat} is
\begin{equation}
-2\lambda\partial_z\big(g\what{\phi}\big)=\frac{4\lambda g(p)B_p}{(z-p)^3}+\frac{2\lambda\big(g'(p)B_p+g(p)R_p\big)}{(z-p)^2}+\mathcal O(1).
\end{equation}
On the left-hand side, $B_p$ and $R_p$ are functions of the moduli. Varying the moduli changes $B_p$ and $R_p$ and, for $p=z_2$, also moves the pole itself, while the pole at $0$ stays fixed. This gives
\begin{equation}
\delta_{RG}\what{\phi}=\frac{2B_p\delta_{RG}p}{(z-p)^3}+\frac{R_p\delta_{RG}p+\delta_{RG}B_p}{(z-p)^2}+\frac{\delta_{RG}R_p}{z-p}+\mathcal O(1).
\end{equation}
Comparing the coefficients of $(z-p)^{-3}$, $(z-p)^{-2}$ and $(z-p)^{-1}$ on the two sides gives
\begin{equation}
\delta_{RG}p=2\lambda g(p),\qquad \delta_{RG}B_p=2\lambda g'(p)B_p,\qquad \delta_{RG}R_p=0.
\end{equation}
Since the pole at $0$ is kept fixed, the first condition at $p=0$ requires $g(0)=0$, which fixes $\gamma_0$ and $\gamma_1$. The remaining five conditions are
\begin{equation}
\label{ec:deltamoduli}
\begin{aligned}
&\delta_{RG}B_0=2\lambda g'(0)B_0,\qquad \delta_{RG}R_0=0,\\
&\delta_{RG}z_2=2\lambda g(z_2),\qquad \delta_{RG}B_{z_2}=2\lambda g'(z_2)B_{z_2},\qquad \delta_{RG}R_{z_2}=0.
\end{aligned}
\end{equation}
Since $R_0=-k_1$ and $R_{z_2}=-k_2$, the two conditions $\delta_{RG}R_p=0$ say that the levels do not flow, as in \cite{Levine:2021fof}. The flow of the other couplings follows from the chain rule,
\begin{equation}
\label{ec:chainrulefors}
\delta_{RG}s=\frac{\partial s}{\partial\zeta_1^+}\delta_{RG}\zeta_1^++\frac{\partial s}{\partial\zeta_2^+}\delta_{RG}\zeta_2^++\frac{\partial s}{\partial\zeta_1^-}\delta_{RG}\zeta_1^-+\frac{\partial s}{\partial\zeta_2^-}\delta_{RG}\zeta_2^-+\frac{\partial s}{\partial z_2}\delta_{RG}z_2,
\end{equation}
and similarly for $t$, $u$ and $b$. The variations of the moduli are determined by \eqref{ec:deltamoduli}. Since $B_p$ and $R_p$ are functions of the moduli, their variations also follow from the chain rule, for instance
\begin{equation}
\delta_{RG}B_0=\frac{\partial B_0}{\partial\zeta_1^+}\delta_{RG}\zeta_1^++\frac{\partial B_0}{\partial\zeta_2^+}\delta_{RG}\zeta_2^++\frac{\partial B_0}{\partial\zeta_1^-}\delta_{RG}\zeta_1^-+\frac{\partial B_0}{\partial\zeta_2^-}\delta_{RG}\zeta_2^-+\frac{\partial B_0}{\partial z_2}\delta_{RG}z_2,
\end{equation}
so that the five conditions \eqref{ec:deltamoduli} become the linear system
\begin{equation}
\begin{pmatrix}
\frac{\partial B_0}{\partial\zeta_1^+}&\frac{\partial B_0}{\partial\zeta_2^+}&\frac{\partial B_0}{\partial\zeta_1^-}&\frac{\partial B_0}{\partial\zeta_2^-}&\frac{\partial B_0}{\partial z_2}\\[2pt]
\frac{\partial R_0}{\partial\zeta_1^+}&\frac{\partial R_0}{\partial\zeta_2^+}&\frac{\partial R_0}{\partial\zeta_1^-}&\frac{\partial R_0}{\partial\zeta_2^-}&\frac{\partial R_0}{\partial z_2}\\[2pt]
0&0&0&0&1\\[2pt]
\frac{\partial B_{z_2}}{\partial\zeta_1^+}&\frac{\partial B_{z_2}}{\partial\zeta_2^+}&\frac{\partial B_{z_2}}{\partial\zeta_1^-}&\frac{\partial B_{z_2}}{\partial\zeta_2^-}&\frac{\partial B_{z_2}}{\partial z_2}\\[2pt]
\frac{\partial R_{z_2}}{\partial\zeta_1^+}&\frac{\partial R_{z_2}}{\partial\zeta_2^+}&\frac{\partial R_{z_2}}{\partial\zeta_1^-}&\frac{\partial R_{z_2}}{\partial\zeta_2^-}&\frac{\partial R_{z_2}}{\partial z_2}
\end{pmatrix}
\begin{pmatrix}\delta_{RG}\zeta_1^+\\ \delta_{RG}\zeta_2^+\\ \delta_{RG}\zeta_1^-\\ \delta_{RG}\zeta_2^-\\ \delta_{RG}z_2\end{pmatrix}
=\begin{pmatrix}\delta_{RG}B_0\\ \delta_{RG}R_0\\ \delta_{RG}z_2\\ \delta_{RG}B_{z_2}\\ \delta_{RG}R_{z_2}\end{pmatrix}
=2\lambda\begin{pmatrix}g'(0)B_0\\ 0\\ g(z_2)\\ g'(z_2)B_{z_2}\\ 0\end{pmatrix}.
\end{equation}
The matrix does not depend on $\lambda$, and it is invertible as long as the four zeroes are distinct. This allows us to solve for $\delta_{RG}\zeta_1^\pm$, $\delta_{RG}\zeta_2^\pm$ and $\delta_{RG}z_2$, which we do in Block~3 of the supplementary notebook. 

Inserting the solution into equation \eqref{ec:chainrulefors} gives $\delta_{RG}s$, $\delta_{RG}t$, $\delta_{RG}u$ and $\delta_{RG}b$ as functions of the moduli (Block~4 of the supplementary notebook). Their terms of order $\lambda$ and $\lambda^2$ give the one- and two-loop flows. On the other hand, using \eqref{ec:sutbintermsofmoduli}, the right-hand side of \eqref{ec:rgflowGxG} is also a function of the moduli. Comparing the two in Block~5 of the supplementary notebook, we find that they agree identically, at one and two loops. Our algorithm therefore reproduces the two-loop RG flow of \cite{Levine:2021fof}, in their scheme, without any redefinition of the couplings.

\subsection{Comments on scheme choices}
It is very surprising that we find exactly the RG flow as computed in \cite{Levine:2021fof} without any need for a change of schemes.

Let us briefly explain why the ambiguity in the choice of scheme in the holographic context and the field theory is the same.  From the holographic perspective, a scheme is a choice of ways of turning a classical spectral curve into a planar spectral curve (for the two-loop RG flow, we only need the planar spectral curve modulo $\lambda^2$). In the case of the $G \times G$ model, the one form has $3$ second order poles, $2$ chiral zeroes and $2$ anti-chiral zeroes.   Above, we took the $3$ second order poles to be at $0$, $z_2$ and $\infty$, but after a change of coordinates we can place them at  $0,1,\infty$.  Then, the model has $5$ parameters.  Two of the parameters are the residues of the one-form at the three poles (which of course sum to zero).  The remaining three are the periods of the one-form from one of the zeroes to the other three.

We choose a scheme so that the residues of $\what{\alpha}$ at the three second order poles do not depend on the 't Hooft coupling $\lambda$.  This reduces the number of parameters to $3$.  This choice  is very natural from the field theory perspective: the residue of $\what{\alpha}$ is the effective Chern-Simons level associated to the $U(N)$ symmetry at that pole. This should be the same in the quantum as in the classical theory. (In theories which have different numbers of chiral and anti-chiral defects, anomaly inflow will modify the sum of the residues of $\what{\alpha}$ at the second order poles). 

With this  choice, there is still some ambiguity in the choice of planar spectral curve quantizing the classical spectral curve. At order $\lambda$, this ambiguity is given by the freedom to add to $\nabla_+$, $\nabla_-$, and $\what{\alpha}$, a one form $\lambda \beta$, where $\beta$ has second order poles at $0$, $1$, $\infty$ with vanishing residue.  There is a $3$ dimensional space of such one-forms $\beta$, and we can view the choice of such a $\beta$ as a one-loop counter-term.   We have fixed this choice in such a way that $\what{\alpha}$ has the simplest possible expression.

In the integrable field theory, there is also a $3$-dimensional space of one-loop counter-terms, which do not modify the Wess-Zumino terms. (The coefficients of the Wess-Zumino terms are the residues of the one-form at the poles).   These counter-terms modify the coefficients of the $\op{Tr} J_+^{(i)}J_-^{(j)}$ terms in the action without changing the WZ terms.  Of the $4$ possible counter-terms which are current bilinears, only $3$ appear as integrable deformations which do not modify the WZ terms.    (This is because the $s,t,u,b$ which are the coefficients of the current bilinears must satisfy a polynomial equation, equation 2.3 of \cite{Levine:2021fof}, in order to define an integrable model). 

We therefore see that the three possible one-loop counter-terms in the holographic set-up match the three one-loop integrable counter-terms in the field theory set-up.  We have chosen a way to fix the one-loop holographic counter-terms, by asking that $\what{\alpha}$ takes the simplest possible form. This corresponds to a way to fix the field theory counter-terms.  Miraculously, we seem to have landed on exactly the scheme used in \cite{Levine:2021fof}.
\subsection{RG flow for order defects}
We have seen  that a chiral order defect at a point $\Sigma_+$ introduces a pole in the the $\lambda$-connection $\nabla_+$ so that the monodromy of the connection $\lambda^{-1} \nabla_+$ around the defect is $e^{-2 \pi \i N_f / N}$.  Similarly, for an anti-chiral order defect at  a point $\Sigma_-$, the monodromy of the connection $-\lambda^{-1}\nabla_-$ is $e^{-2\pi \br{N}_f / N}$.  

Let us consider the case when $\Sigma = \CP^1$ with one-form $\d z$ and chiral order defects at $z_i^+$ and anti-chiral order defects at $z_i^-$.  Let us take $N_f = 1$ at each defect; we  this the most general case, as models with $N_f > 1$ arise when several chiral (or antichiral) defects collide. 

This corresponds to an integrable field theory which, for each $z_i^+$ has $N$ chiral fermions $\psi^{(i)}$ transforming in the fundamental representation of $GL(N)$, and $N$ fermions $\eta^{(i)}$ in the  anti-fundamental representation of $GL(N)$.  Similarly, for each $z_i^-$ there are anti-chiral fermions the fundamental and the anti-fundamental. The Lagrangian is  
\begin{equation} 
\sum \int \psi^{(i)} \dbar \eta^{(i)} + \int \br{\psi}^{(j)} \partial \br{\eta}^{(j)} + \sum \frac{1}{z_i^+ - z_j^-} \int \psi^{(i)} \eta^{(i)} \br{\psi}^{(j)} \br{\eta}^{(j)}  
\end{equation} 
Our goal is to calculate the planar two-loop RG flow of this model at two loops.  The space of coupling constants is the space of the $z_i^{\pm}$. The RG flow is a vector field on this space.   We will find that the RG flow is
\begin{equation} 
\begin{split} 
 \lambda \sum \partial_{z_i^-} - \lambda \sum \partial_{z_j^+} + 2 \lambda^2 N^{-1} \sum_{i,j} \frac{1}{z_i^- - z_j^+}(\partial_{z_i^-} - \partial_{z_j^+} )    + O(\lambda^3) 
\end{split} 
\end{equation} 
Recalling that $\lambda = \hbar N$, we can rewrite this in a way which makes the topology of the planar diagrams more transparent.  Let us have $N_f^{(i)}$ flavours at the chiral defects and $\br{N}_f^{(i)}$ flavours at the antichiral defects, which can be achieved by having some defects collide. Then, the RG flow is
\begin{equation} 
\begin{split} 
\sum \hbar N \partial_{z_i^-} - \sum \hbar N \partial_{z_j^+} + \sum_{i,j} 2 \hbar^2 N  \frac{1}{z_i^- - z_j^+}( N_f^{(j)} \partial_{z_i^-} - \br{N}_f^{(i)} \partial_{z_j^+} )    + O(\lambda^3) 
\end{split} 
\end{equation}

We will start by building a planar spectral  curve for this model.  This model behaves in a different way from the model with order defects, as effectively the one-form $\alpha$ defining the theory already has some $\lambda$ dependence: 
\begin{equation} 
\alpha =  \d z  - \sum N^{-1} \lambda \frac{\d z}{z - z_i^+} + \sum N^{-1} \lambda \frac{\d z}{z - z_i^-} 
\end{equation}
Working modulo $\lambda^2$, we have a planar spectral curve defined by
\begin{equation} 
\begin{split} 
\nabla_+ &= \lambda \partial + \alpha\\
\what{\alpha}&= \alpha  \\ 
\nabla_- &= - \lambda \partial + \alpha. 
\end{split} 
\end{equation}
 Here the connections $\nabla_{\pm}$ are written in the frame $\d z$.  Setting $\what{\alpha} = \alpha$ satsfies the Bernoulli equation, modulo $\lambda^2$ (the point is that $\lambda \partial_z \alpha$ is of order $\lambda^2$ so that $\nabla_+ \alpha = \alpha^2$).  

Let us extend this solution to the next order in $\lambda$, by setting
\begin{equation} 
 \nabla_+ = \lambda \partial + \alpha + \lambda^2 \alpha_2^+
\end{equation}
We will determine $\alpha_2^+$ by solving the Bernoulli equation to find $\what{\alpha}$, and then requiring that $-\nabla_+ + 2 \what{\alpha}$ has only first-order poles at $z_i^-$. 

The solution to the Bernoulli equation is
\begin{equation}
\begin{split} 
 \what{\alpha} &= \alpha + \lambda \partial_z \alpha + \lambda^2 \alpha_2^+ \\
&=  \d z  - \sum \lambda N^{-1} \frac{\d z}{z - z_i^+} + \sum \lambda N^{-1} \frac{\d z}{z - z_i^-} +  \sum \lambda^2 N^{-1} \frac{\d z}{(z - z_i^+)^2} - \sum \lambda^2  N^{-1} \frac{\d z}{(z - z_i^-)^2} + \lambda^2 \alpha_2^+ + O(\lambda^3)
\end{split} 
\end{equation}
If we set
\begin{equation} 
\nabla_- = -\nabla_+ + 2 \what{\alpha} 
\end{equation}
then 
\begin{equation} 
\nabla_- = -\lambda \partial + \d z  - \sum   N^{-1} \lambda \frac{\d z}{z - z_i^+} + \sum \lambda  N^{-1} \frac{\d z}{z - z_i^-} + 2 \sum \lambda^2  N^{-1} \frac{\d z}{(z - z_i^+)^2} - 2 \sum \lambda^2  N^{-1}   \frac{\d z}{(z - z_i^-)^2} + \lambda^2 \alpha_2^+ 
\end{equation}
Since $\nabla_-$ is only allowed to have first order poles at $z_i^-$, we set
\begin{equation} 
\alpha_2^+ = 2 \sum    N^{-1} \frac{\d z}{(z - z_i^-)^2}.  
\end{equation}
In which case,
\begin{equation} 
\what{\alpha} =  \d z  - \sum \lambda   N^{-1}  \frac{\d z}{z - z_i^+} + \sum \lambda     N^{-1}  \frac{\d z}{z - z_i^-} +  \sum \lambda^2     N^{-1}  \frac{\d z}{(z - z_i^+)^2} + \sum \lambda^2     N^{-1}  \frac{\d z}{(z - z_i^-)^2} + O(\lambda^3) 
\end{equation}
The vector field $V$ is, again to quadratic order in $\lambda$, 
\begin{equation}
\begin{split} 
 V =& \partial_z +    \sum_i \lambda     N^{-1}  \frac{1}{z - z_i^+}\partial_z - \sum_i \lambda     N^{-1} \frac{1}{z - z_i^-}\partial_z  \\
&+ \sum_{i,j} \lambda^2     N^{-2} \frac{1}{(z-z_i^+)(z-z_j^+)} + \sum_{i,j} \lambda^2    N^{-2}  \frac{1}{(z-z_i^-)(z-z_j^-)}\\
&   - 2 \sum_{i,j} \lambda^2    N^{-2} \frac{1}{(z - z_j^+)(z - z_i^-) }\partial_z \\
 & - \sum_i \lambda^2     N^{-1} \frac{1}{(z - z_i^+)^2}\partial_z - \sum_i     N^{-1} \lambda^2 \frac{1}{(z - z_i^-)^2}\partial_z\\ 
=& \partial_z +     \sum_i \lambda     N^{-1} \frac{1}{z - z_i^+}\partial_z - \sum_i \lambda    N^{-1}  \frac{1}{z - z_i^-}\partial_z \\
  & - \sum_i \lambda^2     N^{-1} \frac{1}{(z - z_i^+)^2}\partial_z - \sum_i \lambda^2     N^{-1} \frac{1}{(z - z_i^-)^2}\partial_z \\ 
&+ \sum_{i,j} \lambda^2   N^{-2}    \frac{1}{(z-z_i^+)(z-z_j^+)} + \sum_{i,j}     N^{-2} \lambda^2 \frac{1}{(z-z_i^-)(z-z_j^-)}\\
&+  2 \sum_{i,j}    N^{-2}  \lambda^2\frac{1}{z_j^+ - z_i^-} \left(  \frac{1}{z - z_j^+} - \frac{1}{z - z_i^-} \right)\partial_z 
\end{split} 
\end{equation} 

We will decompose $V$ as $V = V_+ + V_-$ where $V_+$ extends over the chiral patch and $V_-$ over the anti-chiral patch.
 
As before, the flow is obtained by appling the Lie derivative of $V_+$ to $(\nabla_+, \nabla_-,\what{\alpha})$. Note that any one of these determines the others. It is easiest to apply the Lie derivative to $\what{\alpha}$, since it transforms as a form and $\nabla_{\pm}$ are connections. 

We have  
\begin{equation}
\begin{split} 
V_+  =&   - \sum  N^{-1} \lambda \frac{1}{z - z_i^-}\partial_z  + \sum \lambda^2    N^{-2}   \frac{1}{(z-z_i^-)(z-z_j^-)}\partial_z\\
 & - \sum \lambda^2  N^{-1}  \frac{1}{(z - z_i^-)^2}\partial_z -  2   \sum \lambda^2  N^{-2}   \frac{1}{(z_j^+ - z_i^-)(z - z_i^-)} \partial_z  
\end{split} 
\end{equation}
Here, we have chosen to place the constant vector field $\partial_z$ in $V_-$, which is allowed as $\partial_z$ extends over all of $\CP^1$
.

Retaining only the terms up to order $\lambda^2$, we have
\begin{equation} 
\begin{split} 
 \iota_{V_+} \what{\alpha} = &  - \sum \lambda   N^{-1}  (z-z_i^-)^{-1}  - \sum \lambda^2  N^{-1} (z - z_i^-)^{-2}  + \sum \lambda^2    N^{-2}   (z-z_i^-)^{-1}(z-z_j^+)^{-1} \\ 
&- 2\lambda^2    N^{-2}   \sum (z_j^+ - z_i^-)^{-1}  (z-z_i^-)^{-1}     \\
=&    - \sum \lambda  N^{-1}  (z-z_i^-)^{-1} - \sum \lambda^2    N^{-1} (z - z_i^-)^{-2}\\
&- \lambda^2 \sum    N^{-2}   (z_j^+ - z_i^-)^{-1} \left( (z-z_j^+)^{-1} - (z-z_i^-)^{-1} \right)   - 2\lambda^2   N^{-2}   \sum (z_j^+ - z_i^-)^{-1}  (z-z_i^-)^{-1}    \\
  =&    - \sum   N^{-1}  \lambda (z-z_i^-)^{-1} - \sum \lambda^2  N^{-1}  (z - z_i^-)^{-2} - \lambda^2    N^{-2}   \sum (z_j^+ - z_i^-)^{-1} \left(  (z-z_j^+)^{-1} + (z-z_i^-)^{-1} \right) 
\end{split} 
\end{equation}
The Lie derivative $\mc{L}_{V_+} \what{\alpha}$ is obtained by applying $\partial = \d z \partial_z$ to the expression above.  This means that
\begin{equation} 
\delta_{RG} \what{\alpha} = - 2 \lambda \partial (\iota_{V_+}  \what{\alpha} ). 
\end{equation}

 Note that each term in $\iota_{V_+} \what{\alpha}$ depends only on the difference of $z$ and some $z_i^{\pm}$. Therefore we can replace $\partial_z$ by $-\partial_{z_i^{\pm}}$ in each term. This brings us to the equation for the RG flow: 
\begin{equation}
\begin{split} 
& \delta_{RG}    \left(  \d z  - \sum \lambda N^{-1}   \frac{\d z}{z - z_i^+} + \sum \lambda  N^{-1}    \frac{\d z}{z - z_i^-} +  \sum \lambda^2  N^{-1}    \frac{\d z}{(z - z_i^+)^2} + \sum \lambda^2   N^{-1} \frac{\d z}{(z - z_i^-)^2} \right)  
=\\
&       -2\sum \lambda^2  N^{-1}  \d z\partial_{z_i^-}  (z-z_i^-)^{-1} - 2 \sum  N^{-1}  \lambda^3 \d z \partial_{z_i^-} (z - z_i^-)^{-2} \\ 
 &- 2\lambda^3  N^{-2} \sum (z_j^+ - z_i^-)^{-1}\left( \d z \partial_{z_j^+}(z-z_j^+)^{-1} - \d z \partial_{z_i^-} (z-z_i)^{-1} \right)  
\end{split} 
\end{equation}
This gives rise to the identities 
 \begin{equation} 
\begin{split}
 \delta_{RG} z_i^- &= -2\lambda - 2 \lambda^2N^{-1} \sum_j  \frac{1}{z_i^- - z_j^+}  \\ 
\delta_{RG} z_i^+ &=   -2 \lambda^2 N^{-1} \sum_j \frac{1}{z_i^+ - z_j^-}
\end{split} 
\end{equation}
Equivalently, after an overall shift of $z_i^{\pm}$ by $-\lambda$, the flow is 
 \begin{equation} 
\begin{split}
 \delta_{RG} z_i^- &= -\lambda - 2 \lambda^2 N^{-1} \sum_j \frac{1}{z_i^- - z_j^+}  + O(\lambda^3)\\ 
\delta_{RG} z_i^+ &= \lambda  - 2 \lambda^2 \sum_j N^{-1} \frac{1}{z_i^+ - z_j^-} + O(\lambda^3)
\end{split} 
\end{equation}
If we allowed ourselves different numbers of flavours $N_f^{(i)}$ at chiral defects and $\br{N}_f^{(i)}$ at anti-chiral defects, and write $\lambda = \hbar N$,  we find 
\begin{equation} 
\begin{split}
 \delta_{RG} z_i^- &= -\hbar N - 2 \hbar^2 \sum_j N N_f^{(j)}  \frac{1}{z_i^- - z_j^+}  + O(\lambda^3)\\ 
\delta_{RG} z_i^+ &= \hbar N  - 2 \hbar^2 \sum_j N \br{N}_f^{(j)} \frac{1}{z_i^+ - z_j^-} + O(\lambda^3)
\end{split} 
\end{equation}
This expression makes the topology of the planar diagrams contributing to the RG flow clear.

Note that the two-loop RG flow is scheme independent: it can not be removed by an algebraic change of coordinates on the space of coupling constants (which is the manifold parameterized by the $z_i$).  Any such change of coordinates would necessarily involve $\log (z_i^- - z_j^+)$.

\section{$4d$ Chern-Simons from generalized complex geometry}
Now let us turn to the holographic derivation of our results. The first step is to realize $4d$ CS with disorder defects as the theory living on a brane in some topological string theory.   It turns out that the space-time geometry we need for this topological string theory is a generalized Calabi-Yau manifold \cite{Hitchin:2003cxu, Gualtieri:2003dx}.

\begin{definition}
Given a $6$-manifold $M$, an (odd) generalized CY structure on $M$ is the following data:
\begin{enumerate}
\item A closed odd form $\phi \in \Omega^{odd}(M,\C)$ with complex coefficients. Generally $\phi$ will be of mixed degree.
\item The form $\phi$ is required to be pure. This means the following. At each point $p \in M$, we can consider pairs $(V,\eta) \in (T_p M \oplus T^\ast_p M) \otimes_{\R} \C$ of a complex-coefficient vector and one-form.  The annihilator of $\phi$ is the collection of those pairs $(V,\eta)$ such that, at $p$,
\begin{equation}
V \vee \phi + \eta \wedge \phi = 0.
\end{equation}
Purity of $\phi$ means that the annihilator of $\phi$ is of rank the dimension of $M$.
\item Further $\phi$ must satisfy a non-degeneracy condition, stating that at each point in $M$ a certain inner product $\ip{\phi, \br{\phi}}$ is non-zero. This will hold for all examples we consider.
\end{enumerate}
\end{definition}

Here are some examples, together with the corresponding topological strings theory.
\begin{enumerate}
\item If $M$ is given an ordinary Calabi-Yau structure with holomorphic volume form $\Omega$, then $\Omega$ also defines a generalized CY structure.  The topological string theory we study in this geometry is the $B$-model.
\item If $M$ is an ordinary CY, equipped with a closed $(1,0)$ form $\alpha$, then $\Omega + \alpha$ is a generalized CY.  We can interpret $\alpha$ as a Poisson bivector; the topological string theory is the non-commutative $B$-model. 
\item If $M$ is a product of a symplectic manifold $X$ with symplectic form $\omega$, and a Riemann surface $\Sigma$ with a nowhere vanishing one-form $\alpha$, then setting $\phi = e^{\i \omega} \alpha$ defines a generalized CY structure on $M$. The topological string here is a product of the $A$-model on $X$ and the $B$-model on $\Sigma$.  
\end{enumerate}

Two generalized CY structures are equivalent if they are related by a transformation which is a combination of a coordinate transformation, together with a transformation $\phi \mapsto e^{\d A} \phi$ for some real one-form $A$.   Multiplying $\phi$ by $e^B$ for a real two-form $B$ amounts to changing the $B$-field, and adding on an exact $B$-field is a gauge trivial operation.

A $B$-field transformation can move us between the non-commutative $B$-model and the mixed $A/B$ model.  Suppose that $X$ is a holomorphic symplectic manifold of real dimension $4$, with holomorphic symplectic form $\Omega$.  Suppose $\Sigma$ is a Riemann surface with a nowhere-vanishing one form. Then $X \times \Sigma$ is a generalized Calab-Yau in two different ways.  If we view $X \times \Sigma$ as a non-commutative $B$-model, we can give it the generalized CY structure given by the form
$$
e^\Omega \alpha = \Omega \wedge \alpha + \alpha.
$$ 
If we view $X$ as a real symplectic manifold with symplectic form $\op{Im} \Omega$, then we have the generalized CY structure given by $e^{\i \op{Im} \Omega}\alpha$.  These differ by the $B$-field transformation $e^{\op{Re} \Omega}$. We see that if $\Omega$ is exact, then the generalized CY structures giving the non-commutative $B$-model and the mixed $A/B$ model are equivalent.

It is known \cite{Kapustin:2005vs} that, in this context, the non-commutative $B$-model and the mixed $A-B$-model on $X\times \Sigma$ are equivalent.  Let us review why the space of states of the topological strings are the same, working in a local patch where $X = \C^2$ and $\Sigma = \C$.
 
The space of states of the $B$-model topological field theory are $\PV^{\ast,\ast}(\C^3)$ with differential $\dbar$.   Making the $B$-model non-commutative turns on the Poisson cohomology differential $\{\pi,-\}$. Let us describe this in a local patch $\C^3 = \R^4 \times \C$ with coordinates $p,q,z$ and Poisson tensor $\pi = \partial_{p} \wedge \partial_{q}$. Then, $\PV^{\ast,\ast}(\C^2)$ with the differential $\dbar + \{\pi,-\}$ is isomorphic to forms on $\C^2$ with differential $\dbar + \partial$.  Thus, the space of states of the non-commutative $B$-model is $\Omega^\ast(\R^4) \what{\otimes} \PV^{\ast,\ast}(\C)$, with differential $\d_{\R^4} + \dbar_{\C}$.  

The space of states of the $B$-model topological string theory differ from those of the TFT because we need to gauge worldsheet diffeomorphisms.  This implements the constraint that fields are in the kernel of the divergence operator $\partial$ on polyvector fields.  The divergence free constraint, which comes from gauging worldsheet diffeomorphisms in the $B$-model topological string, is only non-trivial in the $\PV^{\ast,\ast}(\C)$ direction. In the $\R^4 = \C$ direction, it is homotopically trivial.  Thus, this constraint does not effect the equivalence between non-commutative $B$-model and mixed $A/B$ model.  

More generally, if we have a Calabi-Yau $3$-fold  $M$ equipped with a holomorphic Poisson tensor $\pi$, in the region where $\pi$ is of maximal rank (meaning the map $\pi:  T^\ast M \to T M$ is of rank $2$) then the non-commutative $B$-model on $X$ is equivalent to a mixed $A/B$ model.
 
One can check that the theories on branes in the mixed $A/B$ model, and in the non-commutative $B$-model, are also equivalent. We will see this in more detail later when we show how to engineer $4d$ CS from the theory on a brane. 

The relation between the mixed $A/B$ model topological string and $4d$ Chern-Simons was first discussed in \cite{Yamazaki:2019prm}.   
\begin{lemma}
In the mixed $A/B$ model on $\R^4 \times \C$ where $\C$ has the one-form $\d z$, the theory living on a stack of branes wrapping $\R^2 \times \C$ is $4d$ Chern-Simons with gauge group $\mf{gl}_N$. 
\end{lemma}
Note that the topological string naturally leads to analytically continued gauge theories with complex gauge group, and to analytically-continued integrable models.  
\begin{proof}
We need to calculate the open-string field theory of this brane and match it with $4d$ Chern-Simons theory.  This has already been discussed in the literature \cite{Yamazaki:2019prm}, but let us give a different which will generalize to the case with defects.  The open-string fields for a Lagrangian $A$-brane wrapping $\R^2 \subset \R^4$ are $\Omega^\ast(\R^2)$. Those for a $B$-wrapping wrapping $\C$ in the $B$-model on this curve are $\Omega^{0,\ast}(\C)$. The fields for the mixed $A/B$ model are obtained by tensoring these together, yielding the dg algebra $\Omega^\ast(\R^2) \what{\otimes}\Omega^{0,\ast}(\C)$ (where $\what{\otimes}$ indicates the completed tensor product).    The open-string field theory is obtained in the standard way by performing a cohomological shift by $1$, and using the Chern-Simons action.  If we have a stack of $N$ branes, the open-string field theory has fields
\begin{equation}
\mc{A} \in \Omega^\ast(\R^2) \what{\otimes}\Omega^{0,\ast}(\C) \otimes \mf{gl}_N[1]
\end{equation}
with action
\begin{equation}
\int \alpha \op{tr} ( \half \mc{A} \d \mc{A} + \tfrac{1}{3} \mc{A}^3 ). 
\end{equation}
This is precisely the 4d CS action in the BV formalism. 
\end{proof}

Next, we need to build a generalized CY manifold from a classical spectral curve, so that the theory living on a brane is $4d$ Chern-Simons theory with defects. 
  
\begin{lemma}
Let $(\Sigma = \Sigma_+ \cup \Sigma_-,\alpha)$ be a classical spectral curve. Let $\Sigma_{\text{finite}}$ be the complement of the second-order poles of $\alpha$.  

Then, there is a rank $2$ vector bundle $V$ over $\Sigma_{\text{finite}}$, and a generalized Calabi-Yau structure  on the $6$-fold $V \times \R^2$, so that the theory living on a stack of branes wrapping $\Sigma_{\text{finite}} \times \R^2$ is $4d$ CS with chiral and anti-chiral disorder defects, and gauge group $U(N)$.  
\end{lemma}
\begin{proof}
Over $\Sigma_+$ the generalized CY is the ordinary CY3 
$$T^\ast \Sigma_+ \times \C$$
 with form $\Omega + \alpha$, where $\Omega$ is the natural holomorphic $3$-form on the Calabi-Yau $T^\ast \Sigma \times \C$.  On $\Sigma_-$ it is 
$$T^\ast \Sigma_- \times \br{\C}.$$
On $\Sigma_+ \cap \Sigma_-$ it is the mixed $A/B$ model geometry
\begin{equation*} 
\Sigma_+ \cap \Sigma_- \times \R^4 
\end{equation*}
with 
$$\phi = e^{\i \omega} \alpha$$ 
On $\Sigma_+ \cap \Sigma_-$ we can trivialize the cotangent bundle of $\Sigma$ using $\alpha$, and use this trivialization to glue $(\Sigma_+ \cap \Sigma_-) \times \R^4$ to the other patches.  More explicitly, the one-form $\alpha$ gives an embedding of Calabi-Yau manifolds
$$
(\Sigma_+ \cap \Sigma_-) \times \C^2 \subset T^\ast \Sigma_+ \times \C
$$ 
where the CY structure on $(\Sigma_+ \times \Sigma_-) \times \C^2$ is given by the $3$-form $\alpha \d p \d q$, where $p,q$ are holomorphic coordinates on $\C^2$. This embedding is the identiy on $\Sigma_+ \times \Sigma_-$ and on the copy of $\C$ with coordinate $q$. It sends the function $p$ to the linear function on the cotangent fibres of $\Sigma_+$ given by $1/\alpha$.  

A $B$-field transformation makes the generalized CY structure on $(\Sigma_+ \cap \Sigma_-) \times \R^4$ equivalent to the one given by the product of the symplectic structure on $\R^4$ and the Calabi-Yau structure on $\Sigma_+ \cap \Sigma_-$. 

Similarly, we have an embedding 
$$
(\Sigma_+ \cap \Sigma_-) \times \br{\C}^2 \subset T^\ast \Sigma_- \times \br{\C}
$$
Here, we view $(\Sigma_+ \times \Sigma_-) \times \br{\C}^2$ as a Calabi-Yau manifold with volume form $-\alpha \d \br{p} \d \br{q}$, using the opposite complex structure on $\C^2$. This embedding is the same as above, except it sends the function $\br{p}$ -- holomorphic in the complex structure we are using -- to the function $1/\alpha$ linear on the cotangent fibres of $\Sigma_-$. 

Now, $(\Sigma_+ \cap \Sigma_-) \times \br{\C}^2$ and   $(\Sigma_+ \cap \Sigma_-) \times \C^2$ are isomorphic generalized complex manifolds, by an isomorphism which is the identity on the underlying manifolds but involves a $B$-field transformation.  This allows us to build our global generalized CY manifold.

The resulting manifold is of the form $V \times \R^2$ where $V \to \Sigma$ is a rank $2$ real vector bundle over $\Sigma$.  On $\Sigma_{\pm}$, $V$ is the cotangent bundle of $\Sigma$, but it is glued together on the overlap by trivializing using $\alpha$ and then applying complex conjugation.

Let us now study what happens at the zeroes of $\alpha$.  Near each zero, we describe the system as a non-commutative $B$-model.  Let us first look at the open-string field theory on a brane in the $B$-model without a Poisson tensor. Consider $\C^3$ with coordinates $u_1,u_2,z$ and a brane placed at $u_1 = 0$. The open-string states on the brane is
\begin{equation}
\Omega^{0,\ast}(\C^2)[\eps] \otimes \mf{gl}_N
\end{equation}
were $\eps$ is a fermionic parameter of ghost number $1$. This complex is given the differential $\dbar$. When we shift the ghost number by $1$ to produce the space of fields of the string field theory, $\eps$ becomes of ghost number $0$ and is the scalar field related to motion of the brane in the normal direction.  

To start with let us assume we have a simple zero. Let us turn on a Poisson tensor $z \partial_{u_1} \partial_{u_2}$ corresponding to a one-form $z \d z$.  This deforms the action by introducing the differential $z \eps \partial_{u_2}$ .

On the locus $z \neq 0$, the $B$-model space-time $\C^3$ is identified with the mixed $A/B$ model geometry $\R^4 \times \C$ by the transformation $q = u_2$, $p = z^{-1} u_1$.  Similarly, we can identify the $B$-model open string field theory with the mixed $A/B$ model by identifying 
\begin{equation}
z \eps = \d q. 
\end{equation}
This is natural, because $\eps$ transforms in the opposite way to $u_1$. 

What this means is that on the defect, at $z = 0$, the field $z^{-1} \d q$ is regular, i.e.\ the $\d q$ component of the 4d CS gauge field is allowed to have a pole.  This is precisely the disorder defect at a first-order pole of the one-form studied in \cite{Costello:2019tri}.    

If we have a zero of order $k$ in the one-form, the same argument shows that the $\d q$ component of the gauge field has a pole of order $k$.
 
Since this analysis goes through whether we use chiral or anti-chiral defects, the result follows.
\end{proof}

\section{Second-order poles in the one-form} \label{sec:neuman}
In \cite{Costello:2019tri} it was important to consider one-forms which have second order poles, as well as zeroes. This is of course necessary if our Riemann surface is $\CP^1$. The most familiar integrable field theories arise in this way. For instance, the principal chiral model arises from a one-form on $\CP^1$ which has two second order poles and two zeroes.

The boundary conditions used in \cite{Costello:2019tri} at second-order poles are of Dirichlet type: they require that the 4d CS gauge field $A$ vanishes at the location of a second-order pole.  If we start with 4d CS with gauge group $G$, this will lead to models which have a $G$-symmetry associated to every second order pole. This is because Dirichlet boundary conditions force gauge transformations to be the identity at the location of the second-order pole, so that constant gauge transformations at the pole become symmetries.

For example, the principal chiral model is a $\sigma$-model with target $G$, and so has two $G$ symmetries from the left and right actions. These symmetries come from the $G$-action at the two second-order poles of the one-form on $\CP^1$.

Clearly this is not consistent with a holographic analysis, since the gravity dual of a large $N$ gauge theory can not see anything charged under $U(N)$.  When one studies the holographic dual of three-dimensional vector models, this problem is resolved by coupling them to Chern-Simons gauge theory, which does not introduce any new degrees of freedom but restricts us to gauge-invariant states.  

We can do something similar for two-dimensional theories. Instead of coupling to Chern-Simons theory, we can couple to two-dimensional topological BF theory. This does not introduce any new degrees of freedom, but does impose gauge invariance.  

Coupling the $2d$ integrable model to topological $BF$ theory means that we introduce a $GL(N)$ gauge field $A$ which couples to the current $J$ of the $2d$ model, and also a Lagrange multiplier field enforcing $F(A) = 0$.  The new terms in the Lagrangian are
\begin{equation}
\int B F(A) + F(A) J\,.
\end{equation} 
There is such a $GL(N)$ symmetry for each second order pole in the one-form $\alpha$. The current $J$ for the symmetry is expressed in terms of the Lax matrix of the $2d$ model, which is the $4d$ Chern-Simons gauge field. The boundary conditions mean that, at each second order pole, the Lax matrix vanishes; but the current is the first derivative of the Lax matrix at the pole. 

Let us see what we need to do to $4d$ Chern-Simons to introduce this coupling.  Let us work in a  patch where the one-form has a second-order pole, and takes the form 
\begin{equation}
2 \pi \i \alpha = z^{-2} \d z + a z^{-1} \d z + O(1).
\end{equation}
Initially, we have  Dirichlet boundary conditions require that $A_{4d} =  0$ at $z = 0$.  Since the theory with Dirichlet boundary conditions has a global symmetry, we can couple to a background two-dimensional $GL(N)$ gauge field $A_{2d}$. This is achieved by modifying the boundary conditions so that $A_{4d} = A_{2d}$ at $z = 0$.

It is standard in many contexts that gauging Dirichlet boundary conditions result in Neumann boundary conditions. We will argue that this is the case here: we claim that gauging the $2d$ symmetry replaces Dirichlet boundary conditions with a modified Neumann boundary condition, which imposes the constraint that 
\begin{equation}
\half a A + \partial_z A = 0 
\end{equation} 
at $z = 0$. 

In this Neumann type boundary condition, we also require that the gauge transformations $\chi$ satisfy the same equation, $\half a \chi + \partial_z \chi = 0$ at $z = 0$.  With this restriction on the gauge transformations the Neumann-type boundary conditions are gauge invariant. To see this, note that gauge variation of the Chern-Simons Lagrangian is
\begin{equation} 
\delta CS(A) = \d\op{tr} (\chi F(A)) + \d \op{tr}( \d_A \chi \wedge A). 
\end{equation} 
After integrating by parts we find that the gauge variation of the action is
\begin{equation} 
\int \d \alpha \wedge \left(  \op{tr} (\chi F(A)) + \d \op{tr}( \d_A \chi \wedge A) \right) 
\end{equation}
Since $\d \alpha = a \delta_{z = 0} - \partial_z \delta_{z = 0}$, a further integration by parts tells us the gauge variation is
\begin{equation} 
\int_{z = 0} (a + \partial_z)  \left(  \op{tr} (\chi F(A)) + \d \op{tr}( \d_A \chi \wedge A) \right) 
\end{equation}
which vanishes if both $\chi$ and $A$ are annihilated by $\half a + \partial_z$.
 
Now let us see why these modified Neumann boundary conditions lead to two-dimensional BF theory. This is easiest to see when we consider pure $4d$ Chern-Simons with no defects.  There, the one form globally is $\d z$, with a second order pole at $\infty$ with no residue.  We need to show that the two-dimensional theory we find by compactifying on $\CP^1$ is BF theory.

Standard arguments  show that the field content in two dimensions can be read from the zero modes of the theory on $\CP^1$; there are no higher KK modes in this context. The zero modes, in turn, are the Dolbeault cohomology of $\CP^1$ with coefficient in the gauge bundle, with appropriate boundary conditions.

Let $\Omega^{0,\ast}(\CP^1, \partial_{\infty})$ be the Dolbeault complex consisting of Dolbeault forms whose first deriavative in $z$ vanishes when restricted to $\infty$.  It is not hard to show that
\begin{equation} 
\begin{split} 
H^0  \Omega^{0,\ast}(\CP^1, \partial_{\infty}) & = \C \\
 H^1  \Omega^{0,\ast}(\CP^1, \partial_{\infty}) & = \C
\end{split} 
\end{equation}
Indeed, there is only one holomorphic function whose first derivative vanishes at $\infty$, which is the constant function.  To calculate $H^1$ we can use the short exact seqence of cochain complexes
\begin{equation}
0 \to \Omega^{0,\ast}(\CP^1, \partial_{\infty}) \to \Omega^{0,\ast}(\CP^1) \to \C \to 0
\end{equation} 
where the map to $\C$ is given by differentiating at $\infty$.   The long exact sequence in cohomology immediately tells us that the first cohomology is $\C$. It is represented by $\delta_{z = 0}$; this $(0,1)$ form is $\dbar z^{-1}$, but since the first derivative of $z^{-1}$ does  not vanish at $\infty$, $\delta_{z = 0}$ is not exact in the complex we are considering. 

This analysis tells us that the effective two-dimensional theory has a field $A \in \Omega^1(\R^2) \otimes \mf{gl}_N$, which is constant on $\CP^1$ and comes from the class in $H^0    \Omega^{0,\ast}(\CP^1, \partial_{\infty})$; and a field $B \in \Omega^0 (\R^2)\otimes \mf{gl}_N$, whose uplift to $\R^2 \times \CP^1$ is the $(0,1)$ form $B \delta_{z = 0}$. The $4d$ Chern-Simons action on these fields reduces to the BF action. 
 
This analysis immediately implies that if we introduce order defects to $4d$ Chern-Simons on $\R^2 \times \CP^1$, the effective $2d$ theory is obtained from the integrable models studied in \cite{Costello:2019tri} by coupling topological BF theory to the gauge symmetry at the second order pole of $\d z$. In the case of disorder defects, we can have more than one second order pole in the one form. Every time we change the boundary condition at a second order pole from Dirichlet to Neumann, we introduce $2d$ BF gauge fields which gauge the symmetry living at the Dirichlet boundary condition.  

\section{Backreactions in generalized complex geometry}
Our goal is to study the integrable field theory holographically. To do this, we will analyze the backreaction.  A  backreaction is a modification of the generalized CY structure on $V \times \R^2$ to introduce a source on the brane. This means that the mixed-degree form $\phi$ defining the geometry is no longer closed, but satisfies
\begin{equation} 
\d \phi = 2\pi\i \lambda \delta_{\Sigma \times \R^2} 
\end{equation} 
where $\lambda$ is the 't Hooft coupling. (The factor of $2 \pi \i$ is a convenient normalization of the coupling constants). On the complement of the brane, this defines a new generalized complex structure which is the backreacted geometry.

We will first analyze the backreacted geometry in the context when $\lambda$ is a formal parameter. Then we will see how to do this when $\lambda$ is finite, under some very strong hypothesis on $\what{\alpha}$.  

  We start with the case with no defects, where the brane wraps $\R^2 \times \C$ inside $\R^4 \times \C$.  Viewing $\R^4$ as the cotangent bundle of $\R^2$, it is natural to give coordinates $q$ on the $\R^2$ which wraps the brane, and $p$ on the $\R^2$ normal to the brane.  We take $p,q$ to be holomorphic coordinates in a chosen complex structure on $\R^4 = \C^2$, in such a way that when we have chiral defects observables are holomorphic functions of $q$, and with anti-chiral defects they are holomorphic functions of $\br{q}$.

The complement of the brane is $\R_{> 0} \times S^1 \times \R^2 \times \C$, where we choose polar coordinates $(\abs{p},\theta= \op{Arg} p)$ on the copy of $\R^2$ normal to the brane.   The back-reacted geometry is given by the generalized CY with mixed-degree form
\begin{equation} 
\phi = \exp ( \op{Im} \d p \d q ) ( \d z + 2 \pi \i \lambda \d \theta ). 
\end{equation}
Equivalently,  the backreacted geometry is  the flat bundle over $\R_{> 0} \times S^1 \times \R^2$ with fibre $ \C$, where the monodromy around $S^1$ is the automorphism of $\C$ which sends $z \mapsto z + 2 \pi \i \lambda \d \theta$ (where $\lambda$ is the 't Hooft coupling).

\subsection{The framing anomaly of $4d$ CS}\label{sec:framing}
$4d$ Chern-Simons has a framing anomaly \cite{Costello:2017dso} which means that, as Wilson lines bend, they also need to move in the spectral parameter plane. 

We can see the framing anomaly holographically. Before we backreact,  a Wilson line placed at $z_0 \in \C$, in the sum of the exterior powers of the fundamental representation, is given by a Lagrangian brane in the $A$-model on $\R^2 \times z_0$.  Let us give $\R^4$ coordinates $p_1,p_2,q^1,q^2$ with symplectic form $\d p_i \d q^i$, and place the $4d$ Chern-Simons brane at $p_i = 0$.  A Wilson line wrapping a path $K \subset \R^2$ is represented by a Lagrangian brane which wraps the conormal to $K$.

We will see that the backreaction of the $4d$ Chern-Simons brane forces the brane to move in the $z$ plane if the corresponding Wilson line bends. Suppose the Wilson line is on the path $q^i(t)$.  Then the Lagrangian brane lives on the conormal, which has coordinates $s,t$ such that
\begin{equation}
\begin{split} 
 q^i(s,t) &= q^i (t) \\
p_i (s,t) &= s \eps_{ij} \frac{ \partial q^j(t) } { \partial t } \\  
\end{split}  
\end{equation}
Let $\phi(t) = \op{arctan} ( \dot{q}^2(t)  / \dot{q}^1(t) )$ be the angle of the Wilson line at time $t$ to the horizontal axis.  In the momentum plane $p_1,p_2$, the Lagrangian submanifold at time $t$ consists of a line in the direction $\theta(t) = \phi(t) + \pi/2$.   

In the backreacted geometry, the new holomorphic coordinate is $z + 2\pi \i \lambda \theta$ (where as before $\theta$ is the angular coordinate on the momentum plane).  A brane must live at a fixed value of this new coordinate.  Since the angle $\theta$ changes as we move along the Wilson line, so must the position $z(t)$ so that $z + 2 \pi \i \lambda \theta$ remains constant.  Clearly we must have
\begin{equation}
z(t) = -2 \pi i \lambda \phi(t) + z_0 
\end{equation}
Since $\lambda = \hbar N$ and $N$ is the dual Coxeter number of $\mf{gl}_N$, this is exactly the framing anomaly discussed in \cite{Costello:2019xhj}, up to a normalization of $2 \pi \i$.   

\subsection{Backreacting in the presence of defects } 
Consider a classical spectral curve $(\Sigma,\Sigma_+, \Sigma_-,\alpha)$. We want to backreact the generalized complex geometry in the $\Sigma_+$ patch. This is given by a non-commutative $B$-model type geometry, with underlying complex manifold $T^\ast \Sigma_+ \times \C$.  If $\omega$ denotes the canonical holomorphic symplectic structure on $T^\ast \Sigma_+$, the holomorphic volume form is $\Omega = \omega \d q$. The generalized CY structure is given by $\Omega + \alpha$. 

To backreact, we add to this a one-form on $T^\ast \Sigma_+$ which has a first order pole on $\Sigma_+$ with residue $\lambda$.  This data is equivalent to giving a $\lambda$-connection on the canonical bundle of $\Sigma_+$. The one-form is the connection one-form on the frame bundle. 

On the anti-chiral patch, we need to give a one-form on $T^\ast \Sigma_-$ with a pole on $\Sigma_-$ whose residue is $-\lambda$. The difference in sign of the residue is because gluing the chiral and anti-chiral patches involves complex conjugation on the fibres of the cotangent bundle, which reverse orientation. 

Let us choose a local coordinate $z$ on $\Sigma_+$, and trivialize the cotangent bundle using the one-form $\d z$. Let $p$ denote the corresponding holomorphic function on the cotangent fibres. In the holomorphic coordinates $p,q,z$, before backreaction the generalized CY structure is given by
\begin{equation} 
\d p \d q \d z + \alpha 
\end{equation}
where $\alpha = \partial \gamma(z)$ for some locally-defined holomorphic function $\gamma$. 

After backreaction, we change the generalized CY structure to one of the form
\begin{equation} 
\d p \d q \d z + \lambda \d \log p + \mu(z,\lambda) 
\end{equation}
where $\mu(z,\lambda)$ is a one-form which at $\lambda = 0$ is the original one-form $\alpha$.  In these coordinates, where the canonical bundle is trivialized by $\d z$, $\mu(z,\lambda)$ is the connection one-form.

\subsection{Gluing the backreacted geometries}
We have seen that, if $(\Sigma, \Sigma_+, \Sigma_-,\alpha)$ is a classical spectral curve, then to build the backreacted geometry we need the following data:
\begin{enumerate} 
\item A $\lambda$-connection $\nabla_+$  on $\Sigma_+$. 
\item A $-\lambda$-connection $\nabla_-$  on $\Sigma_-$.  
\end{enumerate}
We let $\what{X}_{\pm}$ be the backreacted geometries in each patch.

Here we will see how if we have a solution to the Bernoulli equation, we can glue the patches together. 

Before we give the construction, let us make an observation which explains why the Bernoulli equation is useful.   We let $\what{X}_+ = (T^\ast \Sigma_+\setminus \Sigma_+ ) \times \C$ be the backreacted geometry on the $\Sigma_+$ patch. This only depends on the $\lambda$-connection $\nabla_+$ on $\Sigma_+$.  
\begin{lemma} 
For any nowhere-vanishing one-form $\what{\alpha}$ on a patch of $\Sigma_+$,let $p$ be the coordinate on the cotangent fibres associated to $\what{\alpha}$.

Then $\what{\alpha}$ satisfies the Bernoulli equation if and only if the generalized CY form on $\what{X}_+$ is
\begin{equation} 
e^{\d p \d q} ( \lambda \d \log p + \what{\alpha} ) = \d p \d q \d \what{\alpha} + \lambda \d \log p + \what{\alpha}.
\end{equation} 
\end{lemma}
Having a generalized CY of this form is useful, because it is very close to that given by a product of a symplectic manifold and a Calabi-Yau curve. 

\begin{proof}
Since $p$ is the coordinate obtained from the trivialization given by $\what{\alpha}$, the holomorphic symplectic form on $T^\ast \Sigma_+$ is $\what{\alpha} \d p$, so the holomorphic $3$-form on $T^\ast \Sigma \times \C$ is $\d p \d q \what{\alpha}$.

The one-form component is given by $\lambda \d \log p$ plus the connection one-form in the chosen frame. The Bernoulli equation is the statement that the connection one-form in the frame given by $\what{\alpha}$ is $\what{\alpha}$.  
\end{proof}

For now we work in series in $\lambda$. We let 
\begin{equation} 
U_+ \subset \what{X}_+ 
\end{equation}
be the open subset living over $\Sigma_+ \cap \Sigma_-$, and similarly we have $U_- \subset \what{X}_-$. 
 
 On $\Sigma_+ \cap \Sigma_-$, we can trivialize the cotangent bundle using $\what{\alpha}$, and we let $p_{\pm}$ be the coordinates on $U_{\pm}$ coming from this trivialization.  In this trivialization we have
$$U_+ = (\C \setminus 0) \times \C \times (\Sigma_+ \cap \Sigma_-)$$ 
with the generalized CY structure  given by the form
\begin{equation} 
\d p_+ \d q_+ \what{\alpha} + \lambda \d \log p_+ +  \what{\alpha}. 
\end{equation}
Here $q_+$ is a coordinate on the second copy of $\C$. 
 
Similarly, 
$$U_- = (\C \setminus 0) \times \C \times (\Sigma_+ \cap \Sigma_-)$$ 
with coordinates $p_-, q_-$ on the two copies of $\C$ and generalized CY structure given by 
\begin{equation} 
\d p_- \d q_- \what{\alpha} - \lambda \d \log p_- + \what{\alpha}. 
\end{equation} 
We need to show that these are equivalent after a $B$-field transformation and a coordinate transformation.

  We let
\begin{equation} 
V = \frac{1}{\what{\alpha}} 
\end{equation} 
be the vector field on $\Sigma_+ \cap \Sigma_-$ inverse to $\what{\alpha}$. 

Define a formal diffeomorphism
\begin{equation} 
f_\lambda = e^{2 \lambda (\log \abs{p_+}) V } : U_+ \to U_+  
\end{equation}

Define an isomorphism
\begin{equation} 
\rho : U_+ \iso U_- 
\end{equation}
which sends $q_+ \to \br{q}_-$, $p_+ \to -\br{p}_-$ and is the identity on $\Sigma_+ \cap \Sigma_-$. 
\begin{lemma} 
Under the isomorphism 
$$
\rho \circ f_\lambda : U_+ \to U_-
$$
the generalized CY structures on $U_+$ and $U_-$ agree up to a $B$-field transformation. 
\end{lemma}
\begin{proof}
First, a $B$-field transformation on each generalized CY takes that on $U_+$ to the one
\begin{equation} 
\phi_+ = e^{\i \op{Im} \d p_+ \d q_+} \left(  \what{\alpha} + \lambda \d \log p_+ \right) 
\end{equation}
and on $U_-$ to
\begin{equation} 
\phi_- = e^{-\i \op{Im} \d p_- \d q_-} \left(  \what{\alpha} - \lambda \d \log p_- \right) 
\end{equation}
The isomorphism $\rho : U_+ \to U_-$ allows us to view both generalized CY structures as forms on $U_+$.  
We drop the subscripts $\pm$, and write $p = p_+$, $q = q_+$, $\br{p} = p_-$, $\br{q} = q_-$. We use polar coordinates $p = \abs{p} e^{2 \pi \i \theta}$.   Then, the generalized CY structures on $U_+$  and $U_-$ are
\begin{equation} 
\begin{split} 
\phi_+ &=  e^{\i \op{Im} \d p \d q} \left(  \what{\alpha} + 2\pi \i  \lambda \d \theta+ \lambda \d \log \abs{p}    \right) \\
 \rho^\ast \phi_- &= e^{\i \op{Im} \d p \d q} \left(  \what{\alpha} + 2 \pi \i \lambda \d \theta - \lambda \d \log \abs{p} \right) 
\end{split} 
\end{equation}
They differ by the sign of the $\lambda \d \log \abs{p}$ term.  The isomorphism $f_\lambda$ has the feature that
\begin{equation} 
f_\lambda^\ast \what{\alpha}  = \what{\alpha} + 2 \lambda \d\log \abs{p}. 
\end{equation} 
so that
\begin{equation} 
f_\lambda^\ast\rho^\ast \phi_- = \phi_+.  
\end{equation}
\end{proof}

\subsection{The RG flow}
In the construction of the backreacted geometry, the submanifolds $\what{X}_{\pm}$ have a scaling symmetry, where we scale both $p$ and $q$. However, the transformation $f_\lambda$ which we use to glue the two patches together does not commute with this symmetry.  Instead, if we perform an infinitesimal scaling where $q$ is sent to $(1-\log \mu)q$ and $p$ to $(1+\log \mu) p$, then the transformation $f_\lambda$ is sent to 
\begin{equation} 
2 \lambda V f_\lambda.  
\end{equation} 
This tells us that the geometry where we have performed an infinitesimal rescaling is equivalent to the original construction of the geometry but for a new spectral curve. The new spectral curve is one where the gluing of $\Sigma_+$ to $\Sigma_-$ is modified by precomposing with the vector field $\lambda V$ on $\Sigma_+$. 

This is exactly the RG flow for planar spectral curves as given in definition \ref{def:rgflow}.

This modified spectral curve is the one obtained from the Beltrami differential $2 \lambda V \in H^1(\Sigma, T \Sigma)$.

In this construction, it is very important that the Lie derivative of $V$ preserves the connections $\nabla_{\pm}$. We have seen \ref{lemma:bernoullivector} that this is equivalent to the Bernoulli equation.  This is what guarantees that the new curve has the data of a planar spectral curve.

\subsection{RG flow for line defects}
Line defects in $4d$ CS come from branes of dimension $2$ in the dual geometry. These branes are best described in the mixed $A/B$ model frame, where the wrap Lagrangian branes in the four $A$-model direction and a point in the $B$-model direction. 

In the backreacted geometry, the mixed $A/B$ model description appears over the region $\Sigma_+ \cap \Sigma_-$, where the geometry is
\begin{equation} 
U_t = \Sigma_+ \cap \Sigma_- \times (\R^2 \setminus 0) \times \R^2.  
\end{equation}
The generalized CY structure is
\begin{equation} 
\phi_t = e^{\i \op{Im} \d p \d q} (  \what{\alpha} + 2 \pi \i \lambda \d \theta).
\end{equation} 
In this topological patch, when we set $\lambda = 0$,  the geometry is simply a product of the one-dimesional Calabi-Yau $\Sigma_+ \cap \Sigma_-$ (with one form $\alpha$) and the symplectic manifold $(\R^2 \setminus 0) \times \R^2$, with symplectic form restricted from that on $\R^4$.   When we back react, we deform the one-form $\alpha$ to $\what{\alpha}$, and we also modify the geometry so that it is no longer a product but a fibration $\Sigma \to (\R^2 \setminus 0) \times \R^2$. The fibration has monodromy, as we go around the $\theta$ plane, $2 \pi \i \lambda V$.

To understand this fibration structure, it suffices to work in a local patch where $\what{\alpha} = \d z$ and $V = \partial_z$, in which case we are in the situation discussed in section \ref{sec:framing}.  

We are interested in placing a Lagrangian brane at $\op{Im} p = 0$ and $\op{Im} q = 0$, and a fixed point in $\Sigma_+ \cap \Sigma_-$.  This makes sense, because when we restrict to a fixed angle $\theta$ the fibration of $\Sigma_+ \cap \sigma_-$ over $(\R^2 \setminus 0)$ is of course trivial.    

Our goal is to understand what this brane becomes in the chiral or anti-chiral patch.
 
Recall that the chiral patch of the geometry $\what{X}_+$ is $(T^\ast \Sigma_+ \setminus \Sigma_+) \times \C$. If we restrict to the region $U_+$ of the chiral patch living over $\Sigma_+ \times \Sigma_-$, and we trivialize the canonical bundle of $\Sigma_+ \cap \Sigma_-$ using $\what{\alpha}$, we can identify
\begin{equation} 
U_+ = (\Sigma_+ \cap \Sigma_-) \times (\C \setminus 0) \times \C 
\end{equation}
with form
\begin{equation} 
e^{\d p_+ \d q_+} ( \what{\alpha} + \lambda \d \log p_+). 
\end{equation}
After a $B$-field transformation this becomes
\begin{equation} 
\phi_+ =  e^{\i \op{Im} \d p_+ \d q_+} ( \what{\alpha} + \lambda \d \log p_+). 
\end{equation}
There is an obvious isomorphism $\rho : U_+ \to U_t$ which is the identity on $\Sigma_+ \cap \Sigma_-$ and sends $p_+ \to p$, $q_+ \to q$.  More or less as we did when we identified $U_+$ and $U_-$, we let
\begin{equation} 
g_\lambda = e^{\lambda \log \abs{p_+} V} : U_+ \to U_+ 
\end{equation} 
(note that when identifying $U_+$ with $U_-$ we used $2V$ instead of $V$).  Then,
\begin{equation} 
g_\lambda^\ast \rho^\ast \phi_t = \phi_+  
\end{equation} 
so that $\rho \circ g_\lambda$ is an isomorphism of generalized CY manifolds. 

If we take a Lagrangian brane in $U_t$, cut out by the equations $f(z) = C$, $\op{Im} p = 0$, $\op{Im} q = 0$, then it becomes the brane in $U_+$ cut out by the equations
\begin{equation} 
\begin{split} 
\op{Im} p_+ &= 0\\
\op{Im} q_+ &= 0\\
e^{\lambda \log p_+ V} f(z) = C. 
\end{split} 
\end{equation}
This means the following. The brane is a product of the real axis in the $q$-plane, and a line in $\Sigma_+ \cap \Sigma_-\times (\C \setminus 0)$. This second line lives on the real axis of the $p_+$ plane. Suppose at $p_+ =1$ this line is at some point $z \in \Sigma_+ \cap \Sigma_-$.  Then at some other value of $p_+$, the point in $\Sigma_+ \cap \Sigma_-$ has moved by $\lambda \log p_+$ along the trajectory of the vector field $V$.

Let us consider the moduli space of Lagrangian branes which wrap the line $\op{Re} q = 0$, and which live in the chiral patch (where we work in perturbation theory).  If we label a brane by its position at $p_+ = 1$, we see that this moduli space is $\Sigma_+$, where we should bear in mind the brane will behave badly at the location of disorder defects.   

The RG flow will be a vector field on the moduli space consisting of the planar spectral curve $\Sigma$, together with a point on $\Sigma$ corresponding to the brane.  Clearly, when we apply the RG flow, the position in $\Sigma$ of the brane flows by the vector field $\lambda V$.

 \subsection{Order defects}
There are two kinds of defects in $4d$ Chern-Simons: disorder defects, where the one-form has a zero, and order defects, where we introduce extra degrees of freedom. In this section we will briefly discuss the holographic dual of order defects.

Consider $4d$ Chern-Simons on $\Sigma \times \R^2$, where $\Sigma$ is a classical spectral curve. We can couple to chiral or anti-chiral fermions wrapping $z \times \R^2$, for some point $z \in \Sigma$.  The fermions we will consider will live in the fundamental plus anti-fundamental representation of $GL(N)$. The action coupling the two systems (for the case of chiral fermions) is
\begin{equation} 
\int_{\R^2} \psi_i \dbar \psi^i + \psi_i A_{\br{q}\  j}^i \psi^j 
\end{equation} 
Coupling chiral or anti-chiral fermions in this way will modify integrable field theory by introducing fermionic degrees of freedom coupled. 
If $\Sigma = \C$ with one form $\d z$, with chiral and anti-chiral fermions placed at $z_0$ and $z_1$,  then the resulting integrable field theory is the Thirring type model with Lagrangian
\begin{equation} 
\psi_i \dbar \psi^i + \br{\psi}_i \partial \br{\psi}^i + \frac{1}{(z_0 - z_1) } \psi_i \psi^j \br{\psi}_j \br{\psi}^i. \label{eqn_thirring} 
\end{equation}

In the string theory set-up, chiral and anti-chiral fermions are obtained byintroducing extra branes.  If we have a classical spectral curve $\Sigma = (\Sigma_+, \Sigma_-, \alpha)$, we can introduce chiral fermions at a point in $\Sigma_+$. We have seen that $4d$ CS comes from a brane in a generalized Calabi-Yau manifold which has a patch which is $T^\ast \Sigma_+ \times \C$. In this patch, the string theory is a non-commutative $B$-model. 

The brane which realizes chiral fermions is a $B$-model brane, which wraps the divisor $T^\ast_z \Sigma_+ \times \C$ living over a point $z_0 \in \Sigma_+$.   To get the correct parity of the strings stretched between this brane and the $4d$ Chern-Simons brane, this sheaf needs to have odd parity -- that is, it needs to be
$$
\Pi \Oo(T^\ast_z \Sigma_+ \times \C) \otimes \C^{N_f}.
$$ 
where we have $N_f$ flavours of fermions.

When we backreact, the one-form acquires a pole at the location of the brane.  In a coordinate patch $p,q,z$ where the $4d$ CS brane is at $p = 0$, and the brane introducing chiral fermions is at $z = 0$, the generalized CY giving the backreacted geometry is
\begin{equation} 
\d p \d q \d z + \lambda \d \log p - \frac{N_f}{N} \lambda \d \log z. 
\end{equation} 
This tells us the following: when we have a chiral fermion defect at a point $z \in \Sigma_+$, the $\lambda$-connection $\nabla_+$ acquires a first order pole at the defect, whose residue is $-\lambda \frac{N_f}{N}$. The presence of the pole tells us that the monodromy of the connection $\lambda^{-1} \nabla_+$ around the defect is $e^{-2 \pi \i N_f / N}$.

\section{Non-perturbative construction of the  backreacted geometry} 
Now let us discuss what happens when we do not work in series in $\lambda$, but instead have a planar spectral curve at finite $\lambda$. 

In this section we will make a very strong assumption on the nature of the solution $\what{\alpha}$ to the Bernoulli equation. In practice it is not clear that this assumption holds.  

We will assume that the region $\Sigma_+ \cap \Sigma_-$ is an annulus of some small radius $\eps$.  We will also assume that there exists a coordinate $z$ on this annulus, for $1-\eps < \abs{z} < 1+ \eps$, so that 
\begin{equation} 
\what{\alpha} = \kappa \lambda \d \log z 
\end{equation} 
for some positive real number  $\kappa$.   

The RG flow is given as usual by modifying the gluing of the two patches using the vector field 
$$
V = 1/\what{\alpha} = \kappa^{-1} z \partial_z.
$$
This vector field points along the radial direction of the cylinder.  As we move towards the UV, regluing using $V$ will stretch out the cylinder.  Deep in the UV the cylinder becomes infinite, and the chiral and anti-chiral degrees of freedom are decoupled.  In this limit the theory is essentially free, which we expect from asymptotic freedom of these models.

As we move towards the IR, in contrast, the cylinder well get smaller and at some point we can go no further.      As Davide Gaiotto explained to us, we expect that in a theory with a mass gap, the holographic dual geometry only exists in a region near the boundary (corresponding to the UV).    This phenomenon is visible in this example.

Unfortunately, we do not understand the non-perturbative models well enough to give a quantitive understanding of the mass gap or of the spectrum of massive particles. We hope to reconsider this question in the future.

Let us now describe how to build the non-perturbative geometry.   The non-perturbative geometry will be built by gluing three patches, associated to $\Sigma_+$, $\Sigma_-$, and the intersection $\Sigma_+ \cap \Sigma_-$.

The patch associated to $\Sigma_+$ is an open subset of the manifold $\what{X}_+$ we considered before. Recall that 
$$
\what{X}_+ = (T^\ast \Sigma_+ \setminus \Sigma_+) \times \C
$$
with the generalized CY structure given by the natural Calabi-Yau volume form, plus $\what{\alpha}$. 

We define an open subset of $\what{X}_+$ as follows. On $\Sigma_+ \cap \Sigma_-$, the cotangent bundle is trivialized by $\what{\alpha}$, leading to a holomorphic function $p_+$ on the cotangent fibres of $T^\ast \Sigma_+$. The function $p_+$ is defined on the patch of $\what{X}_+$ living over $\Sigma_+ \cap \Sigma_-$.

 We also have a coordinate $z_+$ on $\Sigma_+ \cap \Sigma_-$, because we assumed that this is a cylinder $1 - \eps < \abs{z_+} < 1 + \eps$.
 
 The open subset
 $$V_+ \subset \what{X}_+$$ 
 we are interested in is the  complement of the region where both $\abs{z_+} \ge 1-\eps$ and $\abs{p_+} \le 1$.  That is, the only points of $\what{X}_+$ we are removing are those that over the cylinder $\Sigma_+ \cap \Sigma_-$; on this cylinder, we are removing the locus where the momentum coordinate $\abs{p_+} \le 1$.

 We have a similar open subset 
 $$V_- \subset \what{X}_-.$$ 
  
The third patch we will build is a cylindrical geometry $Y$, defined as follows. Consider the generalized CY manifold $\C^\times \times \R^4$, where $\R^4 = \C^2$ has coordinates $p,q$ and $\C^\times$ has coordinate $z$. This has  generalized CY form
\begin{equation} 
e^{\i \op{Im} \d p \d q} (\lambda 2 \pi \i \d \theta + \kappa \lambda \d log z )
\end{equation}
Consider the open subset 
\begin{equation} 
Y \subset \C^\times \R^4 
\end{equation}
cut out by the inequalities 
\begin{equation}
\begin{split}
\abs{p} &> 1 \\ 
(1-\eps)  \abs{p}^{-1/ \kappa} < &\abs{z} < (1+\eps) \abs{p}^{1/\kappa}  
\end{split} 
\end{equation}
What this does is:
\begin{enumerate}
\item Remove the region where the momentum $\abs{p}$ is $\le 1$.
\item  For a given value of $p$, we restrict $z$ to be in an annulus inside $\C^\times$. The radius of this annulus increase with $p$ and shrinks to zero as $\abs{p} \to 1$. 
\end{enumerate}

We let 
$$U_+ \subset V_+ \subset \what{X}_+$$ 
be the open subset living over $\Sigma_+ \cap \Sigma_-$, where we also require that $\abs{p_+} > 1$.  Recall that $\Sigma_+ \cap \Sigma_-$ is an annulus, and here the cotangent bundle is trivialized by $\what{\alpha}$. We have coordinates $z_+, p_+, q_+$ on $U_+$, where $\abs{p_+} > 1$ and $1-\eps < \abs{z_+} < 1 + \eps$. The generalized CY form is
\begin{equation} 
e^{\d p_+ \d q_+} (\lambda \d \log p_+ +  \lambda \kappa \d \log z_+) . 
\end{equation} 
Similarly, we have a region 
$$U_- \subset V_- \subset \what{X}_-,$$ 
with coordinates $p_-,q_-,z_-$. 

We need to identify the three patches $(Y, V_+, V_-)$.   We will identify $U_+$ and $U_-$ with open subset of $Y$ and use this to glue $V_+$ and $V_-$ to $Y$.  

We first perform a $B$-field transformation on $U_+$ and $U_-$, bringing the generalized CY forms to
\begin{equation} 
e^{\i \op{Im} \d p_\pm \d q_\pm } ( \pm \lambda \d \log p_{\pm} + \lambda \kappa \d \log z_{\pm} )  = e^{\i \op{Im} \d p_\pm \d q_\pm } (  \lambda 2 \pi \i \d  \theta \pm \lambda \d \log \abs{p_{\pm} } + \lambda \kappa \d \log z_{\pm} )  
\end{equation} 
In contrast, the generalized CY structure on $Y$ is $e^{\i \op{Im} \d p \d q} ( \lambda 2 \pi \i \d \theta + \kappa \lambda \d \log z )$.   We will identify $U_{\pm}$ with patches in $Y$, but this must involve a coordinate transformation which absorbs the factors of $\pm \d \log \abs{p_{\pm}}$ in the generalized CY structure on $U_{\pm}$. 

We define a map
\begin{equation} 
U_+ \to Y
\end{equation}
by
\begin{equation} 
\begin{split} 
p_+ &= p\\ 
q_+ &= q\\
z_+ &= \abs{p}^{-1/\kappa} z 
\end{split} 
\end{equation} 
Since 
$$\d \log p_+ + \lambda \kappa \d \log z_+ = 2 \pi \i \d \theta + \lambda \kappa \d \log z$$
this transformation identifies the generalized CY structure on $U_+$ with the restriction of the one on $Y$. 

On $U_+$, we have $1-\eps < \abs{z_+} < 1 + \eps$.  Therefore, we identify $U_+$ with the region in $Y$ where 
$$(1 - \eps)\abs{p}^{1/\kappa} < \abs{z} < (1+\eps) \abs{p}^{1/\kappa}.$$
This is a region near the boundary of $Y$ where $\abs{z}$ takes its maximal allowed value. 

Similarly, we set
\begin{equation} 
U_- \to Y
\end{equation}
sending
\begin{equation} 
\begin{split} 
p_- &= -\br{p}\\ 
q_- &= \br{q} \\
z_- &= \abs{p}^{1/\kappa} z 
\end{split} 
\end{equation} 
This identifies $U_-$ with an open in $Y$ near the boundary where $\abs{z}$ takes its minimal allowed value.

Then, we use the maps
\begin{equation}
\begin{split}
V_- \hookleftarrow &U_- \hookrightarrow Y \\
V_+ \hookleftarrow &U_+ \hookrightarrow Y
\end{split}
\end{equation}
to glue $V_-$ and $V_+$ to $Y$.  

We see that, as desired, we have built a geometry where for large $\abs{p}$, the tube connecting chiral and anti-chiral patches becomes  very long and decouple into chiral and anti-chiral patches. As  $\abs{p}$ gets smaller, this tube shrinks until we reach a point where the geometry can no longer be continued.  

\section{Future directions}
There are many questions regarding the proposed holographic duality between topological strings on generalized CY manifolds and integrable field theories that we have not addressed in this work. 

The most important is to understand the non-perturbative behaviour of the duality.  Many important aspects of the PCM are understood at large $N$.  For instance, the spectrum and $S$-matrix of the principal chiral model in the planar limit is known (see \cite{Kazakov:2023imu} for recent work). Form factors and operators products are known \cite{Orland:2014qza, Orland:2024uvq}.  We currently do not see any of these aspects of the PCM our holographic dual.  

It would  also be very interesting to connect our appoach with other aspects of the study of the large $N$ PCM in  \cite{Kazakov:2023imu}; in particular, one can hope that the extra continuous direction that found in that paper can be identified with the coordinate $\abs{p}$ in our holographic geometry.  

At the perturbative level, there are many open questions, of which the most basic is to show how to compute quantities other than the $\beta$-function.  To be a fully-fledged holographic duality, we would need to be able to compute things like correlation functions of local operators.  This is difficult, because many operators which are local in the integrable field theory are non-local in $4d$ Chern-Simons; they involve defects which wrap the spectral curve.   Presumably these correspond to branes in the dual geometry.

\section{Acknowledgements}
We would like to thank Davide Gaiotto for collaborating at an early stage of this project and for helpful conversations. KC would also like to thank Sylvain Lacroix and Anders Wallberg for discussions on the RG flow of integrable models.  JL would like to thank Atul Sharma and Tim Adamo for helpful discussions. 

This research was
supported in part by a grant from the Krembil Foundation. KC is supported by
the NSERC Discovery Grant program and by the Perimeter Institute for Theoretical
Physics. Research at Perimeter Institute is supported in part by the Government of
Canada through the Department of Innovation, Science and Economic Development
and by the Province of Ontario through the Ministry of Colleges and Universities. JL is supported by the Royal Society. 

\appendix

\section{Conventions}
\label{app:GxG}

In this appendix we write the couplings of \cite{Levine:2021fof} in terms of the zeroes and poles of our one-form $\phi$, for the PCM, the PCM with WZ term, and the $G\times G$ model of Section~\ref{sec:GxG}. We do so by carefully matching the conventions of \cite{Delduc:2020vxy} and \cite{Levine:2021fof}. We start from the PCM, for which the reduction of four-dimensional Chern--Simons theory to two dimensions was carried out explicitly in \cite{Lacroix:2025ias}. For $\phi=(a^2-z^2)/z^2$, the change of variables $z=aw$ turns our action into that of \cite[eqs.~(2.9), (D.1)]{Lacroix:2025ias}, explicitly
\begin{equation}
-\frac{1}{8\pi^2 i\hbar}\int\phi\mathrm{d}z\wedge\mathrm{CS}(A)=\frac{1}{16\pi^2 i}\int h\frac{1-w^2}{w^2}\mathrm{d}w\wedge\mathrm{CS}(A),\qquad h=-2\hbar^{-1}a.
\end{equation}
Its reduction to two dimensions is $S=\frac{h}{8\pi}\int d^2x\langle J_+,J_-\rangle$, with $J=g^{-1}\mathrm{d}g$ \cite[eq.~(D.10)]{Lacroix:2025ias}. We now compare this with the PCM of \cite[eq.~(1.1)]{Levine:2021fof},
\begin{equation}
S_{\rm LT}=\frac{1}{4\pi\alpha'}\int d^2\xi\Big(-\frac12 s\operatorname{Tr}_{\rm LT}[J_+J_-]\Big).
\end{equation}
The light-cone conventions of the two papers agree, and their bilinear forms are related by $\langle X,Y\rangle=-\operatorname{Tr}_{\rm LT}[XY]$, since both are fixed by the Killing form and the dual Coxeter number. The two actions therefore coincide for $s/\alpha'=h$. Identifying $\alpha'$ with $\hbar$, this gives $s=-2a$: our one-form produces the PCM of \cite{Levine:2021fof} with this coupling.

For the same one-form, with a single double pole at $z_1=0$ and zeroes $\zeta^\pm=\pm a$, eq.~(3.2) of \cite{Delduc:2020vxy} gives $\rho_{11}=a/2$ and $\Bbbk_1=0$, so their action is $S_{\rm DLSS}=-\rho_{11}\int\mathrm{d}t\mathrm{d}x\operatorname{Tr}_{\rm DLSS}(j_+j_-)$, with the same fields, currents and light-cone derivatives as in \cite{Levine:2021fof}. Requiring that it reproduce $S_{\rm LT}$ with $s=-2a$ fixes
\begin{equation}
\operatorname{Tr}_{\rm DLSS}=-\frac{1}{2\pi\alpha'}\operatorname{Tr}_{\rm LT}.
\end{equation}
Having fixed the normalization of the trace, we now turn to the WZ term, which \cite[eq.~(1.1)]{Levine:2021fof} add to the PCM as $\frac{1}{4\pi\alpha'}\int d^2\xi k\mathcal L_{\rm WZ}(g)$, and \cite{Delduc:2020vxy} as $\Bbbk I_{\rm WZ}[g]$. To fix the normalization of this term, we consider the PCM with a WZ term, with one-form $\phi=(a^2-z^2)/(z-b)^2$. For it, eq.~(3.2) of \cite{Delduc:2020vxy} gives $\rho_{11}=a/2$ and $\Bbbk_1=b$, and with the trace normalized as above, the trace terms give
\begin{equation}
-\rho_{11}\int\mathrm{d}t\mathrm{d}x\operatorname{Tr}_{\rm DLSS}(j_+j_-)=\frac{1}{4\pi\alpha'}\int d^2\xi\Big(-\frac12(-4\rho_{11})\operatorname{Tr}_{\rm LT}[J_+J_-]\Big),
\end{equation}
so $s=-4\rho_{11}$. The WZ terms cannot be compared directly, because \cite{Levine:2021fof} do not write $\mathcal L_{\rm WZ}$ explicitly. Both papers do, however, give the equations of motion \cite[eq.~(A.3)]{Levine:2021fof}, \cite[eq.~(2.6)]{Delduc:2020vxy}, in a form where the commutator terms cancel:
\begin{equation}
(s-k)\partial_+J_-+(s+k)\partial_-J_+=0,\qquad \big(\rho_{11}+\tfrac12\Bbbk_1\big)\partial_+J_-+\big(\rho_{11}-\tfrac12\Bbbk_1\big)\partial_-J_+=0.
\end{equation}
Since the two actions describe the same model, the two equations are proportional, $s\mp k=c\big(\rho_{11}\pm\tfrac12\Bbbk_1\big)$. Adding and subtracting gives $c=s/\rho_{11}=-4$ and $k=-\tfrac{c}{2}\Bbbk_1=2\Bbbk_1$, so the WZ terms $\Bbbk_1I_{\rm WZ}[g]$ and $\frac{1}{4\pi\alpha'}\int d^2\xi k\mathcal L_{\rm WZ}(g)$ agree provided
\begin{equation}
I_{\rm WZ}[g]=\frac{1}{2\pi\alpha'}\int d^2\xi\mathcal L_{\rm WZ}(g).
\end{equation}
Having fixed the normalizations of the trace and of the WZ term, we can now turn to the general $G\times G$ model. Since neither relation depends on the one-form, the action of \cite{Delduc:2020vxy} can be written as
\begin{equation}
S_{\rm DLSS}=\frac{1}{2\pi\alpha'}\int d^2\xi\Big(\sum_{i,j}\rho_{ij}\operatorname{Tr}_{\rm LT}\big[J_+^{(i)}J_-^{(j)}\big]+\sum_i\Bbbk_i\mathcal L_{\rm WZ}\big(g^{(i)}\big)\Big).
\end{equation}
Comparing with the action of \cite[eq.~(2.2)]{Levine:2021fof},
\begin{equation}
S_{\rm LT}=\frac{1}{4\pi\alpha'}\int d^2\xi\bigg(-\frac12\begin{pmatrix}s&t+b\\ t-b&u\end{pmatrix}_{\!ij}\operatorname{Tr}_{\rm LT}\big[J_+^{(i)}J_-^{(j)}\big]+k_i\mathcal L_{\rm WZ}\big(g^{(i)}\big)\bigg),
\end{equation}
gives
\begin{equation}
s=-4\rho_{11},\quad u=-4\rho_{22},\quad t=-2(\rho_{12}+\rho_{21}),\quad b=-2(\rho_{12}-\rho_{21}),\quad k_i=2\Bbbk_i.
\end{equation}
Inserting the expressions for $\rho_{ij}$ and $\Bbbk_i$ given in eq.~(3.2) of \cite{Delduc:2020vxy} for the one-form of Section~\ref{sec:GxG} gives \eqref{ec:sutbintermsofmoduli}. Since $\Bbbk_i=-\frac12\operatorname{Res}_{z=z_i}\phi$, the levels are $k_i=-\operatorname{Res}_{z=z_i}\phi$. In the main text we write $\operatorname{Tr}$ for $\operatorname{Tr}_{\rm LT}$.

\section{The polynomials $F^{(4)}$ and $F^{(9)}$}
\label{app:Fpolynomials}

The polynomials $F^{(4)}_{ij}$ and $F^{(9)}_{ij}$ in the two-loop RG flow \eqref{ec:rgflowGxG} of the $G\times G$ model were computed in \cite{Levine:2021fof}. For the reader's convenience, we list them here. At one loop,
\begin{align}
F^{(4)}_{11}={}&s^2t^2 +s^2u^2 +2st^3 -2st^2u +2t^4 -s^2b^2 -4stb^2 -2sub^2 -2t^2b^2 +2t^2bk_1 \nonumber\\
&+4tubk_1 -u^2k_1^2.\label{ec:F411}
\end{align}
\begin{align}
F^{(4)}_{12}={}&s^2tu +3st^2u +stu^2 -t^4 +s^2ub +2st^2b +4stub +su^2b +2t^3b +2t^2ub \nonumber\\
&+2stb^2 +2sub^2 +3t^2b^2 +2tub^2 +s^2tk_2 +st^2k_2 +s^2bk_2 +2stbk_2 -t^2uk_1 \nonumber\\
&-tu^2k_1 -2tubk_1 -u^2bk_1 -t^2k_1k_2.\label{ec:F412}
\end{align}
\begin{align}
F^{(4)}_{21}={}&s^2tu +3st^2u +stu^2 -t^4 -s^2ub -2st^2b -4stub -su^2b -2t^3b -2t^2ub \nonumber\\
&+2stb^2 +2sub^2 +3t^2b^2 +2tub^2 -s^2tk_2 -st^2k_2 +s^2bk_2 +2stbk_2 +t^2uk_1 \nonumber\\
&+tu^2k_1 -2tubk_1 -u^2bk_1 -t^2k_1k_2.\label{ec:F421}
\end{align}
\begin{align}
F^{(4)}_{22}={}&s^2u^2 -2st^2u +2t^4 +2t^3u +t^2u^2 -2sub^2 -2t^2b^2 -4tub^2 -u^2b^2 -4stbk_2 \nonumber\\
&-2t^2bk_2 -s^2k_2^2.\label{ec:F422}
\end{align}
At two loops,
{\allowdisplaybreaks
\begin{align}
F^{(9)}_{11}={}&\tfrac12\big(s^5t^2u^2 +s^4t^4u +12s^4t^3u^2 +6s^4t^2u^3 +s^4u^5 -2s^3t^6 -12s^3t^5u +13s^3t^4u^2 +12s^3t^3u^3 \nonumber\\
&-5s^3t^2u^4 -26s^2t^6u -6s^2t^5u^2 +16s^2t^4u^3 +7st^8 -12st^7u -18st^6u^2 +6t^9 +6t^8u -s^5u^2b^2 \nonumber\\
&-6s^4t^2ub^2 -6s^4tu^2b^2 -4s^4u^3b^2 -20s^3t^4b^2 -24s^3t^3ub^2 +4s^3t^2u^2b^2 -16s^3tu^3b^2 \nonumber\\
&-7s^3u^4b^2 -78s^2t^5b^2 -64s^2t^4ub^2 +44s^2t^3u^2b^2 +12s^2t^2u^3b^2 -98st^6b^2 -100st^5ub^2 \nonumber\\
&+4st^4u^2b^2 -36t^7b^2 -36t^6ub^2 +5s^4ub^4 +22s^3t^2b^4 +64s^3tub^4 +22s^3u^2b^4 +98s^2t^3b^4 \nonumber\\
&+214s^2t^2ub^4 +94s^2tu^2b^4 +8s^2u^3b^4 +115st^4b^4 +200st^3ub^4 +70st^2u^2b^4 +30t^5b^4 \nonumber\\
&+30t^4ub^4 -18s^4t^3bk_2 +4s^4t^2ubk_2 +4s^4tu^2bk_2 -58s^3t^4bk_2 +4s^3t^2u^2bk_2 -58s^2t^5bk_2 \nonumber\\
&-4s^2t^4ubk_2 -18st^6bk_2 +18s^4tb^3k_2 +8s^4ub^3k_2 +82s^3t^2b^3k_2 +44s^3tub^3k_2 +2s^3u^2b^3k_2 \nonumber\\
&+100s^2t^3b^3k_2 +40s^2t^2ub^3k_2 +30st^4b^3k_2 -4s^3t^2u^2bk_1 -4s^3tu^3bk_1 +12s^2t^4ubk_1 \nonumber\\
&-4s^2t^3u^2bk_1 -4s^2t^2u^3bk_1 +14s^2tu^4bk_1 +10st^6bk_1 +44st^5ubk_1 -10st^4u^2bk_1 \nonumber\\
&-48st^3u^3bk_1 +18t^7bk_1 +68t^6ubk_1 +52t^5u^2bk_1 -2s^3u^2b^3k_1 -40s^2t^2ub^3k_1 \nonumber\\
&-68s^2tu^2b^3k_1 -20s^2u^3b^3k_1 -30st^4b^3k_1 -196st^3ub^3k_1 -238st^2u^2b^3k_1 -72stu^3b^3k_1 \nonumber\\
&-42t^5b^3k_1 -156t^4ub^3k_1 -108t^3u^2b^3k_1 -3s^5t^2k_2^2 -6s^4t^3k_2^2 -3s^3t^4k_2^2 +3s^5b^2k_2^2 \nonumber\\
&+12s^4tb^2k_2^2 +2s^4ub^2k_2^2 +10s^3t^2b^2k_2^2 -4s^3t^3uk_1k_2 -4s^3t^2u^2k_1k_2 +10s^2t^5k_1k_2 \nonumber\\
&-4s^2t^3u^2k_1k_2 +16st^6k_1k_2 +4st^5uk_1k_2 +6t^7k_1k_2 -8s^3tub^2k_1k_2 -4s^3u^2b^2k_1k_2 \nonumber\\
&-34s^2t^3b^2k_1k_2 -64s^2t^2ub^2k_1k_2 -14s^2tu^2b^2k_1k_2 -64st^4b^2k_1k_2 -64st^3ub^2k_1k_2 \nonumber\\
&-18t^5b^2k_1k_2 -4s^2u^5k_1^2 -st^4u^2k_1^2 +4st^3u^3k_1^2 +13st^2u^4k_1^2 -2t^6uk_1^2 -10t^5u^2k_1^2 \nonumber\\
&-12t^4u^3k_1^2 +2s^2u^3b^2k_1^2 +34st^2u^2b^2k_1^2 +48stu^3b^2k_1^2 +15su^4b^2k_1^2 +30t^4ub^2k_1^2 \nonumber\\
&+126t^3u^2b^2k_1^2 +96t^2u^3b^2k_1^2 -6s^3t^2bk_1k_2^2 -4s^3tubk_1k_2^2 -8s^2t^3bk_1k_2^2 +4s^2tu^2bk_1^2k_2 \nonumber\\
&+20st^3ubk_1^2k_2 +18st^2u^2bk_1^2k_2 +6t^5bk_1^2k_2 +24t^4ubk_1^2k_2 -24t^2u^3bk_1^3 -30tu^4bk_1^3 \nonumber\\
&+2s^2t^2uk_1^2k_2^2 +st^4k_1^2k_2^2 -6t^3u^2k_1^3k_2 +3u^5k_1^4\big).\label{ec:F911}
\end{align}}
{\allowdisplaybreaks
\begin{align}
F^{(9)}_{12}={}&\tfrac12\big(s^5tu^3 -5s^4t^3u^2 +10s^3t^5u -3s^3t^4u^2 -4s^3t^3u^3 +s^3tu^5 -6s^2t^7 +18s^2t^6u +16s^2t^5u^2 \nonumber\\
&-3s^2t^4u^3 -5s^2t^3u^4 -15st^8 -2st^7u +18st^6u^2 +10st^5u^3 -10t^9 -15t^8u -6t^7u^2 +s^5u^3b \nonumber\\
&-3s^4t^2u^2b +2s^4tu^3b -12s^3t^3u^2b +2s^3tu^4b +s^3u^5b +2s^2t^6b +6s^2t^5ub -18s^2t^4u^2b \nonumber\\
&-12s^2t^3u^3b -3s^2t^2u^4b +4st^7b +18st^6ub +6st^5u^2b +4t^7ub +2t^6u^2b -5s^4tu^2b^2 \nonumber\\
&+s^4u^3b^2 -40s^3t^3ub^2 -86s^3t^2u^2b^2 -24s^3tu^3b^2 +s^3u^4b^2 +18s^2t^5b^2 -77s^2t^4ub^2 \nonumber\\
&-204s^2t^3u^2b^2 -86s^2t^2u^3b^2 -5s^2tu^4b^2 +54st^6b^2 +34st^5ub^2 -77st^4u^2b^2 -40st^3u^3b^2 \nonumber\\
&+32t^7b^2 +54t^6ub^2 +18t^5u^2b^2 -7s^4u^2b^3 -52s^3t^2ub^3 -80s^3tu^2b^3 -20s^3u^3b^3 -22s^2t^4b^3 \nonumber\\
&-200s^2t^3ub^3 -240s^2t^2u^2b^3 -80s^2tu^3b^3 -7s^2u^4b^3 -44st^5b^3 -210st^4ub^3 -200st^3u^2b^3 \nonumber\\
&-52st^2u^3b^3 -16t^6b^3 -44t^5ub^3 -22t^4u^2b^3 -22s^3tub^4 -15s^3u^2b^4 -32s^2t^3b^4 \nonumber\\
&-149s^2t^2ub^4 -104s^2tu^2b^4 -15s^2u^3b^4 -79st^4b^4 -236st^3ub^4 -149st^2u^2b^4 -22stu^3b^4 \nonumber\\
&-38t^5b^4 -79t^4ub^4 -32t^3u^2b^4 +s^5tu^2k_2 -3s^4t^3uk_2 +2s^4t^2u^2k_2 +2s^3t^5k_2 -7s^3t^4uk_2 \nonumber\\
&+2s^3t^3u^2k_2 +5s^2t^6k_2 -7s^2t^5uk_2 +s^2t^4u^2k_2 +5st^7k_2 -3st^6uk_2 +2t^8k_2 +s^5u^2bk_2 \nonumber\\
&-25s^4t^2ubk_2 -8s^4tu^2bk_2 +6s^3t^4bk_2 -72s^3t^3ubk_2 -26s^3t^2u^2bk_2 +26s^2t^5bk_2 \nonumber\\
&-55s^2t^4ubk_2 -16s^2t^3u^2bk_2 +27st^6bk_2 -10st^5ubk_2 +8t^7bk_2 -31s^4tub^2k_2 -5s^4u^2b^2k_2 \nonumber\\
&-50s^3t^3b^2k_2 -121s^3t^2ub^2k_2 -24s^3tu^2b^2k_2 +s^3u^3b^2k_2 -117s^2t^4b^2k_2 -152s^2t^3ub^2k_2 \nonumber\\
&-27s^2t^2u^2b^2k_2 -67st^5b^2k_2 -49st^4ub^2k_2 -6t^6b^2k_2 -9s^4ub^3k_2 -54s^3t^2b^3k_2 \nonumber\\
&-82s^3tub^3k_2 -12s^3u^2b^3k_2 -152s^2t^3b^3k_2 -148s^2t^2ub^3k_2 -22s^2tu^2b^3k_2 -101st^4b^3k_2 \nonumber\\
&-56st^3ub^3k_2 -12t^5b^3k_2 -s^2t^4u^2k_1 -2s^2t^3u^3k_1 -2s^2t^2u^4k_1 -s^2tu^5k_1 +3st^6uk_1 \nonumber\\
&+7st^5u^2k_1 +7st^4u^3k_1 +3st^3u^4k_1 -2t^8k_1 -5t^7uk_1 -5t^6u^2k_1 -2t^5u^3k_1 +16s^2t^3u^2bk_1 \nonumber\\
&+26s^2t^2u^3bk_1 +8s^2tu^4bk_1 -s^2u^5bk_1 +10st^5ubk_1 +55st^4u^2bk_1 +72st^3u^3bk_1 \nonumber\\
&+25st^2u^4bk_1 -8t^7bk_1 -27t^6ubk_1 -26t^5u^2bk_1 -6t^4u^3bk_1 -s^3u^3b^2k_1 +27s^2t^2u^2b^2k_1 \nonumber\\
&+24s^2tu^3b^2k_1 +5s^2u^4b^2k_1 +49st^4ub^2k_1 +152st^3u^2b^2k_1 +121st^2u^3b^2k_1 +31stu^4b^2k_1 \nonumber\\
&+6t^6b^2k_1 +67t^5ub^2k_1 +117t^4u^2b^2k_1 +50t^3u^3b^2k_1 +22s^2tu^2b^3k_1 +12s^2u^3b^3k_1 \nonumber\\
&+56st^3ub^3k_1 +148st^2u^2b^3k_1 +82stu^3b^3k_1 +9su^4b^3k_1 +12t^5b^3k_1 +101t^4ub^3k_1 \nonumber\\
&+152t^3u^2b^3k_1 +54t^2u^3b^3k_1 -3s^5tuk_2^2 -9s^4t^2uk_2^2 +3s^3t^4k_2^2 -6s^3t^3uk_2^2 +3s^2t^5k_2^2 \nonumber\\
&-3s^5ubk_2^2 -24s^4t^2bk_2^2 -8s^4tubk_2^2 -46s^3t^3bk_2^2 -8s^3t^2ubk_2^2 -19s^2t^4bk_2^2 \nonumber\\
&-24s^4tb^2k_2^2 -6s^4ub^2k_2^2 -63s^3t^2b^2k_2^2 -10s^3tub^2k_2^2 -32s^2t^3b^2k_2^2 +7s^2t^4uk_1k_2 \nonumber\\
&+8s^2t^3u^2k_1k_2 -st^6k_1k_2 +4st^5uk_1k_2 +7st^4u^2k_1k_2 -t^6uk_1k_2 -2s^3tu^2bk_1k_2 \nonumber\\
&+30s^2t^3ubk_1k_2 +12s^2t^2u^2bk_1k_2 -2s^2tu^3bk_1k_2 +26st^5bk_1k_2 +76st^4ubk_1k_2 \nonumber\\
&+30st^3u^2bk_1k_2 +20t^6bk_1k_2 +26t^5ubk_1k_2 +3s^3u^2b^2k_1k_2 +49s^2t^2ub^2k_1k_2 \nonumber\\
&+32s^2tu^2b^2k_1k_2 +3s^2u^3b^2k_1k_2 +41st^4b^2k_1k_2 +128st^3ub^2k_1k_2 +49st^2u^2b^2k_1k_2 \nonumber\\
&+32t^5b^2k_1k_2 +41t^4ub^2k_1k_2 -6st^3u^3k_1^2 -9st^2u^4k_1^2 -3stu^5k_1^2 +3t^5u^2k_1^2 +3t^4u^3k_1^2 \nonumber\\
&-8st^2u^3bk_1^2 -8stu^4bk_1^2 -3su^5bk_1^2 -19t^4u^2bk_1^2 -46t^3u^3bk_1^2 -24t^2u^4bk_1^2 \nonumber\\
&-10stu^3b^2k_1^2 -6su^4b^2k_1^2 -32t^3u^2b^2k_1^2 -63t^2u^3b^2k_1^2 -24tu^4b^2k_1^2 -3s^5tk_2^3 -3s^4t^2k_2^3 \nonumber\\
&-3s^5bk_2^3 -6s^4tbk_2^3 -s^3t^2uk_1k_2^2 +10s^2t^4k_1k_2^2 +2s^2t^3uk_1k_2^2 +7st^5k_1k_2^2 +6s^3tubk_1k_2^2 \nonumber\\
&+24s^2t^3bk_1k_2^2 +14s^2t^2ubk_1k_2^2 +19st^4bk_1k_2^2 -2st^3u^2k_1^2k_2 +st^2u^3k_1^2k_2 -7t^5uk_1^2k_2 \nonumber\\
&-10t^4u^2k_1^2k_2 -14st^2u^2bk_1^2k_2 -6stu^3bk_1^2k_2 -19t^4ubk_1^2k_2 -24t^3u^2bk_1^2k_2 +3t^2u^4k_1^3 \nonumber\\
&+3tu^5k_1^3 +6tu^4bk_1^3 +3u^5bk_1^3 +3s^3t^2k_1k_2^3 -4st^3uk_1^2k_2^2 -2t^5k_1^2k_2^2 +3t^2u^3k_1^3k_2\big).\label{ec:F912}
\end{align}}
{\allowdisplaybreaks
\begin{align}
F^{(9)}_{21}={}&\tfrac12\big(s^5tu^3 -5s^4t^3u^2 +10s^3t^5u -3s^3t^4u^2 -4s^3t^3u^3 +s^3tu^5 -6s^2t^7 +18s^2t^6u +16s^2t^5u^2 \nonumber\\
&-3s^2t^4u^3 -5s^2t^3u^4 -15st^8 -2st^7u +18st^6u^2 +10st^5u^3 -10t^9 -15t^8u -6t^7u^2 -s^5u^3b \nonumber\\
&+3s^4t^2u^2b -2s^4tu^3b +12s^3t^3u^2b -2s^3tu^4b -s^3u^5b -2s^2t^6b -6s^2t^5ub +18s^2t^4u^2b \nonumber\\
&+12s^2t^3u^3b +3s^2t^2u^4b -4st^7b -18st^6ub -6st^5u^2b -4t^7ub -2t^6u^2b -5s^4tu^2b^2 \nonumber\\
&+s^4u^3b^2 -40s^3t^3ub^2 -86s^3t^2u^2b^2 -24s^3tu^3b^2 +s^3u^4b^2 +18s^2t^5b^2 -77s^2t^4ub^2 \nonumber\\
&-204s^2t^3u^2b^2 -86s^2t^2u^3b^2 -5s^2tu^4b^2 +54st^6b^2 +34st^5ub^2 -77st^4u^2b^2 -40st^3u^3b^2 \nonumber\\
&+32t^7b^2 +54t^6ub^2 +18t^5u^2b^2 +7s^4u^2b^3 +52s^3t^2ub^3 +80s^3tu^2b^3 +20s^3u^3b^3 +22s^2t^4b^3 \nonumber\\
&+200s^2t^3ub^3 +240s^2t^2u^2b^3 +80s^2tu^3b^3 +7s^2u^4b^3 +44st^5b^3 +210st^4ub^3 +200st^3u^2b^3 \nonumber\\
&+52st^2u^3b^3 +16t^6b^3 +44t^5ub^3 +22t^4u^2b^3 -22s^3tub^4 -15s^3u^2b^4 -32s^2t^3b^4 \nonumber\\
&-149s^2t^2ub^4 -104s^2tu^2b^4 -15s^2u^3b^4 -79st^4b^4 -236st^3ub^4 -149st^2u^2b^4 -22stu^3b^4 \nonumber\\
&-38t^5b^4 -79t^4ub^4 -32t^3u^2b^4 -s^5tu^2k_2 +3s^4t^3uk_2 -2s^4t^2u^2k_2 -2s^3t^5k_2 +7s^3t^4uk_2 \nonumber\\
&-2s^3t^3u^2k_2 -5s^2t^6k_2 +7s^2t^5uk_2 -s^2t^4u^2k_2 -5st^7k_2 +3st^6uk_2 -2t^8k_2 +s^5u^2bk_2 \nonumber\\
&-25s^4t^2ubk_2 -8s^4tu^2bk_2 +6s^3t^4bk_2 -72s^3t^3ubk_2 -26s^3t^2u^2bk_2 +26s^2t^5bk_2 \nonumber\\
&-55s^2t^4ubk_2 -16s^2t^3u^2bk_2 +27st^6bk_2 -10st^5ubk_2 +8t^7bk_2 +31s^4tub^2k_2 +5s^4u^2b^2k_2 \nonumber\\
&+50s^3t^3b^2k_2 +121s^3t^2ub^2k_2 +24s^3tu^2b^2k_2 -s^3u^3b^2k_2 +117s^2t^4b^2k_2 +152s^2t^3ub^2k_2 \nonumber\\
&+27s^2t^2u^2b^2k_2 +67st^5b^2k_2 +49st^4ub^2k_2 +6t^6b^2k_2 -9s^4ub^3k_2 -54s^3t^2b^3k_2 \nonumber\\
&-82s^3tub^3k_2 -12s^3u^2b^3k_2 -152s^2t^3b^3k_2 -148s^2t^2ub^3k_2 -22s^2tu^2b^3k_2 -101st^4b^3k_2 \nonumber\\
&-56st^3ub^3k_2 -12t^5b^3k_2 +s^2t^4u^2k_1 +2s^2t^3u^3k_1 +2s^2t^2u^4k_1 +s^2tu^5k_1 -3st^6uk_1 \nonumber\\
&-7st^5u^2k_1 -7st^4u^3k_1 -3st^3u^4k_1 +2t^8k_1 +5t^7uk_1 +5t^6u^2k_1 +2t^5u^3k_1 +16s^2t^3u^2bk_1 \nonumber\\
&+26s^2t^2u^3bk_1 +8s^2tu^4bk_1 -s^2u^5bk_1 +10st^5ubk_1 +55st^4u^2bk_1 +72st^3u^3bk_1 \nonumber\\
&+25st^2u^4bk_1 -8t^7bk_1 -27t^6ubk_1 -26t^5u^2bk_1 -6t^4u^3bk_1 +s^3u^3b^2k_1 -27s^2t^2u^2b^2k_1 \nonumber\\
&-24s^2tu^3b^2k_1 -5s^2u^4b^2k_1 -49st^4ub^2k_1 -152st^3u^2b^2k_1 -121st^2u^3b^2k_1 -31stu^4b^2k_1 \nonumber\\
&-6t^6b^2k_1 -67t^5ub^2k_1 -117t^4u^2b^2k_1 -50t^3u^3b^2k_1 +22s^2tu^2b^3k_1 +12s^2u^3b^3k_1 \nonumber\\
&+56st^3ub^3k_1 +148st^2u^2b^3k_1 +82stu^3b^3k_1 +9su^4b^3k_1 +12t^5b^3k_1 +101t^4ub^3k_1 \nonumber\\
&+152t^3u^2b^3k_1 +54t^2u^3b^3k_1 -3s^5tuk_2^2 -9s^4t^2uk_2^2 +3s^3t^4k_2^2 -6s^3t^3uk_2^2 +3s^2t^5k_2^2 \nonumber\\
&+3s^5ubk_2^2 +24s^4t^2bk_2^2 +8s^4tubk_2^2 +46s^3t^3bk_2^2 +8s^3t^2ubk_2^2 +19s^2t^4bk_2^2 \nonumber\\
&-24s^4tb^2k_2^2 -6s^4ub^2k_2^2 -63s^3t^2b^2k_2^2 -10s^3tub^2k_2^2 -32s^2t^3b^2k_2^2 +7s^2t^4uk_1k_2 \nonumber\\
&+8s^2t^3u^2k_1k_2 -st^6k_1k_2 +4st^5uk_1k_2 +7st^4u^2k_1k_2 -t^6uk_1k_2 +2s^3tu^2bk_1k_2 \nonumber\\
&-30s^2t^3ubk_1k_2 -12s^2t^2u^2bk_1k_2 +2s^2tu^3bk_1k_2 -26st^5bk_1k_2 -76st^4ubk_1k_2 \nonumber\\
&-30st^3u^2bk_1k_2 -20t^6bk_1k_2 -26t^5ubk_1k_2 +3s^3u^2b^2k_1k_2 +49s^2t^2ub^2k_1k_2 \nonumber\\
&+32s^2tu^2b^2k_1k_2 +3s^2u^3b^2k_1k_2 +41st^4b^2k_1k_2 +128st^3ub^2k_1k_2 +49st^2u^2b^2k_1k_2 \nonumber\\
&+32t^5b^2k_1k_2 +41t^4ub^2k_1k_2 -6st^3u^3k_1^2 -9st^2u^4k_1^2 -3stu^5k_1^2 +3t^5u^2k_1^2 +3t^4u^3k_1^2 \nonumber\\
&+8st^2u^3bk_1^2 +8stu^4bk_1^2 +3su^5bk_1^2 +19t^4u^2bk_1^2 +46t^3u^3bk_1^2 +24t^2u^4bk_1^2 \nonumber\\
&-10stu^3b^2k_1^2 -6su^4b^2k_1^2 -32t^3u^2b^2k_1^2 -63t^2u^3b^2k_1^2 -24tu^4b^2k_1^2 +3s^5tk_2^3 +3s^4t^2k_2^3 \nonumber\\
&-3s^5bk_2^3 -6s^4tbk_2^3 +s^3t^2uk_1k_2^2 -10s^2t^4k_1k_2^2 -2s^2t^3uk_1k_2^2 -7st^5k_1k_2^2 +6s^3tubk_1k_2^2 \nonumber\\
&+24s^2t^3bk_1k_2^2 +14s^2t^2ubk_1k_2^2 +19st^4bk_1k_2^2 +2st^3u^2k_1^2k_2 -st^2u^3k_1^2k_2 +7t^5uk_1^2k_2 \nonumber\\
&+10t^4u^2k_1^2k_2 -14st^2u^2bk_1^2k_2 -6stu^3bk_1^2k_2 -19t^4ubk_1^2k_2 -24t^3u^2bk_1^2k_2 -3t^2u^4k_1^3 \nonumber\\
&-3tu^5k_1^3 +6tu^4bk_1^3 +3u^5bk_1^3 +3s^3t^2k_1k_2^3 -4st^3uk_1^2k_2^2 -2t^5k_1^2k_2^2 +3t^2u^3k_1^3k_2\big).\label{ec:F921}
\end{align}}
{\allowdisplaybreaks
\begin{align}
F^{(9)}_{22}={}&\tfrac12\big(s^5u^4 -5s^4t^2u^3 +16s^3t^4u^2 +12s^3t^3u^3 +6s^3t^2u^4 -18s^2t^6u -6s^2t^5u^2 +13s^2t^4u^3 +12s^2t^3u^4 \nonumber\\
&+s^2t^2u^5 +6st^8 -12st^7u -26st^6u^2 -12st^5u^3 +st^4u^4 +6t^9 +7t^8u -2t^6u^3 -7s^4u^3b^2 \nonumber\\
&+12s^3t^2u^2b^2 -16s^3tu^3b^2 -4s^3u^4b^2 +4s^2t^4ub^2 +44s^2t^3u^2b^2 +4s^2t^2u^3b^2 -6s^2tu^4b^2 \nonumber\\
&-s^2u^5b^2 -36st^6b^2 -100st^5ub^2 -64st^4u^2b^2 -24st^3u^3b^2 -6st^2u^4b^2 -36t^7b^2 -98t^6ub^2 \nonumber\\
&-78t^5u^2b^2 -20t^4u^3b^2 +8s^3u^2b^4 +70s^2t^2ub^4 +94s^2tu^2b^4 +22s^2u^3b^4 +30st^4b^4 \nonumber\\
&+200st^3ub^4 +214st^2u^2b^4 +64stu^3b^4 +5su^4b^4 +30t^5b^4 +115t^4ub^4 +98t^3u^2b^4 \nonumber\\
&+22t^2u^3b^4 -14s^4tu^2bk_2 +48s^3t^3ubk_2 +4s^3t^2u^2bk_2 +4s^3tu^3bk_2 -52s^2t^5bk_2 \nonumber\\
&+10s^2t^4ubk_2 +4s^2t^3u^2bk_2 +4s^2t^2u^3bk_2 -68st^6bk_2 -44st^5ubk_2 -12st^4u^2bk_2 \nonumber\\
&-18t^7bk_2 -10t^6ubk_2 +72s^3tub^3k_2 +20s^3u^2b^3k_2 +108s^2t^3b^3k_2 +238s^2t^2ub^3k_2 \nonumber\\
&+68s^2tu^2b^3k_2 +2s^2u^3b^3k_2 +156st^4b^3k_2 +196st^3ub^3k_2 +40st^2u^2b^3k_2 +42t^5b^3k_2 \nonumber\\
&+30t^4ub^3k_2 -4s^2t^2u^3bk_1 -4s^2tu^4bk_1 +4st^4u^2bk_1 -4st^2u^4bk_1 +18t^6ubk_1 +58t^5u^2bk_1 \nonumber\\
&+58t^4u^3bk_1 +18t^3u^4bk_1 -2s^2u^3b^3k_1 -40st^2u^2b^3k_1 -44stu^3b^3k_1 -8su^4b^3k_1 -30t^4ub^3k_1 \nonumber\\
&-100t^3u^2b^3k_1 -82t^2u^3b^3k_1 -18tu^4b^3k_1 -4s^5u^2k_2^2 +13s^4t^2uk_2^2 -12s^3t^4k_2^2 +4s^3t^3uk_2^2 \nonumber\\
&-10s^2t^5k_2^2 -s^2t^4uk_2^2 -2st^6k_2^2 +15s^4ub^2k_2^2 +96s^3t^2b^2k_2^2 +48s^3tub^2k_2^2 +2s^3u^2b^2k_2^2 \nonumber\\
&+126s^2t^3b^2k_2^2 +34s^2t^2ub^2k_2^2 +30st^4b^2k_2^2 -4s^2t^3u^2k_1k_2 -4s^2t^2u^3k_1k_2 +4st^5uk_1k_2 \nonumber\\
&-4st^3u^3k_1k_2 +6t^7k_1k_2 +16t^6uk_1k_2 +10t^5u^2k_1k_2 -14s^2tu^2b^2k_1k_2 -4s^2u^3b^2k_1k_2 \nonumber\\
&-64st^3ub^2k_1k_2 -64st^2u^2b^2k_1k_2 -8stu^3b^2k_1k_2 -18t^5b^2k_1k_2 -64t^4ub^2k_1k_2 \nonumber\\
&-34t^3u^2b^2k_1k_2 -3t^4u^3k_1^2 -6t^3u^4k_1^2 -3t^2u^5k_1^2 +2su^4b^2k_1^2 +10t^2u^3b^2k_1^2 +12tu^4b^2k_1^2 \nonumber\\
&+3u^5b^2k_1^2 +30s^4tbk_2^3 +24s^3t^2bk_2^3 -18s^2t^2ubk_1k_2^2 -4s^2tu^2bk_1k_2^2 -24st^4bk_1k_2^2 \nonumber\\
&-20st^3ubk_1k_2^2 -6t^5bk_1k_2^2 +4stu^3bk_1^2k_2 +8t^3u^2bk_1^2k_2 +6t^2u^3bk_1^2k_2 +3s^5k_2^4 \nonumber\\
&-6s^2t^3k_1k_2^3 +2st^2u^2k_1^2k_2^2 +t^4uk_1^2k_2^2\big).\label{ec:F922}
\end{align}}

\bibliographystyle{JHEP}
\bibliography{generalizedcg2}
\end{document}